\documentclass[11pt]{article}

\usepackage{hyperref}
\hypersetup{colorlinks=true,citecolor=blue, linkcolor=blue, urlcolor=blue}

\usepackage{palatino}
\usepackage{mathpazo}
\usepackage{amssymb}
\usepackage{amsmath}
\usepackage{enumerate}
\usepackage{amsthm}
\newtheorem{theorem}{Theorem}
\newtheorem{lemma}{Lemma}
\newtheorem{corollary}{Corollary}
\newtheorem{example}{Example}

\usepackage{fullpage}
\usepackage{pstricks}
\usepackage{tikz}
\usetikzlibrary{arrows}
\usetikzlibrary{decorations.markings}

\newcommand{\eat}[1]{}
\newcommand{\emp}{\emptyset}
\newcommand{\termdef}[1]{\textbf{{#1}}}
\newcommand{\ep}{\varepsilon}
\newcommand{\myalph}{\mit\Sigma}
\newcommand{\numberset}{\mathbb}
\newcommand{\N}{\numberset{N}}

\newcommand{\size}[1]{| {#1} |}
\newcommand{\order}[1]{{\cal O}({#1})}

\newcommand{\gram}{{\cal G}}
\newcommand{\myrule}{P}
\newcommand{\nonterm}{V_{N}} 
\newcommand{\myterm}{V_{T}}

\newcommand{\myivar}{V_{I}}

\newcommand{\de}{\rightarrow}

\newcommand{\deriv}[1]{\Rightarrow_{#1}}
\newcommand{\derivv}[2]{\Rightarrow_{#1}^{#2}}

\newcommand{\fo}{\mathbf{fo}}
\newcommand{\ibd}{\mathbf{ib}}
\newcommand{\ebd}{\mathbf{eb}}
\newcommand{\myindex}[1]{{\cal I}(#1)}

\newcommand{\setindex}[1]{{\sf index}(#1)}

\newcommand{\boxnum}[1]{{\setlength{\fboxsep}{1pt}\raisebox{1pt}
{\hspace{1pt}\fbox{\tiny #1}\hspace{1pt}}}}
\newcommand{\ind}[1]{\ensuremath{^{\kern-0.5pt\boxnum{#1}}}}

\newcommand{\mrgbcw}{\texttt{PMCW}}
\newcommand{\mcess}{\texttt{MBW}}

\newcommand{\cw}{\mathbf{cw}}
\newcommand{\ecw}{\mathbf{ecw}}
\newcommand{\emcw}{\mathbf{emcw}}

\newcommand{\wg}{\mathbf{wd}}
\newcommand{\mwg}{\mathbf{mwd}}
\newcommand{\ewg}{\mathbf{ewd}}
\newcommand{\emwg}{\mathbf{emwd}}

\newcommand{\ggh}{H}
\newcommand{\ggw}{W}
\newcommand{\ngh}{2n^4+1}	
\newcommand{\ngw}{6n^4}
\newcommand{\lggh}{H_l}
\newcommand{\lggw}{W_l}
\newcommand{\lgh}{3n^4+1}
\newcommand{\lgw}{12n^4}
\newcommand{\mgh}{2n^4+1}

\newcommand{\mghplusone}{2n^4+2}
\newcommand{\mgw}{8n^4+1}
\newcommand{\mggh}{H_m}
\newcommand{\mggw}{W_m}
\newcommand{\rgh}{3n^4+1}
\newcommand{\rgw}{12n^4}

\newcommand{\nee}{4n^2}
\newcommand{\halfnee}{2n^2}
\newcommand{\kprime}{3n^4+2n^3+2k+2}	
\newcommand{\ekprime}{3n^4+2n^3+2k+4}
\newcommand{\nklb}{3n^8}
\newcommand{\lrklb}{17n^8}
\newcommand{\mklb}{7n^8}
\newcommand{\frlklb}{5n^8}
\newcommand{\lminusmh}{n^4}

\newcommand{\larr}{\nu'}

\newcommand{\grid}{\Gamma}
\newcommand{\cgrid}{\Sigma}

\begin{document}

\title{Synchronous Context-Free Grammars \\
and Optimal Linear Parsing Strategies}
\author{Pierluigi Crescenzi%
\thanks{Dipartimento di Sistemi e Informatica, Universit\`a di Firenze, 
Viale Morgagni, 65, 50134 Firenze, Italy.} 
\and
Daniel Gildea%
\thanks{Computer Science Department, University of Rochester, Rochester, NY 14627.} 
\and
Andrea Marino%
\thanks{Dipartimento di Informatica, Universit\`a di Milano, 
Via Festa del Perdono 7, 20122 Milano, Italy.}
\and
Gianluca Rossi%
\thanks{Dipartimento di Matematica, Universit\`a di Roma \textit{Tor Vergata}, 
Via Ricerca Scientifica 1, 00133 Roma, Italy.} 
\and
Giorgio Satta%
\thanks{Department of Information Engineering, University of Padua, 
Via Gradenigo 6/A, 35131 Padova, Italy.} 
}

\maketitle

\begin{abstract}
Synchronous Context-Free Grammars (SCFGs), also known as
syntax-directed translation schemata 
\cite{Aho:69a,AhoUll72}, 
are unlike context-free grammars in that they do not have a binary
normal form.  In general, parsing with SCFGs takes space and time
polynomial in the length of the input strings, but with 
the degree of the polynomial depending on the permutations 
of the SCFG rules.  We consider linear parsing strategies,
which add one nonterminal at a time.  We 
show that for a given input permutation, 
the problems of finding the linear parsing 
strategy with the minimum space and time complexity are both NP-hard.
\end{abstract}

\section{Introduction}
\label{sec:intro}

Synchronous Context-Free Grammars (SCFGs) are widely used
to model translational equivalence between strings in 
both the area of compilers for programming languages
and, more recently, in the area of machine 
translation of natural languages.
The formalism was first introduced by Lewis 
and Stearns \cite{LewisStearns68} under the name of 
syntax-directed transduction grammars, and was later 
called syntax-directed translation schemata by 
Aho and Ullman \cite{Aho:69a,AhoUll72}.  
The name SCFG, which we use in this article, 
was later introduced in the literature on computational 
linguistics, where the term ``synchronous'' refers to rewriting 
systems that generate strings in both a source and target 
language simultaneously \cite{Shieber:90CO,SH94,Chiang:04}. 
In fact, SCFGs can be seen as a natural extension of the well-known 
rewriting formalism of Context-Free Grammars (CFGs).  More precisely, while a CFG 
generates a set of strings, a SCFG generates a set of string pairs
using essentially the same context-free rewriting mechanism, 
along with some special synchronization between the two derivations, 
as discussed below. 

A SCFG is a string rewriting system based on synchronous rules. 
Informally, a \termdef{synchronous rule} is composed of two CFG rules along with 
a bijective pairing between all the occurrences of the nonterminal symbols 
in the right-hand side of the first rule and 
all the occurrences of the nonterminal symbols 
in the right-hand side of the second rule.  There is no restriction on 
the terminal symbols appearing in the right-hand sides of the two CFG rules.
Two nonterminal occurrences that are associated by the above bijection are
called \termdef{linked} nonterminals.  
Linked nonterminals are not necessarily occurrences of the same nonterminal symbol.  
In what follows, we will often view a synchronous rule as 
a permutation of the nonterminal occurrences in the two right-hand sides, 
combined with some renaming of these occurrences and with some 
insertion and deletion of terminal symbols.   

In a SCFG rewriting is restricted in two ways: 
the two CFG rule components in a synchronous rule must be applied 
simultaneously, and rewriting must take place at linked nonterminals.
Other than that, the application of a synchronous rule is independent 
of the context, similarly to the CFG case.  
As a result, a SCFG generates a pair of strings by means 
of two context-free parse trees that have the same skeleton 
but differ by some reordering and renaming of the nonterminal 
children at each internal node, and by the insertion 
and the deletion of the terminal children of that node.
Moreover, the projection of the generated string pairs 
on both dimensions are still context-free languages.
Thus, the added generative power of a SCFG lies in its ability to 
model long-distance movement of phrase constituents in the translation 
from the source to the target language, through simple 
permutations implemented at the internal nodes in the generated trees, 
something that is not possible with models based on finite-state transducers.  

Recently, SCFGs have received wide attention in the area of
natural language processing, where several variants of SCFGs
augmented with probabilities are currently used for translation
between natural languages.  This is due to the recent surge of
interest in commercial systems for statistical machine
translation working on large scale, real-world applications such
as the translation of text documents from the world-wide web.
However, from a theoretical perspective our knowledge of the
parsing problem based on SCFGs and of several related tasks is
quite limited, with many questions still left unanswered, as
discussed below.  This is rather surprising, in view of the fact
that SCFGs are a very natural extension of the class of CFGs, for
which the parsing problem has been extensively investigated and is
well understood nowadays.

In the context of statistical machine translation, SCFGs are automatically 
induced from parallel corpora, that is, very large collections of 
source texts that come with target translations, and are usually 
enriched with annotations aligning source and target words \cite{ChiangCL,galley-naacl04}.  
Alternative translation models are currently in use in machine translation, 
such as word-to-word translation models \cite{Brown:93} or phrase-based 
translation models \cite{Koehn-naacl03}, which are essentially finite-state 
models.  However, it has been experimentally shown 
that the more powerful generative capacity of SCFGs achieves better accuracy
than finite-state models in real-world machine translation applications \cite{ChiangCL}.

The \termdef{recognition} (or membership) problem for SCFGs is defined as follows.  
Given as input a SCFG $\gram$ and strings $w_1$ and $w_2$, 
we have to decide whether the pair $w_1, w_2$ can be generated by $\gram$.
The \termdef{parsing} problem for SCFGs (or synchronous parsing problem)
is defined for the same input $\gram$, $w_1$ and $w_2$, 
and produces as output some suitable representation of 
the set of all parse trees in $\gram$ that generate $w_1$ and $w_2$.  
Finally, the \termdef{decoding} (or translation) problem for SCFGs 
requires as input a SCFG $\gram$ and a single string $w_1$, 
and produces as output some suitable representation of the set 
of all parse trees in $\gram$ that generate pairs of the form 
$w_1, w_2$, for some string $w_2$.  
In this paper we investigate the synchronous parsing problem, 
which is strictly related to the other two problems, as
will be discussed in more detail in Section~\ref{disc}. 

From the perspective of synchronous parsing, a crucial difference between CFGs and SCFGs
is that SCFGs cannot always be \termdef{binarized}, that is,
cast into a normal form with no more than two nonterminals
on the right-hand side of each rule component in a synchronous rule.
In fact, SCFGs form an infinite hierarchy, where grammars with at most $r$
nonterminals on the right-hand side of a rule
can generate sets of string pairs not achievable
by grammars with the same quantity bounded by $r-1$, 
for each $r>3$ \cite{AhoUll72}.
Binarization is crucial to standard algorithms
for CFG parsing, whether explicitly as a preprocessing transformation
on the grammar, as for instance in the case of the Cocke-Kasami-Younger 
algorithm \cite{Younger:67}, or implicitly through the use of
dotted rule symbols indicating which nonterminals
have already been parsed, as in the case of Earley algorithm \cite{Earley:70}.  
Unfortunately these techniques cannot be applied to SCFGs, 
because of the above restrictions on binarization, 
and the parsing problem for SCFGs seems significantly more complex 
than the parsing problem for CFGs, from a computational perspective.
While parsing for CFGs can be solved in polynomial space and time 
by the above mentioned algorithms, it has been shown in \cite{Satta:05} 
that parsing for SCFGs is NP-hard, when the grammar is considered as part of the input.

Despite this hardness result, when the SCFG is fixed, parsing 
can be performed in time polynomial in the length of the input strings
using the bottom-up dynamic programming framework
described in Section~\ref{sec:prel}.
The degree of the 
polynomial is determined by the maximum complexity
of any rule in the grammar, because rules are parsed independently of
one another.  More precisely, the complexity of a given rule 
is $\order{n^{d(\pi,\sigma)}}$, where $n$ is the sentence length,
and $d$ is some function of 
$\pi$, the permutation associated with the rule, and $\sigma$,
a parsing \termdef{strategy} for the rule.  
This leads us to consider the 
problem of finding the best strategy for a given rule, that is,
finding the $\sigma$ that minimizes $d(\pi,\sigma)$.
We investigate the problem of finding the best \termdef{linear} parsing strategy
for a given synchronous rule, that is, the ways of collecting one after the other
the linked nonterminals in a synchronous rule that result in the
optimization of the space and time complexity for synchronous parsing.
We show that this task is NP-hard.  
This solves an open problem that has been addressed in several 
previously published works; see for instance \cite{gildea-stefankovic:2007:main}, 
\cite{Huang:2009}, and~\cite{head-driven-acl11}.

\paragraph*{Relation with previous work}
The problem that we explore in this article is an instance of 
a \termdef{grammar factorization} problem.  Factorization 
is the general method of breaking a grammar 
rule into a number of equivalent, smaller rules,
usually for the purpose of finding efficient
parsing algorithms.  Our linear parsing strategies
for SCFGs process two linked nonterminals from the
original rule at each step.  Each of these steps is equivalent
to applying a binary rule in another rewriting system,
which must be more general than SCFGs, since
after all SCFGs cannot be binarized.  

The general problem of grammar factorization has received 
a great deal of study recently in the field of computational 
linguistics, with the rise of statistical systems
for natural language translation, as well as systems for
parsing with monolingual grammars that are more powerful
than CFGs \cite{Huang:2009,gomez-naacl09,sagot-satta:2010:ACL,gomez-naacl10}.  
Most work in this area addresses some subclasses
of the very general rewriting framework known as 
Linear Context-Free Rewriting Systems (LCFRS) \cite{LCFRS},
which is equivalent to Multiple Context-Free Grammars \cite{Seki91},
and which subsumes SCFGs and many other formalisms.
Many algorithms have been proposed for efficiently
factorizing subclasses of LCFRS, in order to optimize 
parsing under various criteria.
Our result in this article is a hardness result,
showing that such algorithms cannot generalize
to the widely used and theoretically important 
class of SCFGs.  A related result has been presented 
by Crescenzi et al.\ \cite{head-driven-acl11},
showing that optimal linear factorization for 
general LCFRS is NP-hard.  Their reduction involves 
constructing LCFRS rules that are not valid as SCFG rules.
Indeed, as already mentioned, SCFG rules can be viewed as permutations, and the special structure of these objects makes reductions less straightforward than in the case of LCFRS.
This article therefore strengthens the result in \cite{head-driven-acl11},
showing that even if we restrict ourselves to SCFGs, 
detection of optimal linear parsing strategies is still NP-hard.

\section{Preliminaries}
\label{sec:prel}

In this section we formally 
introduce the class of synchronous context-free grammars, along with the computational problem that we investigate in this article. 
We assume the reader to be familiar with basic definitions from formal language theory,
and we only briefly summarize here the adopted notation.

For any positive integer $n$, we write $[n]$ to denote the set $\{1,\ldots,n\}$,
and for $n=0$ we let $[n]$ be the empty set.  
We also write $[n]_0$ to denote the set $[n] \cup \{0\}$.

\subsection{Synchronous Context-Free Grammars}
\label{ssec:scfg}

Let $\myalph$ be some finite alphabet.  A string $x$ over $\myalph$ is a finite ordered sequence of symbols from $\myalph$.  The length of $x$ is written $\size{x}$; the empty string is denoted by $\ep$ and we have $\size{\ep} = 0$.  
We write $\myalph^\ast$ for the set of all strings over $\myalph$, 
and $\myalph^+ = \myalph^\ast \setminus \{ \ep \}$. 
For strings $x, y \in \myalph^\ast$, $x \cdot y$ denotes the concatenation of $x$ and $y$, which we abbreviate as $xy$. 

A \termdef{context-free grammar} (CFG for short) is a tuple
$\gram = (\nonterm, \myterm, \myrule, S)$, where $\nonterm$ is a finite set of nonterminals,
$\myterm$ is a finite set of terminals with $\myterm \cap \nonterm = \emp$, 
$S \in \nonterm$ is a special symbol called the start symbol, 
and $\myrule$ is a finite set of rules having the form 
$A \de \alpha$, with $A \in \nonterm$ and $\alpha \in (\myterm \cup \nonterm)^*$.
The size of a CFG $\gram$ is defined as 
$\size{\gram} = \sum_{(A \de \gamma) \in \myrule} \; \size{A \gamma}$. 

The \termdef{derive} relation associated with a CFG $\gram$ is written~$\deriv{\gram}$;
we also use the reflexive and transitive closure of $\deriv{\gram}$, 
written $\derivv{\gram}{\ast}$. 
The language (set of strings) derived in $\gram$ is defined as 
$L(\gram) = \{ w \mid S \derivv{\gram}{\ast} w, \; w \in \myterm^\ast\}$. 

In what follows we need to represent bijections between the
occurrences of nonterminals in two strings over $\nonterm \cup \myterm$.  This
can be realized by annotating nonterminals with \termdef{indices} from an infinite
set.  In this article, we draw indices from the set of positive natural numbers $\N$. 
We define $\myindex{\nonterm} = \{A^\boxnum{$t$} \mid A \in \nonterm, \;
t \in \N \}$ and $\myivar = \myindex{\nonterm} \cup \myterm$.  For a string
$\gamma \in \myivar^\ast$, we write $\setindex{\gamma}$ to denote the set of
all indices that appear in symbols in $\gamma$.

Two strings $\gamma_1, \gamma_2 \in \myivar^\ast$ are \termdef{synchronous} if
each index from
$\N$ occurs at most once in $\gamma_1$ and at most once in $\gamma_2$, and
$\setindex{\gamma_1} = \setindex{\gamma_2}$.  Therefore $\gamma_1, \gamma_2$
have the general form
\begin{align*}
\gamma_1 &= u_{1,0} A_{1,1}^\boxnum{$t_1$}
u_{1,1} A_{1,2}^\boxnum{$t_2$} u_{1,2}  \cdots u_{1,r-1} A_{1,r}^\boxnum{$t_r$}
u_{1,r} \\
\gamma_2 &= u_{2,0} A_{2,1}^\boxnum{$t_{\pi(1)}$} u_{2,1}
A_{2,2}^\boxnum{$t_{\pi(2)}$} u_{2,2} \cdots u_{2,r-1}
A_{2,r}^\boxnum{$t_{\pi(r)}$} u_{2,r}
\end{align*}
where $r \geq 0$, $u_{1,i},
u_{2,i} \in \myterm^\ast$ for each $i \in [r]_0$, 
$A_{1,i}^\boxnum{$t_i$}, A_{2,i}^\boxnum{$t_{\pi(i)}$} \in
\myindex{\nonterm}$ for each $i \in [r]$, 
$t_i \neq t_j$ for $i \neq j$, and $\pi$ is a
permutation of the set $[r]$.
Note that, under the above convention, nonterminals $A_{1,i}, A_{2,\pi^{-1}(i)}$
appear with the same index $t_i$, for each $i \in [r]$.  
In a pair of synchronous strings, two nonterminal occurrences 
with the same index are called \termdef{linked} nonterminals. 

A \termdef{synchronous context-free grammar} (SCFG) is a tuple $\gram =
(\nonterm, \myterm, \myrule, S)$,%
\footnote{We overload symbol $\gram$.  It will always be clear from the
context whether $\gram$ refers to a CFG or to a SCFG.}
where $\nonterm$, $\myterm$ and $S$ are defined as for CFGs, 
and $\myrule$ is a finite set of
\termdef{synchronous rules}.  Each synchronous rule has the form $[A_1 \de
\alpha_1, \;\; A_2 \de \alpha_2]$, where $A_1, A_2 \in \nonterm$ and where
$\alpha_1, \alpha_2 \in \myivar^\ast$ are synchronous strings.  
We refer to $A_1 \de \alpha_1$ and $A_2 \de \alpha_2$, respectively, 
as the left and right \termdef{components} of the synchronous rule.
Note that if we ignore the indices annotating the nonterminals in 
$\alpha_1$ and $\alpha_2$, then $A_1 \de \alpha_1$ and $A_2 \de \alpha_2$ 
are context-free rules. 

\begin{example}
\label{ex:scfg}
The list of synchronous rules reported below implicitly defines a SCFG.  Symbols
$s_i$ are rule labels, to be used as references in later examples. 
\[
\begin{array}{ll}
s_1: & [S \de A^\boxnum{1} \; B^\boxnum{2} , \; S \de B^\boxnum{2} \; A^\boxnum{1}  ] \\
s_2: & [A \de a A^\boxnum{1} b , \; A \de b A^\boxnum{1} a ] \\
s_3: & [A \de a b , \; A \de b a ] \\
s_4: & [B \de c B^\boxnum{1} d , \; B \de d B^\boxnum{1} c ] \\
s_5: & [B \de c d , \; B \de d c ]
\end{array}
\]
\end{example}

We now define the notion of derivation associated with a SCFG. 
In a derivation we rewrite a pair of synchronous strings, 
always producing a new pair of synchronous strings.
This is done in several steps, where, at each step, 
two linked nonterminals are rewritten by a synchronous rule. 
We use below the auxiliary notion of
\termdef{reindexing}, which is an injective function $f$ from $\N$ to $\N$.
We extend $f$ to $\myivar$ by letting
$f(A^\boxnum{$t$}) = A^\boxnum{$f(t)$}$ for $A^\boxnum{$t$} \in \myindex{\nonterm}$ and $f(a) = a$ for $a \in \myterm$.
We also extend $f$ to strings in $\myivar^\ast$ by letting $f(\ep) = \ep$ and
$f(X\gamma) = f(X)f(\gamma)$, for each $X \in \myivar$ and $\gamma \in \myivar^\ast$.

Let $\gamma_1, \gamma_2 \in \myivar^\ast$ be two synchronous strings. The
\termdef{derive} relation $[\gamma_1, \;\; \gamma_2] \deriv{\gram}
[\delta_1, \;\; \delta_2]$ holds whenever there exist an index $t$ in
$\setindex{\gamma_1}=\setindex{\gamma_2}$, a synchronous rule $s \in \myrule$ 
of the form $[A_1 \de \alpha_1, \; A_2 \de \alpha_2]$, and some reindexing $f$ such that
\begin{enumerate}
\item
$\setindex{f(\alpha_1)} \cap ( \setindex{\gamma_1} \setminus \{ t \} ) = \emptyset$;
\item
$\gamma_{1} = \gamma'_{1} A_{1}^\boxnum{$t$} \gamma''_{1}$,
$\gamma_{2} = \gamma'_{2} A_{2}^\boxnum{$t$} \gamma''_{2}$;  and
\item
$\delta_{1} = \gamma'_{1} f(\alpha_1) \gamma''_{1}$, $\delta_{2} =
\gamma'_{2} f(\alpha_2) \gamma''_{2}$.
\end{enumerate}
We also write $[\gamma_1, \; \gamma_2] \derivv{\gram}{s} [\delta_1, \; \delta_2]$
to explicitly indicate that the derive relation holds through rule $s$.

Note that $\delta_1, \delta_2$ above are guaranteed to be synchronous strings,
because $\alpha_1$ and $\alpha_2$ are synchronous strings and because of condition~(i) above.  Note also that, for a given pair $[\gamma_1, \; \gamma_2]$ of synchronous strings,
an index $t$ and a synchronous rule as above, there may be infinitely many choices of a reindexing $f$ such that the above constraints are satisfied.
However, all essential results about SCFGs are independent of the specific choice of 
reindexing, and therefore we will not further discuss this issue here.

A derivation in $\gram$ is a sequence $\sigma = s_1 s_2 \cdots s_d$, $d \geq 0$, of synchronous rules $s_i \in \myrule$, $i \in [d]$, with $\sigma = \ep$ for $d = 0$, satisfying the following property.  
For some pairs of synchronous strings $[\gamma_{1,i}, \; \gamma_{2,i}]$, 
$i \in [d]_0$, we have 
$[\gamma_{1,i-1}, \; \gamma_{2,i-1}] \derivv{\gram}{s_i} [\gamma_{1,i}, \;
\gamma_{2,i}]$ for each $i \in [d]$.
We always implicitly assume some canonical form for derivations in $\gram$, by
demanding for instance that each step rewrites a pair of linked nonterminal
occurrences of which the first is leftmost in the left component.
When we want to focus on the specific synchronous strings
being derived, we also write derivations in the form $[\gamma_{1,0}, \;
\gamma_{2,0}] \derivv{\gram}{\sigma} [\gamma_{1,d}, \; \gamma_{2,d}]$, and we write
$[\gamma_{1,0}, \; \gamma_{2,0}] \derivv{\gram}{\ast} [\gamma_{1,d}, \; \gamma_{2,d}]$
when $\sigma$ is not further specified.  The \termdef{translation} 
generated by a
SCFG $\gram$ is defined as
\begin{eqnarray*}
T(\gram) & = & \{ [w_1, \; w_2] \mid [S^\boxnum{1}, \; S^\boxnum{1}]
 \derivv{\gram}{\ast} [w_1, \; w_2], \; w_1, w_2 \in \myterm^\ast \} \,.
\end{eqnarray*}

\begin{example}
Consider the SCFG $\gram$ from Example~\ref{ex:scfg}.
The following is a (canonical) derivation in $\gram$
\begin{eqnarray*}
[S^\boxnum{1}, \; S^\boxnum{1}]
  & \derivv{\gram}{s_1} & [A^\boxnum{1} B^\boxnum{2}, B^\boxnum{2} A^\boxnum{1}] \\
  & \derivv{\gram}{s_2} & [a A^\boxnum{3} b B^\boxnum{2}, 
        B^\boxnum{2} b A^\boxnum{3} a] \\
  & \derivv{\gram}{s_2} & [aa A^\boxnum{4} bb B^\boxnum{2}, 
        B^\boxnum{2} bb A^\boxnum{4} aa] \\
  & \derivv{\gram}{s_3} & [aaabbb B^\boxnum{2}, B^\boxnum{2} bbbaaa] \\
  & \derivv{\gram}{s_4} & [aaabbb c B^\boxnum{5} d, d B^\boxnum{5} c bbbaaa] \\
  & \derivv{\gram}{s_5} & [aaabbb ccdd, ddcc bbbaaa]
\end{eqnarray*}
It is not difficult to see that 
$T(\gram) = \{ [a^p b^p c^q d^q, \; d^q c^q b^p a^p] \mid p, q \geq 1 \}$. 
\end{example}

We conclude this section with a remark.
Our definition of SCFG is essentially the same as the definition of
the syntax-directed transduction grammars in \cite{LewisStearns68}
and the syntax-directed translation schemata in \cite{Aho:69a,AhoUll72},
as already mentioned in the introduction. 
The only difference is that in a synchronous rule
$[A_1 \de \alpha_1, \; A_2 \de \alpha_2]$ we allow $A_1, A_2$ 
to be different nonterminals, while in the above formalisms 
we always have $A_1 = A_2$.  
Although our generalization does not add to the weak generative power 
of the model, that is, the class of translations generated by the two models
are the same, it does increase its strong generative capacity,
that is, the parse tree mappings defined by 
syntax-directed translation schemata are a proper
subset of the parse tree mappings defined by SCFGs.  
As a consequence of this fact, when the definitions of 
the two models are enriched with probabilities, 
SCFGs can define certain
parse tree distributions that cannot be captured by 
syntax-directed translation schemata, as argued in \cite{Satta:05}.  
The above generalization has been adopted 
in several translation models for natural language.

\subsection{Parsing Strategies for SCFGs}
\label{ssec:membership}

Recognition and parsing algorithms for SCFGs are 
commonly used in the area of statistical machine translation.
Despite the fact that the underlying problems are NP-hard, 
it has been experimentally shown that 
the typology of synchronous rules that we encounter in real world 
applications can be processed efficiently, for the most of the cases, 
if we adopt the appropriate parsing strategy, as already 
discussed in Section \ref{sec:intro}.   We are thus interested 
in the problem of finding optimal parsing strategies for 
synchronous rules, under some specific parsing framework.

Standard parsing algorithms for SCFGs exploit dynamic programming techniques,
and are derived as a generalization of the well-known Cocke-Kasami-Younger 
algorithm for word recognition based on CFGs \cite{Younger:67,Hopcroft:79},
which essentially uses a bottom-up strategy. 
All these algorithms are based on the representation described below.  
For a string $w = a_1 \cdots a_n$, $n \geq 1$, and for integers $i, j \in [n]_0$ 
with $i<j$, we write $w[i, j]$ to denote the substring $a_{i+1} \cdots a_j$.
Assume we are given the input pair $[w_1, \; w_2]$. 
To simplify the discussion, we focus on a sample synchronous rule 
containing only occurrences of nonterminal symbols
\begin{eqnarray}
s: && [A_1 \de A_{1,1}^\boxnum{$1$} A_{1,2}^\boxnum{$2$} A_{1,3}^\boxnum{$3$}
   A_{1,4}^\boxnum{$4$} A_{1,5}^\boxnum{$5$} A_{1,6}^\boxnum{$6$}, \nonumber\\
 &&  \; A_2 \de A_{2,1}^\boxnum{$6$} A_{2,2}^\boxnum{$1$} A_{2,3}^\boxnum{$4$} 
   A_{2,4}^\boxnum{$2$} A_{2,5}^\boxnum{$5$}  A_{2,6}^\boxnum{$3$}], \label{eq:rule-s}
\end{eqnarray}
Synchronous rule $s$ can be associated with the permutation $\pi$ 
of the set $[6]$ identified by the sequence $614253$,
which is visualized in Figure~\ref{fig:strategy}a. 
Recall that, for each $k \in [6]$, nonterminals $A_{1,k}, A_{2,\pi^{-1}(k)}$
are linked in rule $s$. 

Assume that, for each $k \in [6]$, 
we have already parsed linked nonterminals $A_{1,k}, A_{2,\pi^{-1}(k)}$, 
realizing that $[A_{1,k}^\boxnum{1}, A_{2,\pi^{-1}(k)}^\boxnum{1}] 
\derivv{\gram}{\ast} [w_1[i_{1,k}, j_{1,k}], \; 
w_2[i_{2,\pi^{-1}(k)}, j_{2,\pi^{-1}(k)}]]$,
for integers $i_{1,k}, j_{1,k} \in [\size{w_1}]_0$ and 
$i_{2,\pi^{-1}(k)}, j_{2,\pi^{-1}(k)} \in [\size{w_2}]_0$.
Informally, we say that linked nonterminals $A_{1,k}, A_{2,\pi^{-1}(k)}$ span
substrings $w_1[i_{1,k}, j_{1,k}]$ and $w_2[i_{2,\pi^{-1}(k)}, j_{2,\pi^{-1}(k)}]$
of the input.  In order to parse the left-hand side nonterminals $A_{1}, A_{2}$, 
we need to gather together all linked nonterminals $A_{1,k}, A_{2,\pi^{-1}(k)}$, and check 
whether the combination of the above derivations can provide a new derivation
of the form $[A_{1}^\boxnum{1}, A_{2}^\boxnum{1}] 
\derivv{\gram}{\ast} [w_1[i_{1,1}, j_{1,6}], \; w_2[i_{2,1}, j_{2,6}]]$.
In other words, we want to know whether the combination 
of the derivations for linked nonterminals $A_{1,k}, A_{2,\pi^{-1}(k)}$
can span two (contiguous) substrings of the input. 

\begin{figure}
\resizebox{4in}{!}{
\begin{tabular}{lp{5.0in}p{5.0in}}
a)\\
&{\begin{pspicture}(-2,-2)(13,7)
\psset{fillcolor=gray}
\rput(-1.5, 6.25){$(A_{1,6},A_{2,1})$}
\rput(-1.5, 5.25){$(A_{1,5},A_{2,5})$}
\rput(-1.5, 4.25){$(A_{1,4},A_{2,3})$}
\rput(-1.5, 3.25){$(A_{1,3},A_{2,6})$}
\rput(-1.5, 2.25){$(A_{1,2},A_{2,4})$}
\rput(-1.5, 1.25){$(A_{1,1},A_{2,2})$}
\rput(-1.5, 0.25){$(A_1,A_2)$}
\psframe[fillstyle=solid](0, 1)(1, 1.5)
\psframe[fillstyle=solid](1, 2)(2, 2.5)
\psframe[fillstyle=solid](2, 3)(3, 3.5)
\psframe[fillstyle=solid](3, 4)(4, 4.5)
\psframe[fillstyle=solid](4, 5)(5, 5.5)
\psframe[fillstyle=solid](5, 6)(6, 6.5)
\psframe[fillstyle=solid](7, 6)(8, 6.5)
\psframe[fillstyle=solid](8, 1)(9, 1.5)
\psframe[fillstyle=solid](9, 4)(10, 4.5)
\psframe[fillstyle=solid](10, 2)(11, 2.5)
\psframe[fillstyle=solid](11, 5)(12, 5.5)
\psframe[fillstyle=solid](12, 3)(13, 3.5)
\psframe[fillstyle=solid](-.2, 0.75)(13.2, 0.75) 
\psframe[fillstyle=solid](0, 0)(6, 0.5)
\psframe[fillstyle=solid](7, 0)(13, 0.5)
\end{pspicture} }
\\
b)\\
&{\begin{pspicture}(-2,-2)(9,3)
\psset{fillcolor=gray}
\rput(-1.5, 2.25){$(A_{1,2},A_{2,4})$}
\rput(-1.5, 1.25){$(s,\sigma,1)$}
\rput(-1.5, 0.25){$(s,\sigma,2)$}
\psframe[fillstyle=solid](0, 1)(1, 1.5)
\psframe[fillstyle=solid](1, 2)(2, 2.5)
\psframe[fillstyle=solid](8, 1)(9, 1.5)
\psframe[fillstyle=solid](10, 2)(11, 2.5)
\psframe[fillstyle=solid](-.2, 0.75)(13.2, 0.75) 
\psframe[fillstyle=solid](0, 0)(2, 0.5)
\psframe[fillstyle=solid](8, 0)(9, 0.5)
\psframe[fillstyle=solid](10, 0)(11, 0.5)
\end{pspicture} }
\\
c)\\
&{\begin{pspicture}(-2,-2)(9,3)
\psset{fillcolor=gray}
\rput(-1.5, 2.25){$(A_{1,3},A_{2,6})$}
\rput(-1.5, 1.25){$(s,\sigma,2)$}
\rput(-1.5, 0.25){$(s,\sigma,3)$}
\psframe[fillstyle=solid](2, 2)(3, 2.5)
\psframe[fillstyle=solid](0, 1)(2, 1.5)
\psframe[fillstyle=solid](12, 2)(13, 2.5)
\psframe[fillstyle=solid](8, 1)(9, 1.5)
\psframe[fillstyle=solid](10, 1)(11, 1.5)
\psframe[fillstyle=solid](-.2, 0.75)(13.2, 0.75) 
\psframe[fillstyle=solid](0, 0)(3, 0.5)
\psframe[fillstyle=solid](8, 0)(9, 0.5)
\psframe[fillstyle=solid](10, 0)(11, 0.5)
\psframe[fillstyle=solid](12, 0)(13, 0.5)
\end{pspicture} }
\\
d)
\\
&{\begin{pspicture}(-2,-1)(13,3)
\psset{fillcolor=gray}
\rput(-1.5, 2.25){$(A_{1,4},A_{2,3})$}
\rput(-1.5, 1.25){$(s,\sigma,3)$}
\rput(-1.5, 0.25){$(s,\sigma,4)$}
\psframe[fillstyle=solid](3, 2)(4, 2.5) 
\psframe[fillstyle=solid](0, 1)(3, 1.5)  
\psframe[fillstyle=solid](9, 2)(10, 2.5) 
\psframe[fillstyle=solid](8, 1)(9, 1.5) 
\psframe[fillstyle=solid](10, 1)(11, 1.5)
\psframe[fillstyle=solid](12, 1)(13, 1.5) 
\psframe[fillstyle=solid](-.2, 0.75)(13.2, 0.75) 
\psframe[fillstyle=solid](0, 0)(4, 0.5)
\psframe[fillstyle=solid](8, 0)(11, 0.5)
\psframe[fillstyle=solid](12, 0)(13, 0.5)
\end{pspicture} }
\\
\end{tabular}}
\caption{a): combining spans
to parse the SCFG rule $s$ of eqn.~(\ref{eq:rule-s}).
b), c) and d): the first three steps in a linear parsing
strategy for this rule.}\label{fig:strategy}
\end{figure}
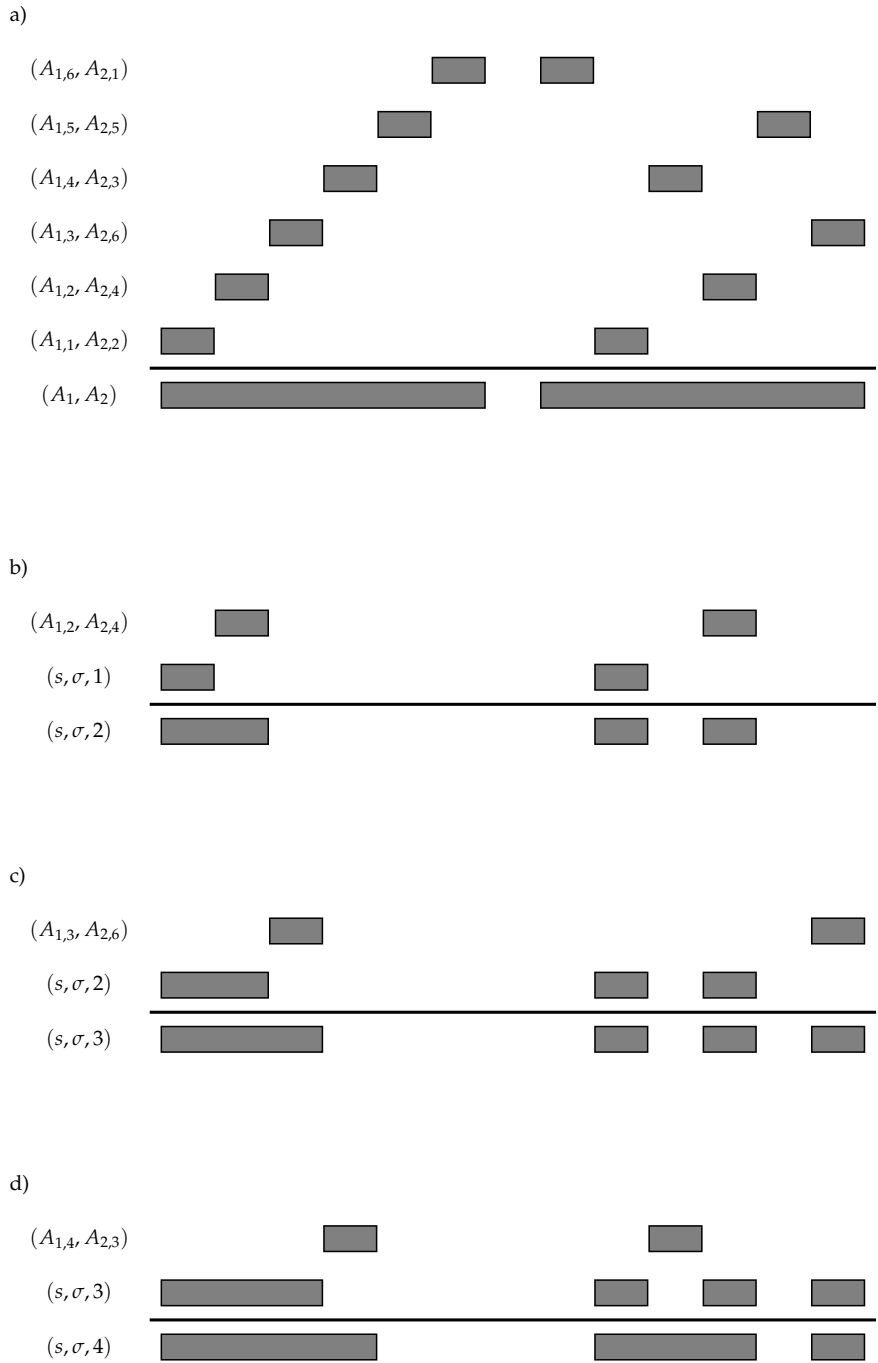

For reasons of computational efficiency, it is advantageous 
to break the parsing of a synchronous rule into several steps, 
adding linked nonterminals one pair at each step,
according to some fixed total ordering, which we call 
a \termdef{linear} parsing strategy. The result of 
the partial analyses obtained at each step is represented by
means of a data structure which we call a \termdef{state}.
To provide a concrete example, let us choose the linear parsing strategy 
$\sigma$ of gathering all the $A_{1k}$'s on the first component 
of rule $s$ from left to right.  At the first step we then collect 
linked nonterminals $A_{1,1}, A_{2,\pi^{-1}(1)} = A_{2,2}$ and construct
the partial analysis represented by the state  
$\langle (s, \sigma, 1), (i_{1,1}, j_{1,1}), (i_{2,2}, j_{2,2}) \rangle$, 
meaning that $A_{1,1}$ spans substring $w_1[i_{1,1}, j_{1,1}]$ 
and $A_{2,2}$ spans substring $w_2[i_{2,2}, j_{2,2}]$.
The first element in the state, $(s, \sigma, 1)$, indicates that
this state is generated from synchronous rule $s$ after the first combination step,
assuming our current strategy $\sigma$.  We refer to this first element as
the \termdef{type} of the state.  

Assuming now that $j_{1,1} = i_{1,2}$, at the second step we 
add to our partial analysis the linked nonterminals $A_{1,2}, A_{2,4}$, 
as shown in Figure~\ref{fig:strategy}b.  We construct
a new state $\langle (s, \sigma, 2),$ $(i_{1,1}, j_{1,2}),$ $(i_{2,2}, j_{2,2}),$ 
$(i_{2,4}, j_{2,4}) \rangle$, meaning that 
$A_{1,1}, A_{1,2}$ together span $w_1[i_{1,1}, j_{1,2}]$, 
$A_{2,2}$ spans $w_2[i_{2,2}, j_{2,2}]$ and 
$A_{2,4}$ spans $w_2[i_{2,4}, j_{2,4}]$.  
Note here that the specific value of $j_{1,1} = i_{1,2}$ is dropped from 
the description of the state, since it will not be
referenced by any further step based on the associated partial analysis.
After adding the third pair of linked nonterminals $A_{1,3}, A_{2,6}$,
and assuming $j_{1,2} = i_{1,3}$,
we create state $\langle (s, \sigma, 3),$ $(i_{1,1}, j_{1,3}),$ $(i_{2,2}, j_{2,2}),$
$(i_{2,4}, j_{2,4}),$ $(i_{2,6}, j_{2,6}) \rangle$, 
spanning four separate substrings of the input,
as shown in Figure~\ref{fig:strategy}c.
Assume now that $j_{2,2} = i_{2,3}$ and $j_{2,3} = i_{2,4}$. 
After adding the linked nonterminals $A_{1,4}, A_{2,3}$, we have that 
the span of $A_{2,3}$ fills in the gap between the spans of 
the previously parsed nonterminals $A_{2,2}$ and $A_{2,4}$,
as shown in Figure~\ref{fig:strategy}d.  We can then collapse
these three spans into a single string, obtaining a new state
$\langle (s, \sigma, 4),$ $(i_{1,1}, j_{1,4}),$ $(i_{2,2}, j_{2,4}),$
$(i_{2,6}, j_{2,6}) \rangle$ which spans three separate substrings of the input.  
Finally, assuming that $j_{2,1} = i_{2,2}$ and $j_{2,4} = i_{2,5}$,
at the next two steps states of type $(s, \sigma, 5)$ and $(s, \sigma, 6)$ 
can be constructed, each spanning two substrings only.

We refer below to the number of substrings spanned by a state
as the \termdef{fan-out} of the state (this notion will be 
formally defined later).  The above example shows that 
the fan-out of each state depends on the 
parsing strategy that we are adopting. 

Bottom-up dynamic programming algorithms for the parsing problem 
for SCFGs are designed 
on the basis of the above state representation for partial analyses. 
These algorithms store in some appropriate data structure 
the states that have already been constructed, 
and then retrieve and combine states in order to construct new states. 
Let $n$ be the maximum length between the input strings $w_1$ and $w_2$.
Because a state with fan-out $f$ may have $\order{n^{2f}}$ instantiations, 
fan-out provides a way of bounding the space complexity of our algorithm.
When we use linear parsing strategies, fan-out is also relevant 
in assessing upper bounds on time complexity.  Consider the basic step of
adding the $(k+1)$-th pair of linked nonterminals to a state of type $(s, \sigma, k)$ 
having fan-out $f$.   As before, we have $\order{n^{2f}}$ instantiations 
for states of type $(s, \sigma, k)$.  We also have $\order{n^4}$ 
possible instantiations for the span of the $(k+1)$-th pair, 
since any pair of linked nonterminals spans exactly two substrings.
However, the $(k+1)$-th pair might share some of its boundaries 
with the boundaries of the state of type $(s, \sigma, k)$,
depending on the permutation associated with the synchronous rule $s$. 
If we define $\delta(s, \sigma, k)$ as the number of independent boundaries 
in the $(k+1)$-th pair, with $0 \leq \delta(s, \sigma, k) \leq 4$, 
we have that all executions of the above step can 
be carried out in time $\order{n^{2f+\delta(s, \sigma, k)}}$.

If we want to optimize the space or the time complexity of a dynamic
programming algorithm cast in the above framework, 
we need to search for a parsing strategy that minimizes the maximum
fan-out of its states, or else a strategy that minimizes 
the maximum value of the sum of the fan-out and the $\delta()$ function.  
This needs to be done for each individual synchronous rule in the grammar. 
In our running example, the critical step is provided by state 
type $(s, \sigma, 3)$ with fan-out 4, leading to space complexity 
of $\order{n^{8}}$.  Furthermore, the combination of state type $(s, \sigma, 2)$ 
(fan-out 3) with linked pair $A_{1,3}, A_{2,6}$ ($\delta(s, \sigma, 3) = 3$) 
leads to time complexity of $\order{n^{9}}$.
However, we can switch to a different strategy $\sigma'$, by collecting 
linked nonterminal pairs in $s$ in the order given by the left components
$A_{1,4}, A_{1,5}, A_{1,2}, A_{1,3}, A_{1,1}, A_{1,6}$.  
According to this new strategy, states of types $(s, \sigma', 2)$
and $(s, \sigma', 3)$ both have fan-out three, while every other state type 
has fan-out two.  This leads to space complexity of $\order{n^{6}}$ for 
rule $s$.  It is not difficult to see that this strategy is also space optimal
for rule $s$, on the basis of the observation that any grouping of two 
linked nonterminals $A_{1,k}, A_{2,\pi^{-1}(k)}$ and 
$A_{1,k'}, A_{2,\pi^{-1}(k')}$ with $k, k' \in [6]$ and $k \neq k'$,
has a fan-out of at least three.
As for the time complexity,  the critical step is the combination 
of state type $(s, \sigma', 2)$ (fan-out 3) with linked pair 
$A_{1,2}, A_{2,4}$ ($\delta(s, \sigma', 3) = 2$), leading to 
time complexity of $\order{n^{8}}$ for this strategy.  
It is not difficult to verify that $\sigma'$ is also a time optimal strategy.

\subsection{Fan-out and Optimization of Parsing}
\label{ssec:fo}

What we have informally shown in the previous section is that, 
under the outlined framework based on state representations for partial analyses, 
we can exploit the properties of the specific permutation 
of a given synchronous rule to reduce the maximum fan-out of states, 
and hence improve the space and time complexity of our parsing algorithms.
In this section, we provide formal definitions of these concepts, 
and introduce the computational problem that is investigated in this article. 

Let $s$ be a synchronous rule with $r > 2$ linked nonterminals, 
and let $\pi_s$ be the permutation representing $s$.
A \termdef{linear parsing strategy} for $s$ is defined as 
a permutation $\sigma_s$ of the set $[r]$.  The intended meaning 
of $\sigma_s$ is that, when parsing the rule $s$, 
the pair of linked nonterminals $A_{1,\sigma(k)}, A_{2,\pi^{-1}(\sigma(k))}$
is collected at the $k$-th step, for each $k \in [r]$, 
as shown in Figure~\ref{fig:strategy}. 

Let us consider state type $(s, \sigma_s, k)$, $k \in [r]$,
defined as in Section~\ref{ssec:membership}.  
We define the count of \termdef{internal boundaries} for $(s, \sigma_s, k)$ as
\begin{align}
\ibd(\pi_s, \sigma_s, k) = & \left| \left\{ h : \sigma_s^{-1}(h) \leq k \wedge \sigma_s^{-1}(h+1) > k \right\} \right| + &   \nonumber \\
& \left|\left\{ h : \sigma_s^{-1}(h) > k \wedge \sigma_s^{-1}(h+1) \leq k \right\}\right| +\nonumber \\
& \left|\left\{ h : \sigma_s^{-1}(\pi_s^{-1}(h)) \leq k \wedge \sigma_s^{-1}(\pi_s^{-1}(h+1)) > i \right\}\right| + \nonumber \\
& \left|\left\{ h : \sigma_s^{-1}(\pi_s^{-1}(h)) > k \wedge \sigma_s^{-1}(\pi_s^{-1}(h+1)) \leq k \right\}\right| \; . 
\label{eq:ibd}
\end{align}
In the definition above, the term $\left|\left\{ h : \sigma_s^{-1}(j) 
\leq i \wedge \sigma_s^{-1}(h+1) > k \right\}\right|$
counts the number of nonterminals $A_{1,i}$ that have already 
been collected at step $k$ and such that nonterminal $A_{1,i+1}$
has not yet been collected.  Informally, this term 
counts the nonterminals in the right-hand side of the first 
CFG rule component of $s$ that represent right internal boundaries
of the span of a state of type $(s, \sigma_s, k)$. 
The second term in the definition counts the number of nonterminals 
in the right-hand side of the same rule component that represent left 
internal boundaries.
Similarly, the remaining two terms count right and left internal
boundary nonterminals, respectively, 
in the right-hand side of the second CFG rule component of $s$. 

For state type $(s, \sigma_s, k)$, $k \in [r]$, we also define 
the count of \termdef{external boundaries} as
\begin{align}
\ebd(\pi_s, \sigma_s, k) = & I(\sigma_s^{-1}(1) \le k) +
 I(\sigma_s^{-1}(n) \le k) + \nonumber \\
& I(\sigma_s^{-1}(\pi_s^{-1}(1)) \le k) +
 I(\sigma_s^{-1}(\pi_s^{-1}(n)) \le k) \; . 
\label{eq:ebd}
\end{align}
The indicator functions $I()$ count the 
number of nonterminals that are placed at the left and right
ends of the right-hand sides of the two rule components
and that have already been collected at step $k$.
Informally, the sum of these functions counts 
the nonterminals that represent external boundaries
of the span of state type $(s, \sigma_s, k)$. 

Finally, the \termdef{fan-out} of state type $(s, \sigma_s, k)$ is defined as
\begin{align}
\fo(\pi_s, \sigma_s, k) = & \frac{1}{2} 
  (\ibd(\pi_s, \sigma_s, k) + \ebd(\pi_s, \sigma_s, k)) \; . 
\label{eq:fo}
\end{align}
Dividing the total number of boundaries by two gives the number of 
substrings spanned by the state type $(s, \sigma_s, k)$.
Observe that the fan-out at step $k$ is a function of both the permutation $\pi_s$
associated with the SCFG rule $s$, and the linear parsing strategy $\sigma_s$.

As discussed in Section~\ref{ssec:membership}, the fan-out at step 
$k$ gives space and time bounds 
on the parsing algorithm relative to that step and parsing strategy $\sigma_s$.  
Thus the complexity of the parsing algorithm relative to synchronous rule $s$ 
depends on the fan-out at the most complex step of $\sigma_s$.  We wish to
find, for an input synchronous rule $s$ with associated permutation $\pi_s$, 
the linear parsing strategy that minimizes quantity
\begin{equation}
 \min_{\sigma} \; \max_{k \in [r]} \; \fo(\pi_s, \sigma, k) \; , 
\label{eq:minfanout}
\end{equation}
where $\sigma$ ranges over all possible linear parsing strategies for $s$. 
Our main result in this article is that this minimization problem is NP-hard.
This is shown by first proving that the optimization of the 
$\ibd(\pi_s, \sigma_s, k)$ component of the fan-out is NP-hard, in the next section, 
and then by extending the result to the whole fan-out in a successive section.

\section{Permutation Multigraphs and Cutwidth}
\label{sec:internal_boundaries}

With the goal of showing that the minimization problem 
in~(\ref{eq:minfanout}) is NP-hard, in this section we investigate 
the minimization problem for the $\ibd(\pi_s, \sigma_s, k)$ component of the fan-out, 
defined in~(\ref{eq:ibd}).  More precisely, given as input a
synchronous rule $s$ with $r > 2$ nonterminals and with associated permutation $\pi_s$, 
we investigate a decision problem associated with the computation of the quantity
\begin{equation}
 \min_{\sigma} \; \max_{k \in [r]} \; \ibd(\pi_s, \sigma, k) \; , 
\label{eq:minibd}
\end{equation}
where $\sigma$ ranges over all possible linear parsing strategies for $s$. 
We do this by introducing a multigraph representation for the synchronous 
rule $s$, and by studying the so-called cutwidth problem for such a multigraph.

\subsection{Permutation Multigraphs}
\label{ssec:pmg}

Our strategy for proving the NP-hardness of
the optimization problem in~(\ref{eq:minibd}) will be to reduce 
to the problem of finding the cutwidth of a certain class of multigraphs,
which represent the relevant structure of the input synchronous 
rule.  In this section we introduce this class of multigraphs, and
discuss its relation with synchronous rules. 
We denote undirected multigraphs as pairs $G = (V,E)$, 
with set of nodes $V$ and multiset of edges $E$.  

A \termdef{permutation multigraph} is a
multigraph $G = (V, A \uplus B)$ such that both $P_A = (V,A)$ and $P_B =
(V,B)$ are Hamiltonian paths, and $\uplus$ is the merge operation 
defined for multisets.  In the following, the edges in
$A$ will be called \termdef{red}, the edges 
in $B$ will be called \termdef{green}.

A permutation multigraph $G = (V, A \uplus B)$ can be thought of as encoding some permutation:
if we identify nodes in $V$ with integers in $[\size{V}]$ according to their
position on path $A$, the order of vertices along path $B$
defines a permutation of set $[\size{V}]$.  
We can therefore use a permutation multigraph to encode the permutation associated 
with a given synchronous rule.  
More precisely, let $s$ be a synchronous rule of the form
\begin{eqnarray}
\label{e:rule_with_pi}
s: && [A_1 \de u_{1,0} A_{1,1}^\boxnum{$1$} u_{1,1} 
         \cdots u_{1,r-1} A_{1,r}^\boxnum{$r$} u_{1,r}, \nonumber \\
   && \; A_2 \de u_{2,0} A_{2,1}^\boxnum{$\pi_s(1)$} u_{2,1} 
         \cdots u_{2,r-1} B_{2,r}^\boxnum{$\pi_s(r)$} u_{2,r}] \; ,
\end{eqnarray}
where $r \geq 2$, $u_{1,i}, u_{2,i} \in \myterm^\ast$ for each $i \in [r]_0$ 
and $\pi_s$ is a permutation of the set $[r]$.
We associate with $s$ the permutation multigraph 
$G_s = (V_s, E_{s,A} \uplus E_{s,B})$ defined as
\begin{itemize}
\item $V_s = \{(A_{1,i}, A_{2,\pi_s^{-1}(i)}) : i \in [r] \}$;
\item $E_{s,A} = \{((A_{1,i},A_{2,\pi_s^{-1}(i)}), (A_{1,j},A_{2,\pi_s^{-1}(j)})) : 
  i,j \in [r] \wedge |i-j| = 1\}$;
\item $E_{s,B} = \{((A_{1,\pi_s(i)},A_{2,i}), (A_{1,\pi_s(j)},A_{2,j})) : 
  i,j \in [r] \wedge |i-j| = 1\}$.
\end{itemize}
To see that $G_s$ is a permutation multigraph, observe that $G_s$ is 
the superposition of the following two Hamiltonian paths
\begin{itemize}
\item $\langle (A_{1,1}, A_{2,\pi_s^{-1}(1)}),((A_{1,2}, A_{2,\pi_s^{-1}(2)}),
  \ldots,(A_{1,r-1}, A_{2,\pi_s^{-1}(r-1)}), (A_{1,r}, A_{2,\pi_s^{-1}(r)})\rangle$;
\item $\langle (A_{1,\pi_s(1)}, A_{2,1}), (A_{1,\pi_s(2)}, A_{2,2}), 
  \ldots, (A_{1,\pi_s(r-1)}, A_{2,r-1}), (A_{1,\pi_s(r)}, A_{2,r})\rangle$.
\end{itemize}
An example permutation multigraph is shown in Figure~\ref{fig:graph},
with one Hamiltonian path shown above and one below the vertices.

\begin{figure}
\begin{center}
\resizebox{\linewidth}{!}{
\begin{tikzpicture}[scale=1]
\draw (0,0) node(1) {$(A_{1,1},A_{2,2})$};
\draw (3,0) node(2) {$(A_{1,2},A_{2,4})$};
\draw (6,0) node(3) {$(A_{1,3},A_{2,6})$};
\draw (9,0) node(4) {$(A_{1,4},A_{2,3})$};
\draw (12,0) node(5) {$(A_{1,5},A_{2,5})$};
\draw (15,0) node(6) {$(A_{1,6},A_{2,1})$};
\draw[color=red] (1) .. controls +(0.5,0.5) and +(-0.5,0.5) .. (2);
\draw[color=red] (2) .. controls +(0.5,0.5) and +(-0.5,0.5) .. (3);
\draw[color=red] (3) .. controls +(0.5,0.5) and +(-0.5,0.5) .. (4);
\draw[color=red] (4) .. controls +(0.5,0.5) and +(-0.5,0.5) .. (5);
\draw[color=red] (5) .. controls +(0.5,0.5) and +(-0.5,0.5) .. (6);
\draw[color=green] (6) .. controls +(0,-2.5) and +(0,-2.5) .. (1);
\draw[color=green] (1) .. controls +(0,-1.5) and +(0,-1.5) .. (4);
\draw[color=green] (4) .. controls +(0,-1) and +(0,-1) .. (2);
\draw[color=green] (2) .. controls +(0,-1.5) and +(0,-1.5) .. (5);
\draw[color=green] (5) .. controls +(0,-1) and +(0,-1) .. (3);
\end{tikzpicture}}
\end{center}
\caption{The permutation multigraph corresponding to the SCFG rule $s$
of eqn.~(\ref{eq:rule-s})}
\label{fig:graph}
\end{figure}
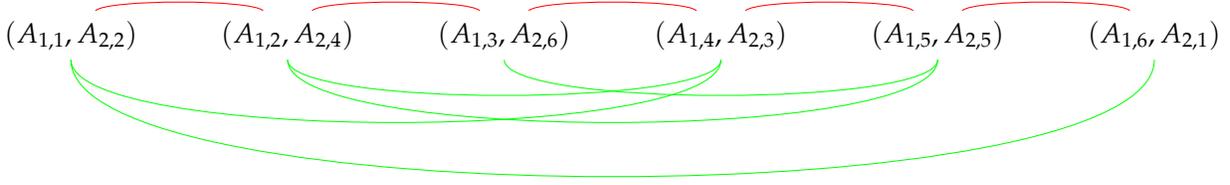

We shall now discuss a mathematical relation between internal boundary
counts for states associated with linear parsing strategies 
for the synchronous rule $s$ and width values for the permutation multigraph $G_s$.  
We first recall the definition of the width and cutwidth of a graph and a multigraph.
Let $G = (V,E)$ be an undirected (multi)graph such that $\size{V} = n >1$.  
A \termdef{linear arrangement} 
of $G$ is a bijective mapping $\nu$ from $V$ to $[n]$.  
We call \termdef{positions} the integer values of $\nu$.  
For any $i \in [n-1]$, the \termdef{width} of $G$ at $i$ with
respect to $\nu$, denoted by $\wg(G,\nu,i)$, is defined as
$\size{\{(u,v) \in E : \nu(u) \leq i < \nu(v)\}}$.   
In the case of a multigraph, the
size of the previous set should be computed taking into account multiple occurrences.  
Informally, $\wg(G,\nu,i)$ is the number of distinct edges crossing over 
the gap between positions $i$ and $i+1$ in the linear arrangement $\nu$.  
To simplify the notation below, we also let $\wg(G,\nu,n) = 0$.  
The \termdef{cutwidth} of $G$ is then defined as 
\begin{eqnarray*}
\cw(G) & = & \min_{\nu} \; \max_{i \in [n]} \; \wg(G,\nu,i) \; ,
\end{eqnarray*}
where $\nu$ ranges over all possible linear arrangements of $G$. 
The cutwidth of the multigraph of Figure~\ref{fig:graph} is
six, which is achieved between $(A_{1,3},A_{2,6})$ and $(A_{1,4},A_{2,3})$ in the linear 
arrangement shown.

Let us now consider synchronous rule $s$ in (\ref{e:rule_with_pi})
and the associated permutation $\pi_s$, 
and let $\sigma_s$ be some linear parsing strategy 
defined for $s$.  The \termdef{linear arrangement associated with} $\sigma_s$
is the linear arrangement $\nu_s$ for permutation multigraph 
$G_s = (V_s, E_s)$ defined as follows.
For each $i \in [r]$, $\nu_s((A_{1,i}, A_{2,\pi_s^{-1}(i)})) = k$ if and only if 
$\sigma_s(k) = (A_{1,i}, A_{2,\pi_s^{-1}(i)})$. 
The following relation motivates our investigation of the cutwidth problem 
for permutation multigraphs in the remaining part of this section. 

\begin{lemma}
\label{lem:ib-cutwidth}
Let $s$ be a synchronous rule with $r > 2$ linked nonterminals, and let
$\sigma_s$ be a linear parsing strategy for $s$.  
Let $\pi_s$ and $G_s$ be the permutation and the permutation multigraph,
respectively, associated with $s$, and let $\nu_s$ be the linear arrangement
for $G_s$ associated with $\sigma_s$.  For every $i \in [r]$ we have 
\begin{align*}
\wg(G_s, \nu_s, i) = \ibd(\pi_s, \sigma_s, i) \;.
\end{align*}
\end{lemma}

\begin{proof}
The lemma follows from the definition of the internal boundary function
in~(\ref{eq:ibd}) and the definition of the permutation multigraph.
The first two terms in~(\ref{eq:ibd}) count edges from the set $E_{s,A}$
crossing the gap at position $i$ in linear arrangement 
$\nu_s$ of $G_s$, while the second two terms in~(\ref{eq:ibd})
count edges from the set $E_{s,B}$.
\end{proof}

Note that Lemma~\ref{lem:ib-cutwidth} directly implies the relation
$\cw(G_s) = \min_{\sigma} \max_{i \in [r]} \ibd(\pi_s, \sigma, k)$. 

In the rest of the present section we investigate the 
\termdef{permutation multigraph cutwidth} problem, or \mrgbcw\ for short.
An instance of \mrgbcw\ consists of a permutation multigraph $G$ 
and an integer $k$, and we have to decide whether $\cw(G) \leq k$. 
We show that the \mrgbcw\ problem is NP-complete. We reduce from 
the minimum bisection width problem, or \mcess\ for short. The \mcess\ problem 
consists of deciding whether, given a graph $G$ and an integer $k$, 
there is a partition of the nodes of $G$ into two equal size subsets $V_1$ and $V_2$, 
such that the number of edges with one endpoint in $V_1$ and one endpoint in $V_2$ is not greater than $k$.  It is known that the \mcess\ problem is NP-complete 
even when restricted to cubic graphs, that is, 3-regular graphs, 
with no multi-edges and no self-loops~\cite{Bui:1987}.  
We use this variant of the \mcess\ problem in our reduction. 
Our proof that the \mrgbcw\ problem is NP-complete is a modification of the proof reported in~\cite[Theorem 4.1, p. 434]{Makedon:1985p2188}, showing that the problem of deciding whether an undirected graph has (modified) cutwidth not greater than a given integer is NP-complete for graphs with maximum vertex degree of three.

\subsection{Construction of Permutation Multigraph $G'$}
\label{ssec:gprime}

Throughout the rest of this section, we let 
$G=(V,E)$ be a cubic graph where $V=\{v_1, \ldots, v_n\}$ is the set of its vertices.
Note that $n > 1$ must be even. 
We also let $k$ be an arbitrary positive integer. We construct a permutation multigraph $G'$ and an integer $k'$ such that $\langle G,k\rangle$ is a positive instance of \mcess\ if and only if $\langle G',k'\rangle$ is a positive instance of \mrgbcw\ (this statement will be shown in Sections~\ref{ssec:g2gprime} and~\ref{ssec:gprime2g}).

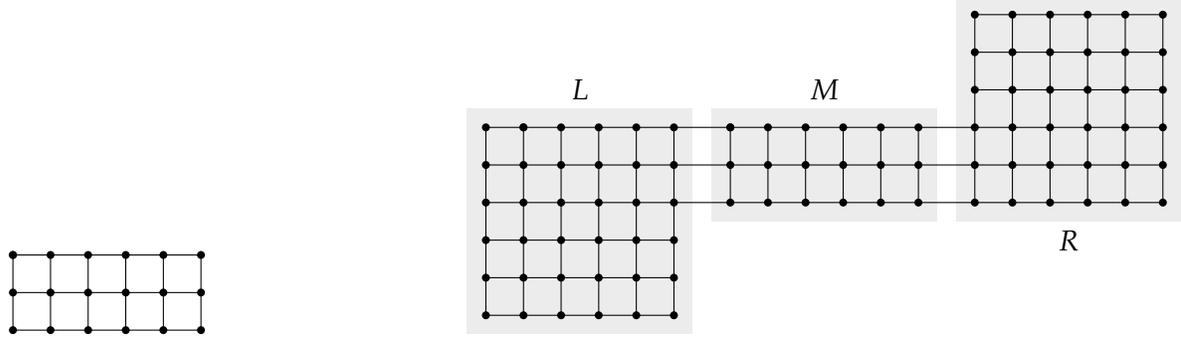
\begin{figure}[t]
\begin{center}
\qquad \begin{tikzpicture}[scale=1]
\draw (1.25,2.5)-- (1.25,3.5);
\draw (1.75,3.5)-- (1.75,2.5);
\draw (2.25,3.5)-- (2.25,2.5);
\draw (2.75,3.5)-- (2.75,2.5);
\draw (3.25,3.5)-- (3.25,2.5);
\draw (3.75,3.5)-- (3.75,2.5);
\draw (1.25,3.5)-- (3.75,3.5);
\draw (3.75,2.5)-- (1.25,2.5);
\draw (1.25,3)-- (3.75,3);
\fill (1.25,3.5) circle (1.5pt);
\fill (1.25,3) circle (1.5pt);
\fill (1.75,3) circle (1.5pt);
\fill (1.75,3.5) circle (1.5pt);
\fill (2.25,3.5) circle (1.5pt);
\fill (2.25,3) circle (1.5pt);
\fill (2.75,3.5) circle (1.5pt);
\fill (2.75,3) circle (1.5pt);
\fill (3.25,3.5) circle (1.5pt);
\fill (3.25,3) circle (1.5pt);
\fill (3.75,3.5) circle (1.5pt);
\fill (3.75,3) circle (1.5pt);
\fill (1.25,2.5) circle (1.5pt);
\fill (1.75,2.5) circle (1.5pt);
\fill (2.25,2.5) circle (1.5pt);
\fill (2.75,2.5) circle (1.5pt);
\fill (3.25,2.5) circle (1.5pt);
\fill (3.75,2.5) circle (1.5pt);
\fill (1.25,3.5) circle (1.5pt);
\end{tikzpicture} \qquad \hfill \begin{tiny}
\begin{tikzpicture}[scale=1]
\definecolor{zzttqq}{rgb}{0.25,0.25,0.25}
\fill[color=zzttqq,fill=zzttqq,fill opacity=0.1] (-2.25,3.75) -- (0.75,3.75) -- (0.75,0.75) -- (-2.25,0.75) -- cycle;
\fill[color=zzttqq,fill=zzttqq,fill opacity=0.1] (4.25,5.25) -- (4.25,2.25) -- (7.25,2.25) -- (7.25,5.25) -- cycle;
\fill[color=zzttqq,fill=zzttqq,fill opacity=0.1] (1,3.75) -- (4,3.75) -- (4,2.25) -- (1,2.25) -- cycle;

\draw (-2,3.5)-- (7,3.5);

\draw (-2,3)-- (7,3);
\draw (-2,2.5)-- (7,2.5);
\draw (-2,2)-- (0.5,2);
\draw (-2,1.5)-- (0.5,1.5);
\draw (-2,1)-- (0.5,1);
\draw (4.5,5)-- (7,5);
\draw (7,4.5)-- (4.5,4.5);
\draw (4.5,4)-- (7,4);
\draw (-2,3.5)-- (-2,1);
\draw (-1.5,3.5)-- (-1.5,1);
\draw (-1,1)-- (-1,3.5);
\draw (-0.5,3.5)-- (-0.5,1);
\draw (0,1)-- (0,3.5);
\draw (0.5,1)-- (0.5,3.5);
\draw (1.25,2.5)-- (1.25,3.5);
\draw (1.75,3.5)-- (1.75,2.5);
\draw (2.25,3.5)-- (2.25,2.5);
\draw (2.75,3.5)-- (2.75,2.5);
\draw (3.25,3.5)-- (3.25,2.5);
\draw (3.75,3.5)-- (3.75,2.5);
\draw (4.5,5)-- (4.5,2.5);
\draw (5,2.5)-- (5,5);
\draw (5.5,5)-- (5.5,2.5);
\draw (6,2.5)-- (6,5);
\draw (6.5,5)-- (6.5,2.5);
\draw (7,2.5)-- (7,5);
\fill (-2,3.5) circle (1.5pt);
\fill (-2,3) circle (1.5pt);
\fill (-2,2.5) circle (1.5pt);
\fill (-2,2) circle (1.5pt);
\fill (-1.5,3.5) circle (1.5pt);
\fill (-1,3.5) circle (1.5pt);
\fill (-0.5,3.5) circle (1.5pt);
\fill (0,3.5) circle (1.5pt);
\fill (0.5,3.5) circle (1.5pt);
\fill (-1.5,3) circle (1.5pt);
\fill (-1,3) circle (1.5pt);
\fill (-0.5,3) circle (1.5pt);
\fill (0,3) circle (1.5pt);
\fill (0.5,3) circle (1.5pt);
\fill (-1.5,2.5) circle (1.5pt);
\fill (-1,2.5) circle (1.5pt);
\fill (-0.5,2.5) circle (1.5pt);
\fill (0,2.5) circle (1.5pt);
\fill (0.5,2.5) circle (1.5pt);
\fill (-1.5,2) circle (1.5pt);
\fill (-1,2) circle (1.5pt);
\fill (-0.5,2) circle (1.5pt);
\fill (0,2) circle (1.5pt);
\fill (0.5,2) circle (1.5pt);
\fill (1.25,3.5) circle (1.5pt);
\fill (1.25,3) circle (1.5pt);
\fill (1.75,3) circle (1.5pt);
\fill (1.75,3.5) circle (1.5pt);
\fill (2.25,3.5) circle (1.5pt);
\fill (2.25,3) circle (1.5pt);
\fill (2.75,3.5) circle (1.5pt);
\fill (2.75,3) circle (1.5pt);
\fill (3.25,3.5) circle (1.5pt);
\fill (3.25,3) circle (1.5pt);
\fill (3.75,3.5) circle (1.5pt);
\fill (3.75,3) circle (1.5pt);
\fill (4.5,3) circle (1.5pt);
\fill (4.5,3.5) circle (1.5pt);
\fill (4.5,4) circle (1.5pt);
\fill (4.5,4.5) circle (1.5pt);
\fill (5,4.5) circle (1.5pt);
\fill (5.5,4.5) circle (1.5pt);
\fill (6,4.5) circle (1.5pt);
\fill (6.5,4.5) circle (1.5pt);
\fill (7,4.5) circle (1.5pt);
\fill (7,4) circle (1.5pt);
\fill (7,3.5) circle (1.5pt);
\fill (7,3) circle (1.5pt);
\fill (6.5,3) circle (1.5pt);
\fill (6,3) circle (1.5pt);
\fill (5.5,3) circle (1.5pt);
\fill (5,3) circle (1.5pt);
\fill (5,3.5) circle (1.5pt);
\fill (5.5,3.5) circle (1.5pt);
\fill (6,3.5) circle (1.5pt);
\fill (6.5,3.5) circle (1.5pt);
\fill (6.5,4) circle (1.5pt);
\fill (6,4) circle (1.5pt);
\fill (5.5,4) circle (1.5pt);
\fill (5,4) circle (1.5pt);
\fill (1.25,2.5) circle (1.5pt);
\fill (1.75,2.5) circle (1.5pt);
\fill (2.25,2.5) circle (1.5pt);
\fill (2.75,2.5) circle (1.5pt);
\fill (3.25,2.5) circle (1.5pt);
\fill (3.75,2.5) circle (1.5pt);
\fill (4.5,2.5) circle (1.5pt);
\fill (5,2.5) circle (1.5pt);
\fill (5.5,2.5) circle (1.5pt);
\fill (6,2.5) circle (1.5pt);
\fill (6.5,2.5) circle (1.5pt);
\fill (7,2.5) circle (1.5pt);
\fill (-2,1.5) circle (1.5pt);
\fill (-1.5,1.5) circle (1.5pt);
\fill (-1,1.5) circle (1.5pt);
\fill (-0.5,1.5) circle (1.5pt);
\fill (0,1.5) circle (1.5pt);
\fill (0.5,1.5) circle (1.5pt);
\fill (-2,1) circle (1.5pt);
\fill (-1.5,1) circle (1.5pt);
\fill (-1,1) circle (1.5pt);
\fill (-0.5,1) circle (1.5pt);
\fill (0,1) circle (1.5pt);
\fill (0.5,1) circle (1.5pt);
\fill (4.5,5) circle (1.5pt);
\fill (5,5) circle (1.5pt);
\fill (5.5,5) circle (1.5pt);
\fill (6,5) circle (1.5pt);
\fill (6.5,5) circle (1.5pt);
\fill (7,5) circle (1.5pt);
\fill (-1.5,3.5) circle (1.5pt);
\fill (1.25,3.5) circle (1.5pt);
\fill (-0.75,4) node {\normalsize $L$};
\fill (5.75,2) node {\normalsize $R$};
\fill (2.5,4) node {\normalsize $M$};
\end{tikzpicture}
\end{tiny} \qquad
\end{center}
\caption{The $\grid[3,6]$ grid (left), whose cutwidth is 4, and the composed grid $\cgrid[6,6,3,6]$ (right) with grids $L$, $M$ and $R$ shaded.}
\label{fig:nodegadget}
\end{figure}

Let $\ggh$ and $\ggw$ be positive integers.  We will make use of a grid gadget $X = \grid[\ggh,\ggw]$ with $\ggh$ rows and $\ggw$ columns; for an example, see the left part of Figure~\ref{fig:nodegadget}. More precisely, for any $h \in [\ggh]$ and for any $w \in [\ggw]$, the grid $X$ includes a node $x^{h,w}$. Moreover, for any $h \in [\ggh]$ and for any $w \in [\ggw-1]$, there is an edge $\left(x^{h,w},x^{h,w+1}\right)$, and, for any $h \in [\ggh-1]$ and for any $w \in [\ggw]$, there is an edge $\left(x^{h,w},x^{h+1,w}\right)$. It is known that, for any $\ggh$ and $\ggw$ greater than 2, $\cw(X) = \min\{\ggh+1, \ggw+1\}$~\cite{DBLP:conf/wg/RolimSV95}. 

For positive integers $\lggh$, $\lggw$, $\mggh$, and $\mggw$ with $\lggh > \mggh$, we will also exploit a \termdef{composed grid} $\cgrid[\lggh,\lggw,\mggh,\mggw]$, which is formed by combining two grid gadgets $L = \grid[\lggh,\lggw]$ and $R = \grid[\lggh,\lggw]$ with a grid gadget $M = \grid[\mggh,\mggw]$, as shown in the right part of Figure~\ref{fig:nodegadget}. The nodes of $L$, $R$, and $M$ will be denoted as $l^{h,w}$, $r^{h,w}$, and $m^{h,w}$, respectively. Besides the edges of the three grids $L$, $R$, and $M$, for any $h \in [\mggh]$, the composed grid $\cgrid[\lggh,\lggw,\mggh,\mggw]$ also includes the edges $\left(l^{h,\lggw}, m^{h,1}\right)$ and $\left(m^{h,\mggw}, r^{\lggh-\mggh+h,1}\right)$.

The target graph $G'$ consists of several grid gadgets. More specifically, it has one grid $G_i=\grid[\ngh,\ngw]$, $i \in [n]$, for each of the $n$ nodes of the source cubic graph $G$. The nodes of $G_i$ will be denoted as $g_i^{h,w}$. In addition, $G'$ has a composed grid $S = \cgrid[\lgh,\lgw,\mgh,\mgw]$.  
For each grid $G_i$, $i \in [n]$, we add to $G'$ a sheaf of $\nee$ edges
connecting distinct nodes in the first or in the last row of
$G_i$ to $\nee$ distinct nodes of the first row of $M$, as 
will be explained in detail below. In addition, for each edge
$(v_i,v_j) \in E$ with $i<j$, we add to $G'$ two edges, each edge connecting a node in
the first row of $G_i$ with a node in the last row of $G_j$. The
choice of all of the above connections will be done in a way that
guarantees that $G'$ is a permutation multigraph.

Before providing a mathematical specification of $G'$, we informally summarize the organization of the edges of $G'$ connecting the $M$ and $G_i$ grids.
When scanning the columns of $M$ from left to right, we will have a first column that has no connection to any of the $G_i$ components, followed by a first block of $\nee$ 
columns with connections to the $G_1$ component, 
followed by a second block of $\nee$ columns with connections to $G_2$, and so on
up to the $n$-th block of $\nee$ columns with connections to $G_n$.
The remaining columns of $M$ do not have any
connection with the $G_i$ components.

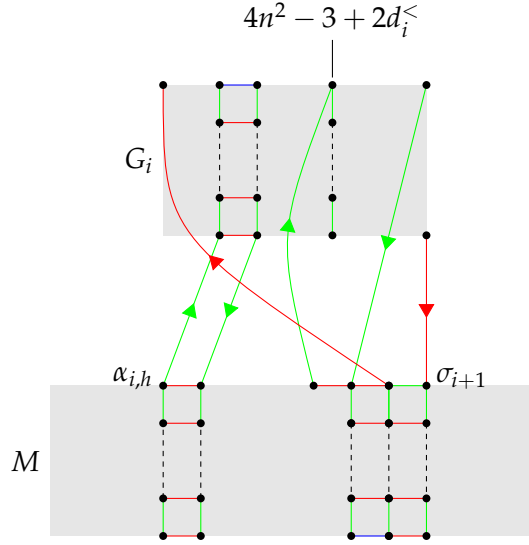
\begin{figure}[t]
\begin{center}
\begin{tikzpicture}[scale=1]
\fill[dotted,fill=black,fill opacity=0.1] (-2,2.25) -- (1.5,2.25) -- (1.5,0.25) -- (-2,0.25) -- cycle;
\fill[dotted,fill=black,fill opacity=0.1] (-3.5,-1.75) -- (3,-1.75) -- (3,-3.75) -- (-3.5,-3.75) -- cycle;
\draw (-2,1.25) node[left] {$G_i$};
\draw (-3.5,-2.75) node[left] {$M$};
\draw[color=green] (-2,-1.75)-- (-2,-2.25);
\draw[color=green] (-2,-3.25)-- (-2,-3.75);
\draw[color=green] (-1.5,-1.75)-- (-1.5,-2.25);
\draw[color=green] (-1.5,-3.25)-- (-1.5,-3.75);
\draw[color=green] (-0.75,2.25)-- (-0.75,1.75);	
\draw[color=green] (-0.75,0.75)-- (-0.75,0.25);
\draw[color=green] (-1.25,2.25)-- (-1.25,1.75);
\draw[color=green] (-1.25,0.75)-- (-1.25,0.25);
\draw[color=blue] (-1.25,2.25)-- (-0.75,2.25);	
\draw[color=red] (-1.25,1.75)-- (-0.75,1.75);
\draw[color=red] (-1.25,0.75)-- (-0.75,0.75);
\draw[color=red] (-1.25,0.25)-- (-0.75,0.25);
\draw[color=red] (-2,-1.75)-- (-1.5,-1.75);	
\draw[color=red] (-2,-2.25)-- (-1.5,-2.25);
\draw[color=red] (-2,-3.25)-- (-1.5,-3.25);
\draw[color=red] (-1.5,-3.75)-- (-2,-3.75);
\draw [dash pattern=on 2pt off 2pt] (-2,-3.25)-- (-2,-2.25);
\draw [dash pattern=on 2pt off 2pt] (-1.5,-3.25)-- (-1.5,-2.25);
\draw [dash pattern=on 2pt off 2pt] (-0.75,1.75)-- (-0.75,0.75);
\draw [dash pattern=on 2pt off 2pt] (-1.25,1.75)-- (-1.25,0.75);
\tikzset{
red/.style={draw=red, postaction={decorate},
decoration={markings,mark=at position .55 with {\arrow[draw=red,fill=red]{triangle 60}}}},
green/.style={draw=green, postaction={decorate}, decoration={markings,mark=at position .55 with {\arrow[draw=green,fill=green]{triangle 60}}}},
blue/.style={draw=blue, postaction={decorate},
decoration={markings,mark=at position .75 with {\arrow[draw=green,fill=green]{triangle 60}},mark=at position .25 with {\arrow[draw=red,fill=red]{triangle 60 reversed}}}}
}
\draw[green] (-2,-1.75)-- (-1.25,0.25);		
\draw[color=green] (0.5,-2.25)-- (0.5,-1.75);	
\draw[color=green] (1,-2.25)-- (1,-1.75);	
\draw[color=red] (0.5,-2.25)-- (1,-2.25);	
\draw[color=green] (1,-3.75)-- (1,-3.25);	
\draw[color=green] (0.5,-3.75)-- (0.5,-3.25);	
\draw[color=red] (0.5,-3.25)-- (1,-3.25);	
\draw[green] (-0.75,0.25)-- (-1.5,-1.75);	
\draw[color=green] (0.25,2.25)-- (0.25,1.75);
\draw[color=green] (0.25,0.75)-- (0.25,0.25);
\draw [dash pattern=on 2pt off 2pt] (0.5,-2.25)-- (0.5,-3.25);
\draw [dash pattern=on 2pt off 2pt] (1,-2.25)-- (1,-3.25);
\draw [dash pattern=on 2pt off 2pt] (1.5,-2.25)-- (1.5,-3.25);
\draw [dash pattern=on 2pt off 2pt] (0.25,1.75)-- (0.25,0.75);
\draw[green] (0.0,-1.75).. controls (-0.5,0.25) .. (0.25,2.25); 
\draw[color=red]  (0.0,-1.75)-- (0.5,-1.75);
\draw[color=red]  (0.5,-1.75)-- (1,-1.75);
\draw[color=green]  (1,-1.75)-- (1.5,-1.75);
\draw[green] (1.5,2.25)-- (0.5,-1.75); 
\draw[red] (1,-1.75) .. controls (-2,0.25) .. (-2,2.25); 
\draw[red] (1.5,0.25) -- (1.5,-1.75); 
\draw[color=green] (1,-1.75)-- (1,-2.25);
\draw[color=green] (1.5,-1.75)-- (1.5,-2.25);
\draw[color=red]  (1.5,-2.25)-- (1,-2.25);
\draw[color=red]  (1,-3.25)-- (1.5,-3.25);
\draw[color=red]  (1,-3.75)-- (1.5,-3.75);
\draw[color=blue]  (1,-3.75)-- (0.5,-3.75);
\draw[color=green] (1.5,-3.75)-- (1.5,-3.25);
\fill (-2,-1.75) circle (1.5pt);
\fill (-2,2.25) circle (1.5pt);
\fill (0.0,-1.75) circle (1.5pt);
\fill (-1.25,2.25) circle (1.5pt);
\fill (-1.25,1.75) circle (1.5pt);
\fill (-1.25,0.75) circle (1.5pt);
\fill (-1.25,0.25) circle (1.5pt);
\fill (-0.75,2.25) circle (1.5pt);
\fill (-1.5,-1.75) circle (1.5pt);
\fill (-0.75,0.75) circle (1.5pt);
\fill (-0.75,0.25) circle (1.5pt);
\fill (-0.75,1.75) circle (1.5pt);
\fill (-2,-2.25) circle (1.5pt);
\fill (-1.5,-2.25) circle (1.5pt);
\fill (-2,-3.25) circle (1.5pt);
\fill (-2,-3.75) circle (1.5pt);
\fill (-1.5,-3.25) circle (1.5pt);
\fill (-1.5,-3.75) circle (1.5pt);
\fill (1.5,2.25) circle (1.5pt);
\fill (1.5,0.25) circle (1.5pt);
\fill (0.5,-1.75) circle (1.5pt);
\fill (0.5,-2.25) circle (1.5pt);
\fill (1,-1.75) circle (1.5pt);
\fill (1,-2.25) circle (1.5pt);
\fill (0.5,-3.25) circle (1.5pt);
\fill (0.5,-3.75) circle (1.5pt);
\fill (1,-3.75) circle (1.5pt);
\fill (0.5,-2.25) circle (1.5pt);
\fill (1,-3.25) circle (1.5pt);
\fill (0.25,0.25) circle (1.5pt);
\fill (0.25,0.75) circle (1.5pt);
\fill (0.25,1.75) circle (1.5pt);
\fill (0.25,2.25) circle (1.5pt);
\fill (1,-1.75) circle (1.5pt);
\fill (1.5,-1.75) circle (1.5pt);
\fill (1,-2.25) circle (1.5pt);
\fill (1.5,-2.25) circle (1.5pt);
\fill (1,-3.25) circle (1.5pt);
\fill (1.5,-3.25) circle (1.5pt);
\fill (1.5,-3.75) circle (1.5pt);
\fill (1,-3.75) circle (1.5pt);

\draw (0.25,2.35)-- (0.25,2.85);	
\draw (0.25,2.75) node[above] {$\nee-3+2d^<_i$};
\draw (1.5,-1.65) node[right] {$\sigma_{i+1}$};
\draw (-2,-1.65) node[left] {$\alpha_{i,h}$};
\end{tikzpicture}
\end{center}
\caption{Block 1 is composed by all columns of $G_i$ connected with $M$ as in the  pattern displayed by the two green columns at the left of position $\nee-3+2d^<_i$.  Block 3 is composed by all columns of $G_i$ starting with the column at position $\nee-3+2d^<_i$ and extending to the right.}
\label{fig:extraedgeconnections1}
\end{figure}

Looking at one of the grids $G_i$, 
the columns are organized into three blocks, when scanning from left to right. 
\begin{itemize}

\item 
{\bf Block 1} (left portion of $G_i$ in Figure \ref{fig:extraedgeconnections1}): 
The first column has two edges connecting to the $M$ component, one 
from its top vertex and one from its bottom vertex, and the 
remaining columns in the block each have a single edge connecting to $M$ 
from the column's bottom vertex.  
This block extends from column with index 1 to column  
$\nee -4 -2d_i^{>}$, where $d_i^{>}$ denotes the number of edges of $G$
of the form $(v_i, v_j)$ with $j>i$, that is, the number of ``forward'' 
neighbors of $i$.  The block therefore contains a total of 
$\nee -3 -2d_i^{>}$ edges connecting $G_i$ to $M$. 

\item 
{\bf Block 2} (Figure \ref{fig:extraedgeconnections2}): 
This block represents the edges from the source graph $G$.
For each edge $(v_i, v_j)$ of $G$ such that $i<j$,
we have two columns each having a single edge connecting to $M$
from the column's bottom vertex, and a single edge connecting to 
the grid $G_j$ from the column's top vertex. 
This is followed by two columns for each edge $(v_i, v_j)$ in $G$ such that $j<i$,
with each column having its bottom vertex connected to the
grid $G_j$ and no connection to $M$.
Block 2 extends from column $\nee -3 -2d_i^{>}$ to column  
$\nee -4 +2d_i^{<}$, where $d_i^{<}$ denotes the number of edges of $G$
of the form $(v_i, v_j)$ with $j<i$, that is, the number of 
``backward'' neighbors of $i$. 
The number of edges connecting $G_i$ to $M$ is $2d_i^{>}$, and the number of edges connecting $G_i$ to grids $G_j$ with $j \neq i$ is $2d_i^{>} + 2d_i^{<} = 6$, where the equality follows from the fact that $G$ is a cubic graph. 

\item 
{\bf Block 3} (right portion of $G_i$ in Figure \ref{fig:extraedgeconnections1}):
The first column 
has the top vertex connected to $M$.  All remaining columns of $G_i$ do not have connections with $M$, with the exception of the rightmost column, which has two edges connecting its top and bottom vertices to $M$. 
This block extends from column with index $\nee -3 +2d_i^{<}$ to column with index 
$\ngw$, and the block contains $3$ edges connecting $G_i$ to $M$. 

\end{itemize}
Altogether, the above blocks provide a total number of edges connecting $G_i$ and $M$ equal to $(\nee -3 -2d_i^{>}) + 2d_i^{>} + 3 = \nee$, as anticipated.  

To define the edges in each of these three blocks, 
we need to introduce some auxiliary notation.   In what follows, for every $i \in [n]$, we denote by $N^<_{i}$ the set of ``backward'' neighbors of $v_i$, that is, $N^<_{i} = \{j : (v_i,v_j) \in E \wedge j \in [i-1] \}$. 
Similarly, the set of ``forward'' neighbors of $v_i$ is the set 
$N^>_{i} = \{j : (v_i,v_j) \in E \wedge j \in [n] \setminus [i] \}$. 
We then have $d^<_i = \size{N^<_{i}}$ and $d^>_i = \size{N^>_{i}}$.
As already observed, we must have $d^>_i+d^<_i=3$ since $G$ is a cubic graph. 
For every $i \in [n]$ and $j \in [i]$ we also define $d^{<,j}_{i} = \size{N^<_{i} \cap [j-1]}$; in words, $d^{<,j}_{i}$ is the number of backward neighbors of $v_i$ having index strictly smaller than $j$. Note that $d^{<,i}_{i} = d^<_{i}$ for every $i \in [n]$.

For every $i \in [n]$ and $h \in [d^>_i]$, we denote by $n^h_i$ the index of the $h$-th forward neighbor of $v_i$ from left-to-right: formally, $n^h_i$ is the value $j$ such that $j \in N^>_{i}$ and $\size{[j-1] \cap N^>_{i}} = h-1$. Finally, for every $i \in [n+1]$, we define $\sigma_i =  \nee(i-1)+1$.
This quantity will be used in the construction of $G'$ below as an offset when locating the index of the next available column, from left to right, in the $M$ component.  For instance, we have $\sigma_i = 1$ since in $M$ the first block with $\nee$ connections to $G_1$ starts at column 2, as already anticipated.  

We are now ready to specify precisely each of the three blocks 
of edges connecting each $G_i$ to the other grids.

\paragraph{Block 1} 
For any $h \in \left[\halfnee-2-d^{>}_i\right]$, $G'$ includes the edges (see Figure~\ref{fig:extraedgeconnections1})
\[
\left(m^{1,\alpha_{i,h}}, g_i^{\ngh,2h-1}\right),
\left(g_i^{\ngh,2h},m^{1,\alpha_{i,h}+1}\right),
\]
where $\alpha_{i,h} = \sigma_i+(2h-1)$.  
In addition to the above edges, there is one edge connecting node $g_i^{1,1}$ with $M$ that is associated with Block 1.  However, in order to simplify the presentation, it is more convenient to list such edge under Block 3 below.  

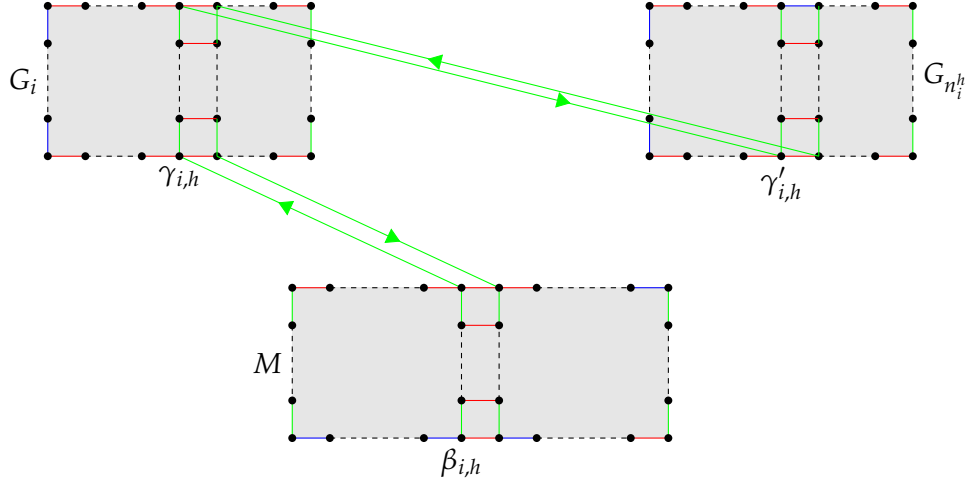
\begin{figure}[t]
\begin{center}
\definecolor{qqqqff}{rgb}{0,0,1}
\begin{tikzpicture}[scale=1]
\fill[dotted,fill=black,fill opacity=0.1] (-2,2.25) -- (1.5,2.25) -- (1.5,0.25) -- (-2,0.25) -- cycle;
\fill[dotted,fill=black,fill opacity=0.1] (6,2.25) -- (9.5,2.25) -- (9.5,0.25) -- (6,0.25) -- cycle;
\fill[dotted,fill=black,fill opacity=0.1] (1.25,-1.5) -- (6.25,-1.5) -- (6.25,-3.5) -- (1.25,-3.5) -- cycle;
\fill (-2,1.25) node[left] {$G_i$};
\fill (9.5,1.25) node[right] {$G_{n^h_i}$};
\fill (1.25,-2.5) node[left] {$M$};
\draw[color=green] (3.5,-1.5)-- (3.5,-2);
\draw[color=green] (3.5,-3)-- (3.5,-3.5);
\draw[color=green] (4,-1.5)-- (4,-2);
\draw[color=green] (4,-3)-- (4,-3.5);
\draw[color=green] (-0.25,2.25)-- (-0.25,1.75);
\draw[color=green] (-0.25,0.75)-- (-0.25,0.25);
\draw[color=green] (7.75,2.25)-- (7.75,1.75);
\draw[color=green] (7.75,0.75)-- (7.75,0.25);
\draw[color=red] (3.5,-1.5)-- (4,-1.5);
\draw[color=red] (3.5,-2)-- (4,-2);
\draw[color=red] (3.5,-3)-- (4,-3);
\draw[color=red] (4,-3.5)-- (3.5,-3.5);
\draw [dash pattern=on 2pt off 2pt] (3.5,-3)-- (3.5,-2);
\draw [dash pattern=on 2pt off 2pt] (4,-3)-- (4,-2);
\draw [dash pattern=on 2pt off 2pt] (-0.25,1.75)-- (-0.25,0.75);
\draw [dash pattern=on 2pt off 2pt] (7.75,1.75)-- (7.75,0.75);
\draw [dash pattern=on 2pt off 2pt] (1.25,-2)-- (1.25,-3);
\draw[color=green] (1.25,-1.5)-- (1.25,-2);
\fill (1.25,-2) circle (1.5pt);
\draw[color=green] (1.25,-3)-- (1.25,-3.5);
\fill (1.25,-3) circle (1.5pt);

\draw [dash pattern=on 2pt off 2pt] (6.25,-2)-- (6.25,-3);
\draw[color=green] (6.25,-1.5)-- (6.25,-2);
\fill (6.25,-2) circle (1.5pt);
\draw[color=green] (6.25,-3)-- (6.25,-3.5);
\fill (6.25,-3) circle (1.5pt);

\draw [dash pattern=on 2pt off 2pt] (1.75,-1.5)-- (3,-1.5);
\draw[color=red] (1.25,-1.5)-- (1.75,-1.5);
\fill (1.75,-1.5) circle (1.5pt);
\draw[color=red] (3,-1.5)-- (3.5,-1.5);
\fill (3,-1.5) circle (1.5pt);

\draw [dash pattern=on 2pt off 2pt] (4.5,-1.5)-- (5.75,-1.5);
\draw[color=red] (4,-1.5)-- (4.5,-1.5);
\fill (4.5,-1.5) circle (1.5pt);
\draw[color=blue] (5.75,-1.5)-- (6.25,-1.5);
\fill (5.75,-1.5) circle (1.5pt);

\draw [dash pattern=on 2pt off 2pt] (1.75,-3.5)-- (3,-3.5);
\draw[color=blue] (1.25,-3.5)-- (1.75,-3.5);
\fill (1.75,-3.5) circle (1.5pt);
\draw[color=blue] (3,-3.5)-- (3.5,-3.5);

\draw (3.5,-3.5) node[below] {$\beta_{i,h}$};

\fill (3,-3.5) circle (1.5pt);

\draw [dash pattern=on 2pt off 2pt] (4.5,-3.5)-- (5.75,-3.5);
\draw[color=blue] (4,-3.5)-- (4.5,-3.5);
\fill (4.5,-3.5) circle (1.5pt);
\draw[color=red] (5.75,-3.5)-- (6.25,-3.5);
\fill (5.75,-3.5) circle (1.5pt);
\draw [dash pattern=on 2pt off 2pt] (-2,1.75)-- (-2,0.75);
\draw[color=blue] (-2,2.25)-- (-2,1.75);
\fill (-2,1.75) circle (1.5pt);
\draw[color=blue] (-2,0.75)-- (-2,0.25);
\fill (-2,0.75) circle (1.5pt);

\draw [dash pattern=on 2pt off 2pt] (1.5,1.75)-- (1.5,0.75);
\draw[color=green] (1.5,2.25)-- (1.5,1.75);
\fill (1.5,1.75) circle (1.5pt);
\draw[color=green] (1.5,0.75)-- (1.5,0.25);
\fill (1.5,0.75) circle (1.5pt);

\draw [dash pattern=on 2pt off 2pt] (-1.5,2.25)-- (-0.75,2.25);
\draw[color=red] (-2,2.25)-- (-1.5,2.25);
\fill (-1.5,2.25) circle (1.5pt);
\draw[color=red] (-0.75,2.25)-- (-0.25,2.25);
\fill (-0.75,2.25) circle (1.5pt);

\draw [dash pattern=on 2pt off 2pt] (0.25,2.25)-- (1,2.25);
\draw[color=red] (-0.25,2.25)-- (0.25,2.25);
\fill (0.25,2.25) circle (1.5pt);
\draw[color=red] (1,2.25)-- (1.5,2.25);
\fill (1,2.25) circle (1.5pt);

\draw [dash pattern=on 2pt off 2pt] (-1.5,0.25)-- (-0.75,0.25);
\draw[color=red] (-2,0.25)-- (-1.5,0.25);
\fill (-1.5,0.25) circle (1.5pt);
\draw[color=red] (-0.75,0.25)-- (-0.25,0.25);
\fill (-0.75,0.25) circle (1.5pt);

\draw [dash pattern=on 2pt off 2pt] (0.25,0.25)-- (1,0.25);
\draw[color=red] (-0.25,0.25)-- (0.25,0.25);
\fill (0.25,0.25) circle (1.5pt);
\draw[color=red] (1,0.25)-- (1.5,0.25);
\fill (1,0.25) circle (1.5pt);
\draw [dash pattern=on 2pt off 2pt] (6,1.75)-- (6,0.75);
\draw[color=blue] (6,2.25)-- (6,1.75);
\fill (6,1.75) circle (1.5pt);
\draw[color=blue] (6,0.75)-- (6,0.25);
\fill (6,0.75) circle (1.5pt);

\draw [dash pattern=on 2pt off 2pt] (9.5,1.75)-- (9.5,0.75);
\draw[color=green] (9.5,2.25)-- (9.5,1.75);
\fill (9.5,1.75) circle (1.5pt);
\draw[color=green] (9.5,0.75)-- (9.5,0.25);
\fill (9.5,0.75) circle (1.5pt);

\draw [dash pattern=on 2pt off 2pt] (6.5,2.25)-- (7.25,2.25);
\draw[color=red] (6,2.25)-- (6.5,2.25);
\fill (6.5,2.25) circle (1.5pt);
\draw[color=red] (7.25,2.25)-- (7.75,2.25);
\fill (7.25,2.25) circle (1.5pt);

\draw [dash pattern=on 2pt off 2pt] (8.25,2.25)-- (9,2.25);
\draw[color=blue] (7.75,2.25)-- (8.25,2.25);
\fill (8.25,2.25) circle (1.5pt);
\draw[color=red] (9,2.25)-- (9.5,2.25);
\fill (9,2.25) circle (1.5pt);

\draw [dash pattern=on 2pt off 2pt] (6.5,0.25)-- (7.25,0.25);
\draw[color=red] (6,0.25)-- (6.5,0.25);
\fill (6.5,0.25) circle (1.5pt);
\draw[color=red] (7.25,0.25)-- (7.75,0.25);
\fill (7.25,0.25) circle (1.5pt);

\draw [dash pattern=on 2pt off 2pt] (8.25,0.25)-- (9,0.25);
\draw[color=red] (7.75,0.25)-- (8.25,0.25);
\fill (8.25,0.25) circle (1.5pt);
\draw[color=red] (9,0.25)-- (9.5,0.25);
\fill (9,0.25) circle (1.5pt);
\tikzset{
    green/.style={draw=green, postaction={decorate},
        decoration={markings,mark=at position .65 with {\arrow[draw=green,fill=green]{triangle 60}}}}
}
\draw[green] (3.5,-1.5) -- (-0.25,0.25);	
\draw (-0.25,0.25) node[below] {$\gamma_{i,h}$};
\draw[green] (-0.25,2.25)-- (7.75,0.25);	
\draw[green] (8.25,0.25)-- (0.25,2.25);		
\draw[green] (0.25,0.25)-- (4,-1.5);		
\draw (7.75,0.25) node[below] {$\gamma'_{i,h}$};
\fill (3.5,-1.5) circle (1.5pt);
\fill (-0.25,0.25) circle (1.5pt);
\fill (-0.25,2.25) circle (1.5pt);
\fill (0.25,0.75) circle (1.5pt);
\draw[color=red] (-0.25,0.75)-- (0.25,0.75);
\draw[color=green] (0.25,0.25)-- (0.25,0.75);
\fill (0.25,1.75) circle (1.5pt);
\draw[color=red] (-0.25,1.75)-- (0.25,1.75);
\draw[color=green] (0.25,2.25)-- (0.25,1.75);
\draw [dash pattern=on 2pt off 2pt] (0.25,1.75)-- (0.25,0.75);
\fill (8.25,0.75) circle (1.5pt);
\draw[color=red] (7.75,0.75)-- (8.25,0.75);
\draw[color=green] (8.25,0.25)-- (8.25,0.75);
\fill (8.25,1.75) circle (1.5pt);
\draw[color=red] (7.75,1.75)-- (8.25,1.75);
\draw[color=green] (8.25,2.25)-- (8.25,1.75);
\draw [dash pattern=on 2pt off 2pt] (8.25,1.75)-- (8.25,0.75);
\fill (7.75,0.25) circle (1.5pt);
\fill (4,-1.5) circle (1.5pt);
\fill (7.75,2.25) circle (1.5pt);
\fill (7.75,1.75) circle (1.5pt);
\fill (-0.25,0.75) circle (1.5pt);
\fill (-0.25,1.75) circle (1.5pt);
\fill (7.75,0.75) circle (1.5pt);
\fill (3.5,-2) circle (1.5pt);
\fill (4,-2) circle (1.5pt);
\fill (3.5,-3) circle (1.5pt);
\fill (3.5,-3.5) circle (1.5pt);
\fill (4,-3) circle (1.5pt);
\fill (4,-3.5) circle (1.5pt);
\fill (-2,2.25) circle (1.5pt);
\fill (1.5,2.25) circle (1.5pt);
\fill (1.5,0.25) circle (1.5pt);
\fill (-2,0.25) circle (1.5pt);
\fill (6,2.25) circle (1.5pt);
\fill (9.5,2.25) circle (1.5pt);
\fill (9.5,0.25) circle (1.5pt);
\fill (6,0.25) circle (1.5pt);
\fill (1.25,-1.5) circle (1.5pt);
\fill (6.25,-1.5) circle (1.5pt);
\fill (6.25,-3.5) circle (1.5pt);
\fill (1.25,-3.5) circle (1.5pt);
\end{tikzpicture}
\end{center}
\caption{Block 2 is composed by all columns of $G_i$ representing connections between grids of $G'$ that correspond to edges of the source graph $G$.}
\label{fig:extraedgeconnections2}
\end{figure}

\paragraph{Block 2}
We now add to $G'$ the edges that are derived from the original graph $G$.
For every $i \in [n]$ and for every $h \in \left[d^>_i\right]$, that is, for any forward edge $(v_i, v_{n^h_i})$ in $G$, $G'$ includes the four edges (see Figure~\ref{fig:extraedgeconnections2})
\[
\left(m^{1,\beta_{i,h}}, g_i^{\ngh,\gamma_{i,h}}\right),
\left(g_i^{1,\gamma_{i,h}}, g_{n^h_i}^{\ngh,\gamma'_{i,h}}\right),
\left(g_{n^h_i}^{\ngh,\gamma'_{i,h}+1}, g_i^{1,\gamma_{i,h}+1}\right),
\left(g_i^{\ngh,\gamma_{i,h}+1}, m^{1,\beta_{i,h}+1}\right),
\]
where $\beta_{i,h} = \sigma_i + \nee - 4 -2d^>_i + (2h-1)$,
$\gamma_{i,h} = \nee - 4 -2d^>_i + (2h-1)$,
and $\gamma'_{i,h} = \nee - 3 + 2d^{<,i}_{n^h_i}$. 
Observe that $\beta_{i,1} = \alpha_{i,\halfnee-2-d^{>}_i} + 1$, so that 
the two runs of edges defined by Block 1 and Block 2, connecting grid $G_i$ to grid $M$, 
are one next to the other. 
We have thus added to $G'$, for any edge $(v_i,v_j)$ of $G$, two edges 
connecting the two grids $G_i$ and $G_j$, and two edges connecting 
grid $G_i$ to grid $M$.   

Combining Block 1 and Block 2 together, we have added 
a total of $\nee-4$ edges from grid $G_i$ to grid $M$, for every $i \in [n]$.  

\paragraph{Block 3}
Finally, $G'$ includes the four edges (see Figure~\ref{fig:extraedgeconnections1})
\[
\left(m^{1,\sigma_{i+1} - 3}, g_i^{1,\nee-3+2d^<_i}\right),
\left(g_i^{1,\ngw},m^{1,\sigma_{i+1}-2}\right),
\left(m^{1,\sigma^{\phantom{4}}_{i+1}-1},g_i^{1,1}\right),
\left(g_i^{\ngh,\ngw},m^{1,\sigma_{i+1}}\right) \;.
\]

Combining all three blocks, we have a total of $\nee$ edges
connecting $G_i$ to $M$.  We further note that each vertex of
$G_i$ has at most one edge connecting to vertices outside $G_i$;
this property will play an important role later in our proofs.

So far we have specified $G'$ as if it were a (plain) graph;  however, $G'$ is a permutation multigraph.   The coloring of the edges of $G'$, that is, the definition of the
two edge sets $A$ and $B$, is specified below by describing the Hamiltonian path of red edges and the Hamiltonian path of green edges. Some of the edges specified above are included in both the red and green paths; these are double edges in the multigraph $G'$.

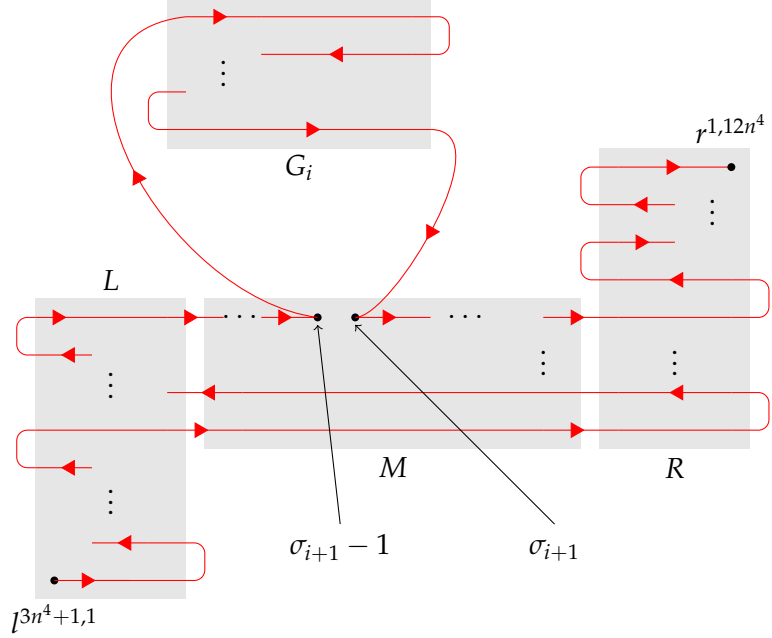
\begin{figure}[t]
\begin{center}
\definecolor{qqqqff}{rgb}{0,0,1}
\begin{tikzpicture}[scale=1.0]



\tikzset{
red/.style={draw=red, postaction={decorate},
decoration={markings,mark=at position .55 with {\arrow[draw=red,fill=red]{triangle 60}}}},
blue/.style={draw=blue, postaction={decorate},
decoration={markings,mark=at position .75 with {\arrow[draw=green,fill=green]{triangle 60}},mark=at position .25 with {\arrow[draw=red,fill=red]{triangle 60 reversed}}}}
}

\fill[dotted,fill=black,fill opacity=0.1] (2.50,2.00) -- (7.50,2.00) -- 
   (7.50,4.00) -- (2.50,4.00) -- cycle;
\draw (1.25,4.25) node{$L$};
\fill[dotted,fill=black,fill opacity=0.1] (2.25,4.00) -- (2.25,0.00) -- 
   (0.25,0.00) -- (0.25,4.00) -- cycle;
\draw (5.00,1.75) node{$M$};
\fill[dotted,fill=black,fill opacity=0.1] (7.75,2.00) -- (7.75,6.00) -- 
   (9.75,6.00) -- (9.75,2.00) -- cycle;
\draw (8.75,1.75) node{$R$};
\fill[dotted,fill=black,fill opacity=0.1] (2.00,8.00) -- (5.50,8.00) -- 
   (5.50,6.00) -- (2.00,6.00) -- cycle;
\draw (3.75,5.75) node{$G_i$};

\fill (0.50,0.25) circle (1.5pt);
\draw (0.50,0.25) node[below=5pt] {$l^{\lgh,1}$};
\draw[red] (0.50,0.25) -- (1.50,0.25); 
\draw[color=red,rounded corners=4pt] (1.00,0.25) -- (2.50,0.25) -- 
   (2.50,0.75) -- (1.75,0.75);
\draw[red] (1.75,0.75) -- (1.00,0.75); 
\draw[red] (1.00,1.75) -- (0.25,1.75); 
\draw[color=black] (1.25,1.40) node {$\vdots$};
\draw[color=red,rounded corners=4pt] (0.25,1.75) -- (0.00,1.75) -- 
   (0.00,2.25) -- (0.25,2.25);
\draw[color=red] (0.25,2.25) -- (2.00,2.25);
\draw[red] (2.00,2.25) -- (3.00,2.25);
\draw[color=red] (3.00,2.25) -- (7.00,2.25);
\draw[red] (7.00,2.25) -- (8.00,2.25);
\draw[color=red] (8.00,2.25) -- (9.50,2.25);
\draw[color=red,rounded corners=4pt] (9.50,2.25) -- (10.00,2.25) -- 
   (10.00,2.75) -- (9.50,2.75);
\draw[red] (9.50,2.75) -- (8.00,2.75);
\draw[color=red] (8.00,2.75) -- (3.00,2.75);
\draw[red] (3.00,2.75) -- (2.00,2.75);
\draw[color=black] (1.25,2.95) node {$\vdots$};
\draw[color=black] (7.00,3.25) node {$\vdots$};
\draw[color=black] (8.75,3.25) node {$\vdots$};
\draw[red] (1.00,3.25) -- (0.25,3.25); 
\draw[color=red,rounded corners=4pt] (0.25,3.25) -- (0.00,3.25) -- 
   (0.00,3.75) -- (0.25,3.75);
\draw[red] (0.25,3.75) -- (1.00,3.75); 
\draw[color=red] (1.00,3.75) -- (2.00,3.75);
\draw[red] (2.00,3.75) -- (2.75,3.75); 
\draw[color=black] (3.00,3.75) node{$\cdots$};
\draw[red] (4.50,3.75) -- (5.50,3.75); 
\draw[color=black] (6.00,3.75) node{$\cdots$};
\draw[red] (3.25,3.75) -- (4.00,3.75); 
\draw[red] (7.00,3.75) -- (8.00,3.75); 
\draw[color=red] (8.00,3.75) -- (9.50,3.75);
\draw[color=red,rounded corners=4pt] (9.50,3.75) -- (10.00,3.75) -- 
   (10.00,4.25) -- (9.50,4.25);
\draw[red] (9.50,4.25) -- (8.00,4.25); 
\draw[color=red,rounded corners=4pt] (8.00,4.25) -- (7.50,4.25) -- 
   (7.50,4.75) -- (8.00,4.75);
\draw[red] (8.00,4.75) -- (8.75,4.75); 
\draw[color=black] (9.25,5.25) node {$\vdots$};
\draw[red] (8.75,5.25) -- (7.75,5.25); 
\draw[color=red,rounded corners=4pt] (7.75,5.25) -- (7.50,5.25) -- 
   (7.50,5.75) -- (8.00,5.75);
\draw[red] (8.00,5.75) -- (9.50,5.75); 
\fill (9.50,5.75) circle (1.5pt);
\draw (9.50,5.75) node[above=5pt] {$r^{1,\lgw}$};

\fill (4.00,3.75) circle (1.5pt);
\draw (4.30,1.00) node[below] {$\sigma_{i+1}-1$};
\draw[->] (4.30,1.00) -- (4.00,3.65);
\fill (4.50,3.75) circle (1.5pt);
\draw (7.15,1.00) node[below=2pt] {$\sigma_{i+1}$};
\draw[->] (7.15,1.00) -- (4.50,3.65);
\draw[red] (4.00,3.75) .. controls (2.00,4.00) and 
    (0.00,7.50) .. (2.25,7.75); 
\draw[red] (5.50,6.25) .. controls (6.50,6.25) and 
    (5.00,3.75) .. (4.50,3.75); 

\draw[red] (2.25,7.75) -- (3.25,7.75); 
\draw[color=red] (3.25,7.75) -- (5.25,7.75); 
\draw[color=red,rounded corners=4pt] (5.25,7.75) -- (5.75,7.75) -- 
   (5.75,7.25) -- (5.25,7.25);
\draw[red] (5.25,7.25) -- (3.25,7.25); 
\draw[color=black] (2.75,7.10) node {$\vdots$};
\draw[color=red,rounded corners=4pt] (2.25,6.75) -- (1.75,6.75) -- 
   (1.75,6.25) -- (2.25,6.25);
\draw[red] (2.25,6.25) -- (5.50,6.25); 

\end{tikzpicture}
\end{center}
\caption{The red path across graph $G'$.  The displayed pattern that connects grid $M$ with grid $G_i$ is repeated for every $i \in [n]$.}
\label{fig:redpath}
\end{figure}

\paragraph{Red path} 
The red path is schematically represented in Figure~\ref{fig:redpath}.
The path starts from $l^{\lgh,1}$ (that is, the bottom left corner of $L$) and travels horizontally through the $\lminusmh$ bottom lines of $L$, alternating the left-to-right and the right-to-left directions and moving upward, until it reaches node $l^{\mgh,1}$ coming from previous node $l^{\mghplusone,1}$. Then the path continues traveling horizontally through the three components $L$, $M$ and $R$, once again alternating the horizontal directions.  The path eventually reaches the node $m^{1,1}$ coming from previous node $l^{1,\lgw}$, since $\mgh$ is always odd. At this point, the red path continues horizontally, from left to right, until it arrives at some node $m^{1,\sigma_{i+1}-1}$, $i \in [n]$; let us call $x_i$ such a node. Observe that $x_i$ is the penultimate node in the first row of $M$, from left to right, that is connected with a node in the leftmost column of $G_i$ (see Figure~\ref{fig:extraedgeconnections1}). 

Next, the path moves one step forward from $x_i$ to the leftmost
column of $G_i$, reaching node $g_i^{1,1}$. Now, the path
starts traveling horizontally from the first to the last row of
$G_i$, alternating the left-to-right direction 
with the right-to-left direction, until it arrives at
$g_i^{\ngh,\ngw}$. Afterward, the red path returns to $M$ by reaching
the node to the right of $x_i$. The path then continues
horizontally from left to right, repeating the process of
visiting the components $G_i$ for $i \in [n]$, as described above,
eventually reaching node $r^{\lminusmh+1,1}$ from previous node
$m^{1,\mgw}$. The path can now visit the remaining nodes of $R$
by traveling horizontally and alternating the left-to-right and
the right-to-left directions, until it reaches node $r^{1,\lgw}$,
where the path stops.

\begin{figure}[t]
\begin{center}
\definecolor{qqqqff}{rgb}{0,0,1}
\begin{tikzpicture}[scale=1.00]


\tikzset{
red/.style={draw=red, postaction={decorate},
decoration={markings,mark=at position .55 with {\arrow[draw=red,fill=red]{triangle 60}}}},
green/.style={draw=green, postaction={decorate}, decoration={markings,mark=at position .55 with {\arrow[draw=green,fill=green]{triangle 60}}}},
green-dash/.style={dashed, draw=green, postaction={decorate}, decoration={markings,mark=at position .55 with {\arrow[draw=green,fill=black,fill opacity=0.1]{triangle 60}}}},
blue/.style={draw=blue, postaction={decorate},
decoration={markings,mark=at position .75 with {\arrow[draw=green,fill=green]{triangle 60}},mark=at position .25 with {\arrow[draw=red,fill=red]{triangle 60 reversed}}}}
}

\fill[dotted,fill=black,fill opacity=0.1] (2.25,4.00) -- (2.25,0.00) -- 
   (0.25,0.00) -- (0.25,4.00) -- cycle;
\draw (1.75,0.00) node[below]{$L$};
\fill[dotted,fill=black,fill opacity=0.1] (2.50,2.00) -- (9.50,2.00) -- 
   (9.50,4.00) -- (2.50,4.00) -- cycle;
\draw (4.00,1.75) node{$M$};
\fill[dotted,fill=black,fill opacity=0.1] (9.75,2.00) -- (9.75,6.00) -- 
   (11.75,6.00) -- (11.75,2.00) -- cycle;
\draw (11.00,6.25) node{$R$};
\fill[dotted,fill=black,fill opacity=0.1] (2.00,8.00) -- (8.50,8.00) -- 
   (8.50,6.00) -- (2.00,6.00) -- cycle;
\draw (2.00,7.00) node[left]{$G_i$};

\fill (0.50,3.75) circle (1.5pt);
\draw (0.25,4.00) node[above=5pt] {$l^{1,1}$};
\draw[green] (0.50,3.75) -- (0.50,0.25); 
\draw[color=green,rounded corners=4pt] (0.50,0.25) -- (0.50,-0.25) -- 
   (1.00,-0.25) -- (1.00,0.25);
\draw[color=black] (1.00,1.75) node {$\vdots$};
\draw[green] (1.00,0.25) -- (1.00,1.25); 
\draw[color=black] (2.00,2.75) node {$\vdots$};
\draw[green] (2.00,3.00) -- (2.00,3.75); 

\draw[color=green,rounded corners=4pt] (2.00,3.75) -- (2.00,4.25) -- 
   (2.75,4.25) -- (2.75,3.75);
\draw[green] (2.75,3.75) -- (2.75,2.25); 
\draw[color=green,rounded corners=4pt] (2.75,2.25) -- (2.75,1.75) -- 
   (3.25,1.75) -- (3.25,2.25);
\draw[green] (3.25,2.25) -- (3.25,2.75); 
\draw[color=black] (3.25,3.25) node {$\vdots$};

\draw[green] (5.00,2.75) -- (5.00,4.00); 
\draw[green] (5.00,4.00) .. controls (5.00,5.00) and 
    (2.25,5.00) .. (2.25,6.00); 
\draw[green] (2.25,6.00) -- (2.25,7.75); 
\draw[color=green,rounded corners=4pt] (2.25,7.75) -- (2.25,8.25) -- 
   (2.75,8.25) -- (2.75,7.75);
\draw[green] (2.75,7.75) -- (2.75,6.00); 
\draw[green] (2.75,6.00) .. controls (2.75,5.00) and 
    (5.50,5.00) .. (5.50,4.00); 
\draw[green] (5.50,4.00) -- (5.50,2.75); 
\draw[color=black] (5.00,2.50) node {$\vdots$};
\draw[color=black] (5.50,2.50) node {$\vdots$};
\draw[color=black] (3.25,7.00) node {$\vdots$};

\draw[green] (6.50,2.75) -- (6.50,4.00); 
\draw[green] (6.50,4.00) .. controls (6.50,5.00) and 
    (3.75,5.00) .. (3.75,6.00); 
\draw[green] (3.75,6.00) -- (3.75,8.00); 
\draw[green] (3.75,8.00) .. controls (3.75,8.50) and 
    (3.75,9.00) .. (5.00,9.00); 
\draw[green] (5.00,8.50) .. controls (4.50,8.50) and 
    (4.25,8.50) .. (4.25,8.00); 
\draw[color=black] (5.00,9.00) node[right]{to $G_j$};
\draw[color=black] (5.00,8.50) node[right]{from $G_j$};
\draw[green] (4.25,8.00) -- (4.25,6.00); 
\draw[green] (4.25,6.00) .. controls (4.25,5.00) and 
    (7.00,5.00) .. (7.00,4.00); 
\draw[green] (7.00,4.00) -- (7.00,2.75); 
\draw[color=black] (6.50,2.50) node {$\vdots$};
\draw[color=black] (7.00,2.50) node {$\vdots$};
\draw[color=black] (4.75,7.00) node {$\vdots$};

\draw[green] (8.00,2.75) -- (8.00,4.00); 
\draw[green] (8.00,4.00) .. controls (8.00,5.00) 
    .. (5.75,5.50);  
\draw[green-dash] (5.75,5.50) .. controls (5.00,6.00) and 
    (5.00,9.00) .. (6.00,8.00);  
\draw[green] (6.00,8.00) -- (6.00,6.25); 
\draw[color=green,rounded corners=4pt] (6.00,6.25) -- (6.00,5.75) -- 
   (6.50,5.75) -- (6.50,6.25);
\draw[green] (6.50,6.25) -- (6.50,7.00); 
\draw[green] (7.75,7.00) -- (7.75,6.25); 
\draw[color=green,rounded corners=4pt] (7.75,6.25) -- (7.75,5.75) -- 
   (8.25,5.75) -- (8.25,6.25);
\draw[green] (8.25,6.25) -- (8.25,8.00); 
\draw[green] (8.25,8.00) .. controls (9.25,9.00) and 
    (9.50,7.00) .. (8.75,5.50);  
\draw[green] (8.75,5.50) .. controls (8.50,5.00) 
    .. (8.50,4.00);  
\draw[green] (8.50,4.00) -- (8.50,3.00); 
\draw[color=black] (8.00,2.50) node {$\vdots$};
\draw[color=black] (8.50,2.50) node {$\vdots$};
\draw[color=black] (6.50,7.50) node {$\vdots$};
\draw[color=black] (7.75,7.50) node {$\vdots$};

\draw[color=black] (9.25,3.50) node {$\vdots$};
\draw[green] (9.25,3.00) -- (9.25,2.25); 
\draw[color=green,rounded corners=4pt] (9.25,2.25) -- (9.25,1.75) -- 
   (10.00,1.75) -- (10.00,2.25);
\draw[green] (10.00,1.75) -- (10.00,5.75); 
\draw[color=green,rounded corners=4pt] (10.00,5.75) -- (10.00,6.25) -- 
   (10.50,6.25) -- (10.50,5.75);
\draw[green] (10.50,5.75) -- (10.50,4.25); 
\draw[color=black] (10.50,3.75) node {$\vdots$};
\draw[green] (11.50,4.75) -- (11.50,2.25); 
\draw[color=black] (11.50,5.50) node {$\vdots$};
\fill (11.50,2.25) circle (1.5pt);
\draw (11.50,2.25) node[below=5pt] {$r^{\rgh,\rgw}$};

\draw (4.30,0.00) node[below] {$\sigma_{i+1}-3$};
\draw[->] (4.30,0.00) -- (8.00,2.00);
\draw (6.70,0.00) node[below] {$\sigma_{i+1}-2$};
\draw[->] (6.70,0.00) -- (8.50,2.00);

\end{tikzpicture}
\end{center}
\caption{The green path across graph $G'$.  Grid $M$ is connected with grid $G_i$, $i \in [n]$, through three different patterns, displayed in the figure, one for each of the three blocks of $G_i$.}
\label{fig:greenpath}
\end{figure}

\paragraph{Green path} 
The green path is schematically represented in Figure~\ref{fig:greenpath}.
The path starts from $l^{1,1}$ (that is, the top left corner of $L$). It travels vertically, alternating the top-to-bottom direction with the bottom-to-top direction and moving rightward, until it reaches node $l^{1,\lgw}$ from previous node $l^{2,\lgw}$.  The path then moves to $m^{1,1}$, travels vertically down to $m^{\mgh,1}$, moves one step to the right to node $m^{\mgh,2}$ and again travels vertically up to $m^{1,2}$. From now on, as soon as there is an edge exiting $M$ and reaching some yet unvisited node in the last row of some grid $G_i$, the green path follows such an edge and travels vertically through the current column in $G_i$, until it reaches a node $x_i$ in the first row. We need to distinguish two cases for $x_i$.

\begin{itemize}
\item If $x_i$ has no edge exiting $G_i$, then the green path
      makes a step to the vertex to the right of $x_i$. This means that 
      we are in Block 1 of $G_i$. 
\item On the other hand, if $x_i$ has an edge exiting $G_i$, then the green path
      follows this edge thus reaching a
      node in the last row of a grid $G_j$, for some $j > i$ (as in
      Figure~\ref{fig:extraedgeconnections1}).  This means that we are visiting 
      the first part of Block 2 of $G_i$, where edges  
      $(v_i, v_j)$ in the source graph $G$ with $j > i$ are encoded by our 
      construction.  The path then travels vertically
      through two columns in $G_j$, until it reaches the node in
      the bottom row of the second column, and, from there, it
      returns to $G_i$ at the vertex to the right of $x_i$.
\end{itemize}

From the vertex to the right of $x_i$, the green path continues
vertically down in the current column of $G_i$.  Upon reaching the 
bottom vertex of this column, the path exits $G_i$ and comes back 
to some node in the top row of $M$.  Such node is placed in the column 
of $M$ adjacent at the right to the last column of $M$ that the green path 
had visited before its exit to $G_i$.  Then the green path proceeds 
downward, along the current column of $M$, it moves to the next column 
at the right, and alternates its direction to reach the node in the first 
row of $M$.  The above process is then iterated, until all the columns 
in Block~1 and Block 2 of the current grid $G_i$ have been visited. 

When the green path reaches node $m^{1,\sigma_{i+1} - 3}$, it exits 
$M$ and reaches the top node in the column of $G_i$ 
with index $\nee-3+2d^<_i$ (see Figure~\ref{fig:extraedgeconnections1}),
entering for the first time Block 3 of the current grid $G_i$.
We remark that this step represents a switch in the construction of the 
green path, in the following sense.  Block 1 and Block 2 of $G_i$ 
are visited by the green path in such a way that odd-indexed
columns are visited bottom-to-top and even-indexed
columns are visited top-to-bottom.  On the other hand, 
when Block 3 of $G_i$ is entered, we revert the previous pattern 
in such a way that odd-indexed columns are visited 
top-to-bottom and even-indexed columns are visited bottom-to-top.

Once Block 3 of $G_i$ is entered, the green path 
travels vertically through its columns, by alternating 
direction and moving rightward, never leaving $G_i$
at its intermediate nodes.  In this way the path eventually 
reaches the node $g_i^{1,\ngw}$, at which point it can return to $M$, 
reaching the topmost node in the column with index $\sigma_{i+1}-2$ 
(see again Figure~\ref{fig:extraedgeconnections1}).

At this point the columns with indices $\sigma_{i+1}-2$, $\sigma_{i+1}-1$, and $\sigma_{i+1}$ are visited vertically, alternating the top-to-bottom direction with the bottom-to-top direction and moving rightward.  After this step, the green path is located at the bottom of column $\sigma_{i+1}+1$, coming from the bottom of column $\sigma_{i+1}$, and it moves upward to the first line of $M$, where the path is ready to start visiting the next grid $G_{i+1}$.  This is done by iterating all of the process described above. 

When all of the grids $G_i$ have been visited, there are no more edges exiting $M$. The path then continues vertically, alternating the top-to-bottom direction with the bottom-to-top direction and moving rightward, until it reaches node $m^{\mgh,\mgw}$, 
since $\mgw-\sigma_{n+1}=\mgw-\nee \cdot n$ is always even. 
The path can now move one step to the right to node $r^{\rgh,1}$, and visit the remaining nodes of $R$ by traveling vertically and alternating the vertical direction.
In the end, the green path ends at node $r^{\rgh,\rgw}$.

\paragraph{Double Edges}
Some edges are included in both the red path and the green path.
These are double edges in $G'$, and count double when computing 
the width function.  Double edges occur 
on the first and the last columns 
of the various component grids $G_i$, $L$, $M$, and $R$, 
where the red path crosses from one row to the next, and on the 
upper and lower rows of these grids, where the green path crosses
from one column to the next.  There are no duplicate edges in
the interior of any grid.  There are also no duplicate
edges between any two grids, with the only exceptions of the points
where the green path connects $L$ to $M$, and the point where
the green path connects $M$ to $R$.

To be used later, we need to establish exact values for the cutwidth of the multigraph
grids $G_i$, $L$, $M$, and $R$.  In order to do so, 
we define a multigraph grid $X_m$, 
which consists of the (regular) grid $X=\Gamma[H,W]$ with the following double
edges
\begin{itemize}
\item from $x^{1,i}$ to $x^{1,i+1}$, for any $i$ odd in $[W-1]$
\item from $x^{H,i}$ to $x^{H,i+1}$, for any $i$ even in $[W-1]$
\item from $x^{i,1}$ to $x^{i+1,1}$, for any $i$ odd in $[H-1]$
\item from $x^{i,W}$ to $x^{i+1,W}$, for any $i$ even in $[H-1]$.
\end{itemize}
The set of edges of the grid $X$ is contained in
the set of edges of the multigraph grid $X_m$.  Hence the cutwidth of 
$X_m$ is at least $\min\{H+1,W+1\}$, which is the cutwidth of $X$
as reported at the beginning of Section~\ref{ssec:gprime}.
Without loss of generality, let us assume $H\leq W$, and let us 
consider the linear arrangement $\nu$ of $X_m$ such that, for any
$x^{i,j}$, $\nu(x^{i,j})=(j-1)H+i$ if $j$ is even, and $(j-1)H+H-i$
otherwise.  It is easy to verify that $\nu$ induces a maximum width equal to $H+1$.
This proves that the cutwidth of the multigraph grid $X_m$ 
is still $\min\{H+1,W+1\}$.

Observe that the multigraph grid $R$ which we use in $G'$ 
has a set of edges which is a proper subset of the set of edges 
of $X_m$, for appropriate values of $H$ and $W$.  
The subset relation follows from the fact that $R$ does not have
double edges at the portion of its leftmost column that connects with $M$.   
Similarly, the multigraph grid $L$ has a set of edges 
which is a proper subset of the set of edges of an upside-down instance
of $X_m$.  Thus both $L$ and $R$ have cutwidth $\min\{H+1,W+1\}$. 
Consider now grid $M$.  The set of its edges is a proper subset of the edges 
of an upside-down instance of $X_m$.  
This follows from the fact that $M$ does not have
double edges at its leftmost and rightmost columns, where it connects 
with $L$ and $R$, respectively.  Furthermore, the green path 
through $M$ sometimes leaves the first row of $M$ 
to connect to some grid $G_i$, as depicted in Figure~\ref{fig:greenpath}.
Thus we can claim a cutwidth of $\min\{H+1,W+1\}$ for this grid as well.
It is easy to verify that each of the green paths through grids $L$, $M$ and $R$
corresponds to a linear arrangement that realizes the maximum width 
of $\min\{H+1,W+1\}$ for these grids.

Finally, each grid $G_i$ can be split into two parts. 
The first part consists of what we have called blocks~1 and~2, 
and the second part consists of block~3. 
The first part is a grid with a proper subset of the edges 
of an upside-down instance of $X_m$, for appropriate values of $H$ and $W$.  
This is so because the green path at block~2 repeatedly leaves $G_i$
to connect to three other grids $G_j$, $j \neq i$, irrespective of
whether these connections are backward or forward in the source graph $G$.
The second part of $G_i$ is an instance of $X_m$, for appropriate values of $H$ and $W$.
The difference between these two parts is due to 
the fact that, when moving from block~2 to block~3, 
the green path switches to a ``reversed'' pattern, as already observed
in this section. 
An optimal linear arrangement for $G_i$ can be 
defined by following the green path within each column of this grid,
and moving from one column to the next in a left to right order. 
The maximum width in the first part of $G_i$, with the exception of the last column, 
is then $\min\{H+1,W+1\}$, and this is also the maximum width in the 
second part, with the exception of the first column. 
It is not difficult to verify that, even for the positions
corresponding to the two adjacent columns above, this arrangement
induces a maximum width of $\min\{H+1,W+1\}$.
We have thus found that all of the grid components in $G'$ 
have cutwidth of $\min\{H+1,W+1\}$. 

We conclude our construction by setting $k'=\kprime$ in the target 
instance $\langle G',k'\rangle$ of the \mrgbcw\ problem.

\subsection{\mcess\ to \mrgbcw}
\label{ssec:g2gprime}

We show here that if $G$ admits a partition of its nodes into two equal size subsets, inducing a cut of size at most $k$, then $G'$ admits a linear arrangement $\larr$ whose maximum width is at most $k'$.  

\begin{lemma}
\label{l:g2gprime}
If $\langle G, k\rangle$ is a positive instance of \mcess\ then $\langle G',k'\rangle$ is a positive instance of \mrgbcw.
\end{lemma}

\begin{proof}
We specify a linear arrangement of 
$G'$ having width no greater than $k'$ at each position.
We arrange the vertices within
each grid component $G_i$, as well as within grid components $R$, $M$, and $L$,
to be contiguous to one another.  Within each grid component,
we arrange the vertices in column-major order proceeding
through the columns from left to right; within each column,
we place vertices in the order specified by
the green path of $G'$.  
This is the same linear arrangement for each multigraph grid that 
has been presented in the paragraph ``Double Edges'', at the end of 
Section~\ref{ssec:gprime}.  
Disregarding edges that are not internal to the grid itself,
this results in a maximal width of $H+1$ for each individual grid, as 
already discussed. 

We concatenate the linear arrangements for the grid components
in a manner corresponding to a solution of the \mcess\ 
problem given by $\langle G,k\rangle$.
To this end, let $V_1$ and $V_2$ be the sets in a partition of the vertices
in $G$ such that $\size{V_1} = \size{V_2}$ and at most $k$ edges in $G$
have one endpoint in $V_1$ and the other endpoint in $V_2$.  

Our linear arrangement begins with the grid components 
$G_i$ for all $i$ such that $v_i \in V_1$,
in any order, then concatenates components $L$, $M$, and $R$,
and finally adds $G_i$ for all $i$ such that $v_i \in V_2$, in any order.
Each position in the linear arrangement within component $L$ 
has at most $3n^4 + 2$ edges internal to $L$, since 
the height of $L$'s grid is $3n^4 + 1$.
In addition, each position within $L$ 
has $4n^2$ edges connecting $M$ to each of the $\frac{n}{2}$
components $G_i$ to the left of $L$, for a total of $2n^3$ 
edges.  Finally, each position within $L$ 
has at most $2k$ edges connecting components $G_i$ and $G_j$
for $i,j$ such that $v_i \in V_1$, $v_j \in V_2$ and $(v_i, v_j) \in E$.  
Thus, the total width at each position within $L$ is at most 
$3n^4 + 2 + 2n^3 + 2k = k'$.  The same analysis applies
to each position within $R$.  

At all other positions in the linear arrangement, we have smaller width.
This is because, for positions within each $G_i$,
we have at most $2n^4+2$ edges from the grid $G_i$ itself,
at most $2n^3$ edges from $M$ to any $G_j$ standing on the same 
side as $G_i$ with respect to $M$, and
no more than $3n$ edges from some $G_j$ to some $G_h$, 
since the source graph $G$ is cubic and each connection between 
two vertices in $G$ corresponds to two arcs connecting the associated grids in $G'$.
For positions within $M$, we have at most $2n^4+2$
edges internal to $M$, at most $4n^3$ edges from $M$ to some $G_i$,
and $2k < 3n$ edges from some $G_i$ to some $G_j$.
Finally, positions between grid components have at most 
$2n^3$ edges from $M$ to some $G_i$, at most
$3n$ edges from some $G_i$ to some $G_j$,
and, in the case of positions between $M$ and either $L$ or $R$,
$2n^4+1$ edges connecting $M$ to either $L$ or $R$.
Thus, all positions outside grids $L$ and $R$ have a width 
bound of $2n^4 + \order{n^3}$.

We conclude that the maximum width of the linear arrangement is
that of the $L$ and $R$ components, $3n^4 + 2 + 2n^3 + 2k = k'$.
Then $\langle G',k'\rangle$ is a positive instance of \mrgbcw.
\end{proof}

\subsection{\mrgbcw\ to \mcess}
\label{ssec:gprime2g}

In this section we shall prove that if 
$G'$ admits a linear arrangement 
$\larr$ whose maximum width is at most $k'$, then our source graph 
$G$ admits a partition into two equal size subsets of nodes inducing a cut of
size at most $k$.  To this aim, we need to develop several intermediate
results.   Informally, our strategy is to investigate 
the family of linear arrangements for $G'$ having maximum width 
bounded by $k' + n^2$.  We show that, in these arrangements, two 
important properties hold for the grid components $L, M, R$ and 
$G_i$, $i \in [n]$, of $G'$, described in what follows.  
\begin{itemize}
\item 
The first property states that, for each grid component, a subset of its nodes must appear all in a row in the linear arrangement.  We call such a subset the \textit{kernel} of the grid.  In other words, nodes from different kernels cannot be intermixed, and each linear arrangement induces a total order among the kernels.  In addition, the kernels of the grids $L, M, R$ must appear one after the other in the total order, and the kernels of the grids $G_i$ can only be placed to the left or to the right of the kernels of $L, M, R$.  We therefore call $L, M, R$ the middle grids.  
\item
We illustrate the second property by means of an example.
Consider one of the middle grids, say $L$.  Assume that, under
our liner arrangement, there is a grid $X$ with kernel to the
left of $L$'s kernel and a grid $Y$ with kernel to the right of
$L$'s kernel.  Assume also some edge $e$ of $G'$, connecting a
node $x$ from $X$ with a node $y$ from $Y$.  If $x$ and $y$ are
in the kernels of their respective grids, edge $e$ must cross
over $L$'s kernel, contributing one unit to the width of $G'$ at
each gap $i$ within $L$'s kernel.  If $x$ and $y$ are not in the
kernels of their respective grids, it is possible to ``misplace''
one of these two nodes, say $x$, moving it to the opposite side
with respect to $L$'s kernel, in such a way that $e$ no longer
contributes to the width at $i$.  The second property states
that, if we do this, we will bring new edges, internal to grid
$X$, into the count of width at $i$.  This means that, if our
goal is the one of optimizing the width at $i$, we will have no
gain in misplacing node $x$ or node $y$.
\end{itemize}
With the two properties above, we can then show that exactly $\frac{n}{2}$ of the $G_i$ grids must be placed to the left of the middle grids $L, M, R$, and all of the remaining $G_i$ grids must be placed to the right of the middle grids, which eventually leads to the fact that if $\langle G', k' \rangle$ is a positive instance of \mrgbcw\ then $\langle G, k \rangle$ is a positive instance of \mcess. 

We start with some preliminary results, needed to prove 
the first property above.   Let $V_1$ and $V_2$ be sets of nodes from some graph
with $V_1 \cap V_2 = \emp$, and let $E$ be the set of edges of
the graph.  We write $\delta\left(V_1,V_2\right) =
\size{\left\{(u,v) : (u,v) \in E \wedge u \in V_1 \wedge v\in
  V_2\right\}}$.

\begin{lemma}
\label{lem:partition}
For any grid $X=\grid[\ggh,\ggw]$ with $\ggw \geq 2\ggh+1$ and
for any partition of its nodes in two sets $V_1$ and $V_2$ with
$\size{V_1}\geq \ggh^2$ and $\size{V_2}\geq \ggh^2$, we have
$\delta\left(V_1,V_2\right) \geq \ggh$.
\end{lemma}

\begin{proof}
We distinguish the following three cases.
\begin{enumerate}
\item 
For each $h$ with $1 \leq h \leq \ggh$ there exist $w_{h,1}$ and
$w_{h,2}$ with $1 \leq w_{h,1}, w_{h,2} \leq \ggw$ such that
$x^{h,w_{h,1}} \in V_1$ and $x^{h,w_{h,2}} \in V_2$.  This
implies that, for each row of the grid, there exists at least one
edge connecting one node in $V_1$ to one node in $V_2$. Hence,
$\delta\left(V_1,V_2\right) \geq \ggh$.

\item There exists $h$ with $1 \leq h \leq \ggh$ such that, for
  any $w$ with $1 \leq w \leq \ggw$, $x^{h,w} \in V_1$. In this
  case, for each $w$ with $1 \leq w \leq \ggw$, either there
  exists $h_w$ with $1 \leq h_w \leq \ggh$ such that $x^{h_w,w}
  \in V_2$ (and, hence, the $w$-th column of $X$ contributes to
  $\delta\left(V_1,V_2\right)$ by at least one unit) or else, for all
  $h$ with $1 \leq h \leq \ggh$, $x^{h,w} \in V_1$.  This latter
  case can happen at most $\left \lfloor
  \frac{\left|V_1\right|}{\ggh} \right \rfloor$ times: this
  implies that the former case happens at least $\ggw - \left
  \lfloor \frac{\left|V_1\right|}{\ggh} \right \rfloor$. Hence,
\[\delta\left(V_1,V_2\right) \geq \ggw - \left \lfloor \frac{\left|V_1\right|}{\ggh} \right \rfloor \geq \frac{\ggw\ggh-\left|V_1\right|}{\ggh} = \frac{\left|V_2\right|}{\ggh} \geq \ggh,\]where the last inequality is due to the fact that $\left|V_2\right|\geq \ggh^2$.

\item There exists $h$ with $1 \leq h \leq \ggh$ such that, for
  any $w$ with $1 \leq w \leq \ggw$, $x^{h,w} \in V_2$. We can
  deal with this case similarly to the previous one.
\end{enumerate}
The lemma thus follows.
\end{proof}

\begin{corollary}
\label{cor:kernel}
For any grid $X=\grid[\ggh,\ggw]$ with $\ggw \geq 2\ggh+1$, for any linear arrangement
$\nu$ of $X$, and for any $i$ with 
$\ggh^2 \leq i \leq \ggh\ggw-\ggh^2$, $\wg(X,\nu,i) \geq \ggh$.
\end{corollary}

\begin{proof}
The result follows by observing that, for any $i$ with $\ggh^2
\leq i \leq \ggh\ggw-\ggh^2$, we can define a partition of the
nodes of the grid by including in $V_1$ all the nodes $x$ such
that $\nu(x)\leq i$ and by including in $V_2$ all the other
nodes. Since this partition satisfies the hypothesis of the
previous lemma, we have that $\wg(X,\nu,i) \geq \ggh$.
\end{proof}

Let $\nu$ be an arbitrary linear arrangement for the nodes of $G'$. 
We denote by $\nu_i$ (respectively, $\nu_L$, $\nu_M$, and $\nu_R$) 
the linear arrangement of $G_i$ (respectively, $L$, $M$, and $R$) induced by $\nu$. 
Moreover, for any node $x$ of $G_i$ (respectively, $L$, $M$, and $R$)
and the associated position $p = \nu(x)$ under $\nu$, 
we denote by $p_i = \nu_i(x)$ (respectively, $p_L = \nu_L(x)$, 
$p_M = \nu_M(x)$, and $p_R = \nu_R(x)$) 
the corresponding position of $x$ under $\nu_i$ 
(respectively, $\nu_L$, $\nu_M$, and $\nu_R$). 

We now introduce the notion of kernel, which plays a major role 
in the development of our proofs below.  Consider any linear arrangement 
$\nu$ of $G'$ and any of the grids $G_i$.  The \termdef{kernel} $K^{(\nu)}_i$
relative to $\nu$ and $G_i$ is a set of positions $p$ 
of the nodes of $G'$ under $\nu$ such that 
$\left(\ngh\right)^2 \leq p_i \leq \left(\ngw\right)\left(\ngh\right) - \left(\ngh\right)^2$.  Corollary~\ref{cor:kernel} implies that for any $p \in K^{(\nu)}_i$, 
$\wg(G_i,\nu_i,p_i) \geq \ngh$.   

Similarly, we define the kernel $K^{(\nu)}_L$ (respectively,
$K^{(\nu)}_R$) as the set of positions $p$ of the nodes of $G'$
under $\nu$ such that $\left(\lgh\right)^2 \leq p_L, p_R \leq
\left(\lgw\right)\left(\lgh\right) - \left(\lgh\right)^2$.
Again, Corollary~\ref{cor:kernel} implies that for any $p \in
K^{(\nu)}_L$ (respectively, $p \in K^{(\nu)}_R$), we have
$\wg(L,\nu_L,p_L) \geq \lgh$ (respectively, $\wg(R,\nu_R,p_R)
\geq \rgh$).
%
We define the
kernel $K^{(\nu)}_M$ as the set of positions $p$ of the nodes of
$G'$ under $\nu$ such that
$\left(\mgh\right)^2 \leq p_M \leq
\left(\mgh\right)\left(\mgw\right)-\left(\mgh\right)^2$.
Corollary~\ref{cor:kernel} implies that for any $p \in
K^{(\nu)}_M$ we have $\wg(M,\nu_M,p_M) \geq \mgh$.

Observe that, for any $i \in [n]$, we have $\size{K^{(\nu)}_i}
= (\ngh)(\ngw)-2((\ngh)^2) + 1 \geq \nklb$ for $n$ sufficiently large.
Furthermore, for $n$ sufficiently large, we have $\size{K^{(\nu)}_L} =
\size{K^{(\nu)}_R} = (\lgh)(\lgw) - 2((\lgh)^2) + 1 \geq \lrklb$, and
$\size{K^{(\nu)}_M} = (\mgh)(\mgw)-2((\ngh)^2+\mgw) + 1 \geq \mklb$.

Recall that in our construction in Section~\ref{ssec:gprime} 
we have set $k' = \kprime$.  From now on, we denote by $\larr$ 
any linear arrangement of $G'$ having maximum width at most $k' + n^2$.  
For any two sets of positive integers $A$ and $B$, we will write $A < B$ if
each element of $A$ is smaller than every element in $B$.

\begin{lemma}\label{l2}
Let $\larr$ be a linear arrangement of $G'$ having maximum width at most $k' + n^2$,
and let ${\cal K}^{(\larr)} = \{K^{(\larr)}_L, K^{(\larr)}_M,
K^{(\larr)}_R\} \cup \{K^{(\larr)}_i : i \in [n]\}$. 
For any pair of kernels $K', K'' \in {\cal K}^{(\larr)}$ with $K' \neq K''$, 
either $K' < K''$ or $K'' < K'$.
\end{lemma}

\begin{proof}
We first consider the kernels in $\{K^{(\larr)}_i : i \in [n]\}$.
Let $p, p' \in K^{(\larr)}_i$ be two positions such that $p_i =
p'_i-1$.  Assume that there exists a position $q \in
K^{(\larr)}_j$, $j \neq i$, such that $p < q < p'$.  We know 
(by Corollary~\ref{cor:kernel} and definition of kernel) that
$\wg(G_i,\nu'_i,p_i) \geq \ngh$ and $\wg(G_j,\nu'_j,q_j) \geq
\ngh$.  Since $G_i$ and $G_j$ have disjoint edge sets, and since
in between $p$ and $q$ there is no position associated with a
node from $G_i$, we conclude that $\wg(G',\larr,q) \geq 4n^4+2 >
k' + n^2$, for $n$ sufficiently large.  This is in contrast with our
assumption about the linear arrangement~$\larr$.

Essentially the same argument can be used when we consider all
of the kernels in~${\cal K}^{(\larr)}$.
\end{proof}

Intuitively, the above lemma states that in any linear arrangement 
$\larr$ of $G'$ with maximum width at most $k' + n^2$, the kernels 
of the grid components of $G'$ cannot overlap one with the other.  
As a consequence,
$\larr$ induces an ordering  of the nodes of the source graph $G$ 
which is determined by the positions of the corresponding kernels.

\begin{lemma}\label{l4}
Let $\larr$ be a linear arrangement of $G'$ having maximum width at
most $k' + n^2$.  Then either $K^{(\larr)}_L < K^{(\larr)}_M <
K^{(\larr)}_R$ or $K^{(\larr)}_R < K^{(\larr)}_M < K^{(\larr)}_L$.
\end{lemma}

\begin{proof}
Assuming $K^{(\larr)}_L < K^{(\larr)}_R$, we show below that, under $\larr$, kernel $K^{(\larr)}_R$ cannot be placed in between kernels $K^{(\larr)}_L$ and $K^{(\larr)}_M$.  Essentially the same argument can be used to show that kernel $K^{(\larr)}_L$ cannot be placed in between kernels $K^{(\larr)}_M$ and $K^{(\larr)}_R$.

Assume that we have $K^{(\larr)}_L < K^{(\larr)}_R < K^{(\larr)}_M$.  
Since the number of nodes of $L$ which lie to the
left of $K^{(\larr)}_R$ is at least equal to $\lrklb$, and since
at most $\lminusmh(\lgw)$ nodes of $L$ can belong to its last
$\lminusmh$ rows, we have that at least $\frlklb$ nodes of the
first $\mgh$ rows of $L$ lie to the left of $K^{(\larr)}_R$. On
the other hand, since at least $\mklb$ nodes of $M$ belong to
$K^{(\larr)}_M$, we have that at least $\mklb$ nodes of $M$ lie
to the right of $K^{(\larr)}_R$.

Let us now consider the grid $X=\grid[\mgh,\lgw+\mgw]$ composed 
by the $(\mgh)$ upper rows of $L$ and all of the rows of $M$.
We apply Lemma~\ref{lem:partition} to $X$.  
If we define $V_1$ (respectively, $V_2$) as the
set of nodes of $X$ contained in $K^{(\larr)}_L$
(respectively, $K^{(\larr)}_M$), we have that both $\size{V_1}$ and
$\size{V_2}$ are greater than $(\ngh)^2$.  
Then we have that at least $\ngh$ edges internal 
to $X$ cross over all positions (gaps) of $K^{(\larr)}_R$. 
From the definition of kernels, there are at least $3n^4+1$
edges internal to $K^{(\larr)}_R$ crossing over each position of $K^{(\larr)}_R$.
Adding these together, we have at least $5n^4 + 2$ edges
at each position of $K^{(\larr)}_R$, 
which is greater than $k' + n^2$ (for $n$ sufficiently large). 

The case of $K^{(\larr)}_R < K^{(\larr)}_L$ can be dealt
with in a very similar way and the lemma thus follows.
\end{proof}

In the following, without loss of generality, we will always assume that $K^{(\larr)}_L < K^{(\larr)}_M < K^{(\larr)}_R$.  By applying essentially the same argument from the proof of Lemma~\ref{l4}, we can show that $K^{(\larr)}_i$ cannot lie between $K^{(\larr)}_L$ and $K^{(\larr)}_M$ or between $K^{(\larr)}_M$ and $K^{(\larr)}_R$, which implies the following result.

\begin{lemma}
\label{l3}
Let $\larr$ be a linear arrangement of $G'$ having maximum width at most $k' + n^2$.  
For any $i \in [n]$, either $K^{(\larr)}_i < K^{(\larr)}_L$ or 
$K^{(\larr)}_i > K^{(\larr)}_R$.
\end{lemma}

So far we have seen that kernels always appear in some total
order in the linear arrangements we are interested in, and with
the kernels of grids $L, M$ and $R$ all in a row.  We move on now
with a second property of the family of linear arrangements we
are looking at.  As already described above, this property states
that, if our goal is the one of optimizing the width at certain
gaps, then misplacing nodes that are not in a kernel does not
result in any gain.  We first provide two results about general
grids, and then come back to $G'$ and our linear arrangements.

\begin{lemma}\label{l5}
Let $X = \grid[\ggh,\ggw]$ and let $S$ be a set of nodes of $X$
such that $\size{S} \leq \ggw(\ggh-e-2)$ with $e\geq 0$ and there
exists $w$ with $1 \leq w \leq \ggw$ such that, for any $h$ with
$1 \leq h \leq \ggh$, $x^{h,w}\in S$ (in other words $S$ contains
an entire column of the grid). Then, $\delta(S,\overline{S})$
contains at least $e+2$ edges in distinct rows, where
$\overline{S}$ denotes the set of nodes of the grid which do not
belong to $S$.
\end{lemma}

\begin{proof}
For each $h$ with $1 \leq h \leq \ggh$, either there exists $w$
with $1 \leq w < \ggw$ such that $(x^{h,w} \in S \wedge x^{h,w+1}
\in \overline{S}) \vee (x^{h,w} \in \overline{S} \wedge x^{h,w+1}
\in S)$ (in this case, the row contributes at least by one
horizontal edge to $\delta(S,\overline{S})$), or, for any $w$
with $1 \leq w \leq \ggw$, $x^{h,w} \in S$. This latter case,
however, can happen at most $\left \lfloor \frac{\size{S}}{\ggw}
\right \rfloor$ times. Since $\size{S} \leq \ggw(\ggh-e-2)$, we have
that the first case happens at least $e+2$ times, thus proving
the lemma.
\end {proof}

\begin{lemma}\label{lem:firstrow}
Let $X = \grid[\ggh,\ggw]$ and let $S$ be a set of nodes of $X$
such that $\size{S} \leq \ggw(\ggh - \size{F}-2)$, where $F$ is a subset of
the set of nodes of the first row or of the last row which belong
to $S$. Then, $\delta(S,\overline{S})$ contains at least $\size{F}$
edges not included in the first row or in the last row.
\end{lemma}

\begin{proof}
For each $w$ with $1 \leq w \leq \ggw$ such that $x^{1,w} \in F
\vee x^{\ggh,w} \in F$, either there exists $h$ with $1 \leq h <
\ggh$ such that $(x^{h,w}\in S \wedge x^{h+1,w}\in \overline{S})
\vee (x^{h,w}\in \overline{S} \wedge x^{h+1,w}\in S)$ (in this
case, the column contributes at least by one vertical edge to
$\delta(S,\overline{S})$), or, for any $h$ with $1 \leq h \leq
\ggh$, $x^{h,w}\in S$ (that is, $S$ includes the entire $w$-th
column). If this latter case happens at least once, then we can
apply the previous lemma with $e=\size{F}$, thus obtaining that
$\delta(S,\overline{S})$ contains at least $\size{F}+2$ horizontal
edges on distinct rows, which implies that
$\delta(S,\overline{S})$ contains at least $\size{F}$ edges not
included in the first row or in the last row. Otherwise,
$\delta(S,\overline{S})$ contains at least $\size{F}$ vertical edges:
indeed, if $x^{1,w} \in S$ and $x^{\ggh,w} \not\in S$ or vice
versa, then at least one vertical edge of the $w$-th column is in
$\delta(S,\overline{S})$, otherwise at least two vertical edges
of this column are in $\delta(S,\overline{S})$ (since, in this
case, we have both to exit from $S$ and to enter again in $S$).
\end{proof}

We need to introduce some additional notation.  From now on, we denote by $l^*$ the first gap from left to right occurring between two vertices of $K^{(\larr)}_L$, and we denote by $r^*$ the first gap from left to right occurring between two vertices of $K^{(\larr)}_R$.
For any $i \in [n]$, we define the value $\alpha_i$ as follows. If $K^{(\larr)}_i > K^{(\larr)}_L$, $\alpha_i$ is the number of nodes of the first or of the last row of $G_i$ whose position under $\larr$ is smaller than $l^*$ and which are endpoints of an edge exiting $G_i$.  Otherwise, $\alpha_i$ is the number of nodes of the first or of the last row of $G_i$ whose position is greater than $l^*$ and which are endpoints of an edge exiting $G_i$. Similarly, we denote by $\alpha_M$ the number of nodes of the first row of $M$ whose position is smaller than $l^*$ and which are endpoints of an edge exiting $M$.

\begin{lemma}
\label{l7}
For any linear arrangement $\larr$ of $G'$ having maximum width at
most $k' + n^2$ and for any $i \in [n]$, there exist at least $\alpha_i$
distinct edges within $G_i$ which cross over $l^*$.
\end{lemma}

\begin{proof}
First observe that $\alpha_i\leq \nee+6$ because there 
are $\nee$ edges connecting $G_i$ to $M$ and 6 edges connecting 
$G_i$ to other grids $G_j$. 
We only study the case in which $K^{(\larr)}_i > K^{(\larr)}_L$, since 
the other case can be proved in the same way.

Let $V_i$ be the vertex set of $G_i$.  
Let ${\cal P}(G_i) = \{p : (\larr)^{-1}(p) \in V_i\}$ and let 
$S_i = \{p : (\larr)^{-1}(p) \in V_i \wedge p < l^*\}$ 
(clearly, $\size{S_i} \geq \alpha_i$).
Since $\size{{\cal P}(G_i)} = 12n^8+6n^4 \le 13n^8$ and $\size{K^{(\larr)}_i} \ge 3n^8$, and since $S_i$ is a subset of ${\cal P}(G_i) - K^{(\larr)}_i$, we are guaranteed that
$\size{S_i} \le 10n^8$.  Because $10n^8 \le (\ngw)(\ngh - \alpha_i - 2)$
(assuming $n \ge 4$)
we have the precondition $\size{S_i} \le (\ngw)(\ngh - \alpha_i - 2)$
that we need in order to apply Lemma~\ref{lem:firstrow}.
Lemma~\ref{lem:firstrow} implies that
there exist at least $\size{\alpha_i}$ distinct edges connecting $S_i$
to $\overline{S_i}$ (that is, edges within $G_i$): these edges
clearly cross over $l^*$.
\end{proof}

Similarly, we can prove the following result.

\begin{lemma}\label{l8}
For any linear arrangement $\larr$ of $G'$ having maximum width at most $k' + n^2$, 
there exist at least $\alpha_M$ distinct edges within $M$ which cross over $l^*$.
\end{lemma}

From now on, let $\kappa^{(\larr)}_l$ be the number of kernels $K^{(\larr)}_i$ such that $K^{(\larr)}_i < K^{(\larr)}_L$ and let $\kappa^{(\larr)}_r$ be the number of kernels $K^{(\larr)}_i$ such that $K^{(\larr)}_i > K^{(\larr)}_R$.  Let also $\tau^{(\larr)}$ denote the number of edges $(v_i,v_j)$ in $G$ such that $K^{(\larr)}_i < K^{(\larr)}_L$ and $K^{(\larr)}_j > K^{(\larr)}_L$. 

\begin{lemma}\label{l9}
Let $\larr$ be a linear arrangement of $G'$ having maximum width at most $k' + n^2$.
There exist at least $\kappa^{(\larr)}_l \cdot (\nee) + 2\tau^{(\larr)}$ distinct edges which cross over $l^*$, not including edges internal to $L$ or $R$.
\end{lemma}

\begin{proof}
We define $I^{(\larr)}_L$ as the set of integers $i \in [n]$ such
that $K^{(\larr)}_i < K^{(\larr)}_L$.  Thus we have 
$\size{I^{(\larr)}_L} = \kappa^{(\larr)}_l$.   
We also define $J^{(\larr)}_L$ as the set of integers $j \in [n]$ such
that $K^{(\larr)}_j > K^{(\larr)}_L$ and there exists $i \in I^{(\larr)}_L$ 
with $(v_i,v_j) \in E$, where $E$ is the set of edges of $G$.  

For each $i \in I^{(\larr)}_L$, 
let us consider the $\nee$ distinct edges connecting $G_i$ to $M$, 
along with each pair of edges connecting $G_i$ to each grid $G_j$ such that 
$K^{(\larr)}_j > K^{(\larr)}_L$ and $(v_i,v_j) \in E$.   
Let also ${\cal E}$ be the set of all these edges, for every $i \in I^{(\larr)}_L$.  
Thus we have $\size{{\cal E}}=\kappa_l \cdot (\nee )+2\tau^{(\larr)}$. 

Consider now an arbitrary edge $e \in {\cal E}$.  Let $x$ be one of the two endpoints of $e$, and assume that $x$ belongs to some grid $X$ among the $n+3$ grid components of $G'$.  We say that $x$ is \termdef{misplaced} if, under $\larr$, the kernel of $X$ is placed at some side with respect to $l^*$ and $x$ is placed at the opposite side.   From the definition of ${\cal E}$, it is easy to see that if none of the endpoints of $e$ are misplaced, or else if both of the endpoints of $e$ are misplaced, then $e$ must cross over $l^*$.  On the other hand, if exactly one of the endpoints of $e$ is misplaced, then $e$ does not cross over $l^*$.

Consider then the set of all the misplaced endpoints of some edge in ${\cal E}$.  
By construction of $G'$, these endpoints are distinct and belong to the first row or to 
the last row of some grid component of $G'$.  
Furthermore, the edges in ${\cal E}$ are all single rather than multiple edges, 
as already observed in Section~\ref{ssec:gprime}.
By definition of $\alpha_M$ and $\alpha_i$, $i \in [n]$, 
we have that $\alpha_M+\sum_{i \in (I^{(\larr)}_L \cup J^{(\larr)}_L)} \alpha_i$ 
is greater than or equal to the number of all the misplaced endpoints 
of some edge in ${\cal E}$, and from the above observations we have that 
the latter number is in turn greater than or equal to
the number of edges in ${\cal E}$ which do not cross over $l^*$.  
By Lemmas~\ref{l7} and~\ref{l8}, it follows that,
among the edges within $M$ and among the edges within the
components $G_i$, $i \in (I^{(\larr)}_L \cup J^{(\larr)}_L)$, 
there exist $\alpha_M+\sum_{i \in (I^{(\larr)}_L \cup J^{(\larr)}_L)} \alpha_i$ distinct
edges which cross over $l^*$.  This quantity plus the number of
edges in ${\cal E}$ which cross over $l^*$ gives us the desired result.
\end{proof}

We are now ready to show the inverse relation of the statement 
in Lemma~\ref{l:g2gprime}.  
In what follows we focus our attention on linear 
arrangements of $G'$ having maximum width at most $k'$.
The reason why all of the previous lemmas in this section 
have been stated for linear arrangements with 
maximum width at most $k' + n^2$ is because in 
Section~\ref{sec:fan-out} we need to refer to this extended class.

\begin{lemma}
\label{l:gprime2g}
If $\langle G', k' \rangle$ is a positive instance of \mrgbcw\
then $\langle G, k \rangle$ is a positive instance of \mcess.
\end{lemma}

\begin{proof}
Let $\larr$ be a linear arrangement of $G'$ having maximum width bounded by $k' = \kprime$, and consider quantity $\wg(G',\larr,l^*)$.  From Lemma~\ref{l9} there are at least $\kappa^{(\larr)}_l \cdot (\nee) + 2\tau^{(\larr)}$ distinct edges which cross over $l^*$, not including edges internal to the grids $L$ or $R$.  In addition, recall that there are at least $3n^4 + 1$ edges internal to $L$ that are crossing over $l^*$.  This is because of Corollary~\ref{cor:kernel} and because of the way we have defined kernels. 
If $\kappa^{(\larr)}_l > \frac{n}{2}$, the number of edges contributing to $\wg(G',\larr,l^*)$ would be at least $3n^4 + 1 + (\frac{n}{2}+1) \cdot \nee = 3n^4 + 2n^3 + \nee + 1 > \kprime = k'$, for sufficiently large values of $n$, where the inequality follows from the fact that $k$ is bounded by the number of edges in $G$, which is $\frac{3n}{2}$.  This is against our assumptions on $\larr$.  Thus we must conclude that $\kappa^{(\larr)}_l \leq \frac{n}{2}$.
Similarly, we can prove that $\kappa^{(\larr)}_r$ cannot be
greater than $\frac{n}{2}$. Hence, we have that $\kappa^{(\larr)}_l =
\kappa^{(\larr)}_r = \frac{n}{2}$. 

Using the above fact in Lemma~\ref{l9}, we have that
the number of edges external to $L$ and $R$ crossing over $l^*$ is at least
$2n^3+2\tau^{(\larr)}$.  Including the edges internal to $L$ gives at least 
$3n^4+ 2n^3 +2\tau^{(\larr)}+1$ edges crossing over $l^*$.
Since the width of $l^*$ is at most $\kprime$, 
it also follows that $\tau^{(\larr)} \le k + \frac{1}{2}$.  
This means that the number of edges $(v_i,v_j)$ in $G$ 
such that $K^{(\larr)}_i < K^{(\larr)}_L$ and $K^{(\larr)}_j > K^{(\larr)}_L$
is at most $k$.  Hence, by partitioning the nodes of $G$ according
to the position of their corresponding kernels under $\larr$, we have an equal
size subset partition whose cut is at most $k$.
\end{proof}

\subsection{Cutwidth and Internal Boundaries} 
\label{ssec:internal-boundaries}

We can now present the main results of Section~\ref{sec:internal_boundaries}. 

\begin{theorem}
\label{thm:mrgbcw}
The problem \mrgbcw\ is NP-complete.
\end{theorem}

\begin{proof} 
Let $G$ be a cubic graph with $n>1$ vertices, 
and let $k>0$ be some integer.  From Lemma~\ref{l:g2gprime} and from Lemma~\ref{l:gprime2g}, 
we have that $\langle G, k \rangle$ is a positive instance 
of \mcess\ if and only if $\langle G', k' \rangle$ is a positive instance of \mrgbcw.
This relation shows that an algorithm for \mrgbcw\ could be used 
to solve \mcess, and thus \mrgbcw\ is NP-hard.  

To conclude the proof, we observe that the problem 
\mrgbcw\ is in NP because a linear arrangement of 
a graph can be guessed in polynomial time and 
its maximum width can be computed in polynomial time as well.  
\end{proof}

We can now deal with a decision problem associated with the 
problem of finding a linear parsing strategy for a synchronous 
rule that minimizes the number of internal boundaries, defined in~(\ref{eq:ibd}).

\begin{theorem}
Let $s$ be a synchronous rule with $r$ nonterminals and with 
associated permutation $\pi_s$, 
and let $k$ be some positive integer.  
The problem of deciding whether 
\begin{align*}
 \min_{\sigma} \; \max_{i \in [r]} \; \ibd(\pi_s, \sigma, i) \leq k
\end{align*}
is NP-complete.
\end{theorem}

\begin{proof}
We have already observed that the relation 
$\cw(G_s) = \min_{\sigma} \max_{i \in [r]} \ibd(\pi_s, \sigma, k)$
directly follows from Lemma~\ref{lem:ib-cutwidth}.  The statement then 
follows from Theorem~\ref{thm:mrgbcw}.
\end{proof}

\section{Relating Permutation Multigraphs to SCFGs}
\label{sec:fan-out}

As discussed in Section~\ref{sec:prel}, our main goal is
finding efficient ways of parsing synchronous context-free rules.
In this section, we use our results on permutation multigraphs
to prove NP-hardness for optimizing both space complexity and
time complexity of linear parsing strategies for SCFG rules.
We begin by examining space complexity, and then generalize the
argument to prove our result on time complexity.

\subsection{Space Complexity}
\label{ssec:space}

Optimizing the space complexity of a parsing strategy is equivalent to 
minimizing the maximum that the fan-out function achieves across the steps
of the parsing strategy~(\ref{eq:minfanout}).  
According to the definition of the fan-out function
(\ref{eq:fo}), fan-out consists of
the two terms $\ibd$ and $\ebd$, accounting for the internal and
the external boundaries, respectively, realized at a given step
by a linear parsing strategy.  Let $s$ be a synchronous rule with
$G_s$ the associated permutation multigraph.  We have already
seen in Lemma~\ref{lem:ib-cutwidth} a relation between the $\ibd$
term and the width function for $G_s$.  In order to make precise
the equivalence between the fan-out problem and the cutwidth
problem for $G_s$, we must now account for the $\ebd$ term.

Let $b_R$ and $e_R$ be the first and the last vertices of the red path in $G_s$, and let $b_G$ and $e_G$ be the first and the last vertices of the green path.  We collectively refer to these vertices as the \termdef{endpoints} of $G_s$, and we define $V_e = \{ b_R, e_R, b_G, e_G \}$.   Let $\sigma_s$ be some linear parsing strategy for rule $s$, and let $\nu_s$ be the corresponding linear arrangement for $G_s$, as defined in Section~\ref{ssec:pmg}.  Informally, we observe that under $\sigma_s$ the number of external boundaries at a given step $i$ is the number of vertices from the set $V_e$ that have been seen to the left of the current position $i$ under $\nu_s$ (including $i$ itself).   Using definition~(\ref{eq:fo}) and Lemma~\ref{lem:ib-cutwidth}, this suggests that we can represent the fan-out at $i$ as the width of an \textit{augmented} permutation multigraph containing special edges from the vertices in $V_e$, where the special edges always extend past the right end of any linear arrangement.   We introduce below some mathematical definitions that formalize this idea.  

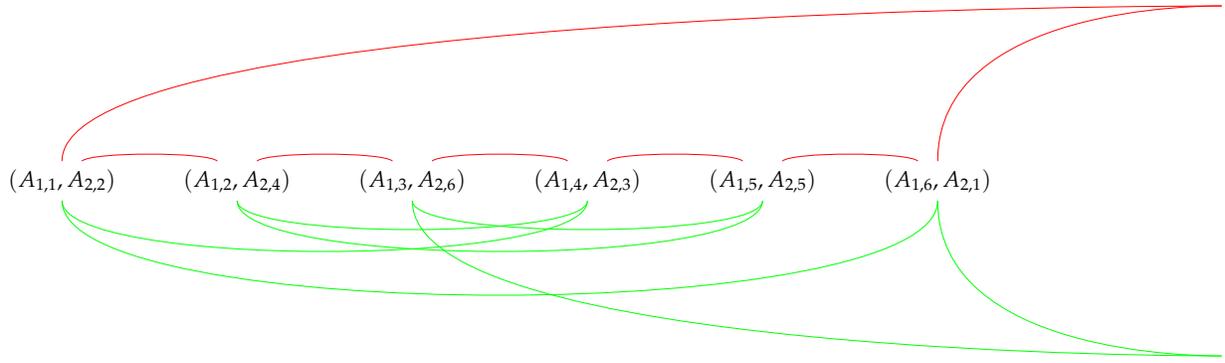
\begin{figure}
\begin{center}
\resizebox{\linewidth}{!}{
\begin{tikzpicture}[scale=1]
\draw (0,0) node(1) {$(A_{1,1},A_{2,2})$};
\draw (3,0) node(2) {$(A_{1,2},A_{2,4})$};
\draw (6,0) node(3) {$(A_{1,3},A_{2,6})$};
\draw (9,0) node(4) {$(A_{1,4},A_{2,3})$};
\draw (12,0) node(5) {$(A_{1,5},A_{2,5})$};
\draw (15,0) node(6) {$(A_{1,6},A_{2,1})$};
\draw[color=red] (1) .. controls +(0.5,0.5) and +(-0.5,0.5) .. (2);
\draw[color=red] (2) .. controls +(0.5,0.5) and +(-0.5,0.5) .. (3);
\draw[color=red] (3) .. controls +(0.5,0.5) and +(-0.5,0.5) .. (4);
\draw[color=red] (4) .. controls +(0.5,0.5) and +(-0.5,0.5) .. (5);
\draw[color=red] (5) .. controls +(0.5,0.5) and +(-0.5,0.5) .. (6);
\draw[color=green] (6) .. controls +(0,-2.5) and +(0,-2.5) .. (1);
\draw[color=green] (1) .. controls +(0,-1.5) and +(0,-1.5) .. (4);
\draw[color=green] (4) .. controls +(0,-1) and +(0,-1) .. (2);
\draw[color=green] (2) .. controls +(0,-1.5) and +(0,-1.5) .. (5);
\draw[color=green] (5) .. controls +(0,-1) and +(0,-1) .. (3);
\draw (20,3) node(eng) {};
\draw (20,-3) node(chi) {};
\draw[color=red] (1) .. controls +(0,3) and +(-0.5,0) .. (eng);
\draw[color=red] (6) .. controls +(0,3) and +(-0.5,0) .. (eng);
\draw[color=green] (3) .. controls +(0,-3) and +(-0.5,0) .. (chi);
\draw[color=green] (6) .. controls +(0,-3) and +(-0.5,0) .. (chi);
\end{tikzpicture}}
\end{center}
\caption{Extended 
edges for the permutation multigraph corresponding to the SCFG rule 
$s$ of eqn.~(\ref{eq:rule-s}).
}\label{fig:extended}
\end{figure}

Assume $G_s$ has $n$ nodes and set of edges $E$.
We define the \termdef{extended width} at position $i \in [n]-1$ to be
\begin{eqnarray} 
\ewg(G_s, \nu_s, i) & = & \wg(G_s, \nu_s, i) + \sum_{v \in V_e} I(\nu_s(v) \le i) \; .
\label{eq:ewidth}
\end{eqnarray}
The contribution of the endpoints to the extended width can
be visualized as counting, at each position in the linear 
arrangement, an additional set of edges running from the
endpoint of the red and green paths all the way to the right
end of the linear arrangement, as shown in 
Figure~\ref{fig:extended}.  We will refer to these additional
edges as \termdef{extended edges}.   
To simplify the notation below, we also let $\ewg(G_s,\nu_s,n) = 4$.  

Let $\pi_s$ be the permutation associated with synchronous rule $s$.  Observe that the first term in~(\ref{eq:ewidth}) corresponds to the number of internal boundaries $\ibd(\pi_s, \sigma_s, i)$, by Lemma~\ref{lem:ib-cutwidth}, and the second term counts the number of external boundaries $\ebd(\pi_s, \sigma_s, i)$.   
We can then write, for each $i \in [n]$
\begin{eqnarray}
\fo(\pi_s, \sigma_s, i) & = & \frac{1}{2} \ewg(G_s, \nu_s, i) \; ,
\label{eq:fo_ewd}
\end{eqnarray}
which will be used below to assess the complexity of the fan-out problem.  Finally, we define the \termdef{extended cutwidth} of $G_s$ as 
\begin{eqnarray}
\ecw(G_s) & = & \min_{\nu} \; \max_{i \in [n]} \; \ewg(G_s,\nu,i) \; .
\label{eq:ecwidth}
\end{eqnarray}
From~(\ref{eq:fo_ewd}) and~(\ref{eq:ecwidth}) we see that the extended cutwidth of $G_s$ is related to the optimal computational complexity that we can achieve when parsing synchronous rule $s$ with the techniques described in Section~\ref{ssec:membership}.  With such motivation, we investigate below a decision problem related to the computation of the extended cutwidth of a permutation multigraph.  

From now on, we assume that $\langle G, k \rangle$ is an instance of \mcess, where $G$ is a cubic graph with $n$ vertices.  We also assume that $G'$ and $k'$ are constructed from $G$ and $k$ as in Section~\ref{ssec:gprime}.

\begin{lemma}
If $\langle G, k \rangle$ is a positive instance of \mcess, then $\ecw(G') \leq k' + 2$.
\label{l:bisec2ecutwidth}
\end{lemma}

\begin{proof}
Under the assumption that $\langle G, k \rangle$ is a positive instance of \mcess, consider the linear arrangement $\nu$ used in the proof of Lemma~\ref{l:g2gprime} to show that $\langle G', k' \rangle$ is a positive instance of \mrgbcw.  We already know that the maximum (regular) width of $\nu$ is at most $k'$.  With the exception of the first column of grid $L$ and the last column of grid $R$, the extended width under $\nu$ at positions within $L$ and $R$ is two greater than the (regular) width, because the vertices $b_R$ and $b_G$ are both to the left, while $e_R$ and $e_G$ are to the right.  For positions in the first column of $L$, the extended width is one greater than the width, because only $b_G$ is to the left, as depicted in Figure~\ref{fig:greenpath}.

The critical point is the last column of $R$.  We observe that the edges connecting vertices in the grid components $G_i$ and $M$ contribute to the extended width at positions within $R$ always in the same amount.  We thus focus our analysis on the only edges that are internal to $R$.  Recall that $e_R$ is the topmost vertex in the last column of $R$.  We let $i$ be the position of $e_R$ under $\nu$.  At position $i-1$ the contribution to the extended width of the edges internal to $R$ consists of $\rgh$ red edges and one green edge; see again Figure~\ref{fig:greenpath}.  At the next position $i$, one red edge and one green edge internal to $R$ are lost.  However, these two edges are replaced by one new red edge from $R$ and one new extended edge impinging on vertex $e_R$.  Thus the extended width at positions $i-1$ and $i$ must be the same.  For all of the next positions corresponding to vertices in the last column of $R$,  the contribution to the extended width of the edges internal to $R$ always decreases.  

From the above observations, we conclude that the extended width at positions within grid components $L$ and $R$ is bounded by $k' + 2 = \ekprime$. 
As already observed in the proof of Lemma~\ref{l:g2gprime}, the width at all of the remaining positions for $\nu$ is lower by $n^4 + \order{n^3}$, and this must also be the case for the extended width, since this quantity exceeds the (regular) width by at most four.  The existence of linear arrangement $\nu$ thus implies $\ecw(G') \le k' + 2$.
\end{proof}

Let $\larr$ be a linear arrangement of $G'$ having maximum width
at most $k' + n^2$. Then $\larr$ satisfies the hypotheses of all
of the lemmas in Section~\ref{ssec:gprime2g} constraining the linear
arrangement of the kernels of the grid components of $G'$.
However, contrary to the case of the regular cutwidth, the
extended cutwidth is not invariant to a reversal of a linear
arrangement.  This is so because the extended edges always end up
at a position to the right of the right end of any linear
arrangement.  For this reason we can no longer assume that in
$\larr$ we have $L<R$.  Let then $X_L$ and $X_R$ be the leftmost
and the rightmost, respectively, of $L$ and $R$ under $\larr$.
Let also $e^*$ be the rightmost of $l^*$ and $r^*$.

We already know from the first part of the proof of Lemma~\ref{l:gprime2g} that the number of kernels to the left of $K^{(\larr)}_{X_L}$ is $\frac{n}{2}$, and this is also the number of kernels to the right of $K^{(\larr)}_{X_R}$.  Using this fact, the following result can easily be shown using the same argument presented in the proof of Lemma~\ref{l9}. As in Section~\ref{ssec:gprime2g}, let $\tau^{(\larr)}$ denote the number of edges $(v_i,v_j)$ in $G$ such that $K^{(\larr)}_i < K^{(\larr)}_{X_L}$ and $K^{(\larr)}_j > K^{(\larr)}_{X_R}$.  

\begin{lemma}
Let $\larr$ be a linear arrangement of $G'$ having maximum width at most $k' + n^2$.
There are at least $\frac{n}{2} \cdot (\nee) + 2\tau^{(\larr)}$ distinct edges which cross over $e^*$, not including edges internal to the grids $L$ or $R$.
\label{l:e-star}
\end{lemma}

The proof of the next lemma uses arguments very similar to those already exploited in the proof of Lemma~\ref{l7} and in the proof of Lemma~\ref{l:gprime2g}. 

\begin{lemma}
If $\ecw(G') \leq k' + 2$, then $\langle G, k \rangle$ is a positive instance of \mcess.
\label{l:ecutwidth2bisec}
\end{lemma}

\begin{proof}
Let $\larr$ be a linear arrangement of $G'$ having maximum extended width bounded by $k' + 2 = \ekprime$. Since the maximum (regular) width of $G'$ is at most its maximum extended width, $\larr$ satisfies the hypotheses of all of the lemmas in Section~\ref{ssec:gprime2g} constraining the arrangement of the kernels and the misplaced nodes from the grid components of $G'$.  

From Lemma~\ref{l:e-star} there are at least $2n^3 + 2\tau^{(\larr)}$ distinct edges which cross over $e^*$, not including edges internal to the grids $L$ or $R$.  In addition, there are at least $3n^4 + 1$ edges internal to $X_R$ that are crossing over $e^*$.  
Finally, consider the two endpoints of either the red or green path appearing in the first and in the last lines of $X_L$, and let $\alpha_e$ be the number of such endpoints that have been misplaced to the right of $e^*$ under $\larr$.  Note that we have $0 \leq \alpha_e \leq 2$. 

Let ${\cal P}(X_L) = \{ p : (\larr)^{-1}(p) \in X_L \}$ 
and let $S_{X_L} = \{ p : (\larr)^{-1}(p) \in X_L \land p > e^* \}$.
Because $\size{{\cal P}(X_L)} = 36n^8 + 12n^4$ and $\size{K_{X_L}} \ge 17n^8$,
we know that $\size{S_{X_L}} \le \size{{\cal P}(X_L) - K_{L'}} \le 20n^8$.
Therefore $\size{S_{X_L}} \le (12n^4)(3n^4 +1 - \alpha_e -2)$ and
we can apply Lemma~\ref{lem:firstrow} with $X = X_L$ and $S = S_{X_L}$.
From Lemma~\ref{lem:firstrow}, there are at least $\alpha_e$ edges internal to
$X_L$ that cross over $e^*$.  Along with the $2-\alpha_e$ extended edges departing 
from $X_L$, this accounts for two additional edges that cross over $e^*$.

Combining all of the above contributions, the total number of edges crossing over $e^*$ is at least $3n^4 + 2n^3 + 2\tau^{(\larr)} + 3$.  Since the extended width of $e^*$ is at most $\ekprime$, it also follows that $\tau^{(\larr)} \le k + \frac{1}{2}$.  This means that the number of edges $(v_i,v_j)$ in $G$ such that $K^{(\larr)}_i < K^{(\larr)}_L$ and $K^{(\larr)}_j > K^{(\larr)}_L$ is at most $k$.  Hence, by partitioning the nodes of $G$ according to the position of their corresponding kernels under $\larr$, we have an equal
size subset partition whose cut is at most $k$.
\end{proof}

\begin{theorem}
\label{thm:ecw_np}
Let $G$ be a permutation multigraph and let $k$ be a positive integer. 
The problem of deciding whether $\ecw(G) \leq k$ is NP-complete.
\end{theorem}

\begin{proof}
By Lemma~\ref{l:bisec2ecutwidth} and Lemma~\ref{l:ecutwidth2bisec}, an algorithm that decides whether $\ecw(G) \leq k$ can be used to solve \mcess.  
Thus the problem in the statement of the theorem is NP-hard. 
The problem is in NP because the maximum extended width of a linear arrangement can 
be verified in polynomial time.
\end{proof}

\begin{theorem}
\label{thm:fo}
Let $s$ be a synchronous rule with $r$ nonterminals and with 
associated permutation $\pi_s$, 
and let $k$ be a positive integer.  The problem of deciding whether 
\begin{eqnarray*}
\min_{\sigma} \; \max_{i \in [n]} \; \fo(\pi_s, \sigma, i) & \leq & k
\end{eqnarray*}
is NP-complete, where $\sigma$ ranges over all linear parsing
strategies for $s$.
\end{theorem}

\begin{proof}
The NP-hardness part directly follows from Theorem \ref{thm:ecw_np}, along with~(\ref{eq:fo_ewd}) and~(\ref{eq:ecwidth}).  The problem is also in NP, because the maximum value of the fan-out for a guessed linear parsing strategy can be computed in polynomial time.
\end{proof}

We have already discussed how fan-out is directly related to the space complexity of the implementation of a linear parsing strategy.  From Theorem~\ref{thm:fo} we then conclude that optimization of the space complexity of linear parsing for SCFGs is NP-hard.

\subsection{Time Complexity}
\label{ssec:time}

We now turn to the optimization of the time complexity of linear parsing for SCFGs.  
It turns out that at each step of a linear parsing strategy, the time complexity is related to a variant of the notion of width, called modified width, computed for the corresponding position of a permutation graph.   With this motivation,  we investigate below the  modified width and some extensions of this notion, and we derive our main result in a way which parallels what we have already done in Section~\ref{ssec:space}. 

We start with some additional notation.  Let $G = (V,E)$ be an undirected (multi)graph such that $\size{V} = n >1$, and let $\nu$ be some linear arrangement of $G$.  For any $i \in [n]$, the \termdef{modified width} of $G$ at $i$ with respect to $\nu$, written $\mwg(G,\nu,i)$, is defined as $\size{\{(u,v) \in E : \nu(u) < i < \nu(v)\}}$.   Informally, $\mwg(G,\nu,i)$ is the number of distinct edges crossing over the vertex at position $i$-th in the linear arrangement $\nu$.  Again in case of multigraphs the size of the previous set should be computed taking into account multiple edge occurrences.  The following result is a corollary to Lemma~\ref{lem:partition}.  

\begin{corollary}
\label{cor:grid-mcw}
For any grid $X = \grid[\ggh,\ggw]$ with $\ggw \geq 2\ggh+1$, for any linear arrangement $\nu$ of $X$, and for any $i$ with $\ggh^2 < i \leq \ggh\ggw-\ggh^2$, $\mwg(X,\nu,i) \geq \ggh - 2$.
\end{corollary}

\begin{proof}
Modified width at the vertex in position $i$ can be related to the (regular) width of 
the gaps before and after position $i$, and the degree $\Delta(\nu^{-1}(i))$ of the vertex at position $i$
\begin{align*}
\mwg(X,\nu,i) = \frac{1}{2}(\wg(X,\nu,i-1) + \wg(X,\nu,i) - \Delta(\nu^{-1}(i))) \; .
\end{align*}
For each $i$ with $\ggh^2 < i \leq \ggh\ggw-\ggh^2$, we can use Corollary~\ref{cor:kernel} and write  
\begin{align*}
\mwg(X,\nu,i) & \ge H - \frac{1}{2}\Delta(\nu^{-1}(i))) \; .
\end{align*}
Because vertices in a grid have degree at most four, we have $\mwg(X,\nu,i) \ge H - 2$.
\end{proof}

Let $s$ be a synchronous rule and let $G_s$ be the associated permutation multigraph.  Let also $\nu_s$ be a linear arrangement for $G_s$.  We define the \termdef{extended modified width} at $i \in [n]$ as
\begin{eqnarray} 
\emwg(G_s, \nu_s, i) & = & \wg(G_s, \nu_s, i) + \sum_{v \in V_e} I(\nu_s(v) \le i) \; .
\label{eq:emwidth}
\end{eqnarray}
Again, the contribution of the endpoints of $G_s$ to the extended modified width can be visualized as counting, at each position in the linear arrangement, an additional set of edges running from the endpoint of the red and green paths all the way to the right end of the linear arrangement,
as shown in Figure~\ref{fig:extended}.   The \termdef{extended modified cutwidth} of $G_s$ is  
\begin{eqnarray*}
\emcw(G_s) & = & \min_{\nu} \; \max_{i \in [n]} \; \emwg(G_s,\nu,i) \; ,
\end{eqnarray*}
where $\nu$ ranges over all possible linear arrangements for $G_s$. 

From now on, we assume that $\langle G, k \rangle$ is an instance of \mcess, where $G$ is a cubic graph with $n$ vertices.  We also assume that $G'$ and $k'$ are constructed from $G$ and $k$ as in Section~\ref{ssec:gprime}.  

\begin{lemma}
\label{l:emcutwidth2bisec}
If $\langle G, k \rangle$ is a positive instance of \mcess, then $\emcw(G') \leq k' $.
\end{lemma}

\begin{proof}
Consider the linear arrangement used in Lemma~\ref{l:bisec2ecutwidth} to show that $\ecw(G') \leq k' + 2$.  The critical points in that linear arrangement are all within the $L$ and $R$ components.  At all positions in these components, each vertex has two edges
extending to its right and two edges extending to its left (one red and one green in each case).  Thus, at these positions, the extended modified cutwidth is less than the extended cutwidth by two.
\end{proof}

\begin{lemma}
\label{l:bisec2emcutwidth}
If $\emcw(G') \leq k'$, then $\langle G, k \rangle$ is a positive instance of \mcess.
\end{lemma}

\begin{proof}
Assume a linear arrangement $\larr$ for $G'$ having maximum extended modified width bounded by $k'$.  The maximum (non-extended) modified width of $G'$ must be smaller than or equal to $k'$,  and the  maximum (non-extended, non-modified) width of $G'$ must be smaller than or equal to $k'+4$, because the maximum degree of vertices in $G'$ is four.  Since the maximum width of $G'$ under $\larr$ is bounded by $k'+n^2$, we apply Lemma~\ref{l:e-star} and conclude that there are at least $2n^3 + 2\tau^{(\larr)}$ distinct edges which cross over the vertex to the left of the gap $e^*$, not including edges internal to the grids $L$ or $R$.  

By Corollary~\ref{cor:grid-mcw}, the number of edges internal to $X_R$ which cross over the vertex to the left of the gap $e^*$ is at least $3n^4-1$.   Finally, consider the two endpoints of either the red or green path appearing in the first and in the last lines of $X_L$, and let $\alpha_e$ be the number of such endpoints that have been misplaced to the right of $e^*$ under $\larr$, with $0 \leq \alpha_e \leq 2$.   We apply the same argument as in the proof of Lemma~\ref{l:ecutwidth2bisec}, and conclude that there are at least $\alpha_e$ edges internal to $X_L$ that cross over the vertex to the left of the gap $e^*$.  Along with the $2-\alpha_e$ extended edges departing from $X_L$, this accounts for two additional edges that cross over the vertex to the left of the gap $e^*$.

Combining all of the above contributions, the total number of edges crossing over the vertex to the left of the gap $e^*$ is at least $3n^4 + 2n^3 + 2\tau^{(\larr)} + 1$.  Since the extended modified width of the vertex to the left of the gap $e^*$ is at most $k' = \kprime$, it also follows that $\tau^{(\larr)} \le k + \frac{1}{2}$.  This means that the number of edges $(v_i,v_j)$ in $G$ such that $K^{(\larr)}_i < K^{(\larr)}_{X_L}$ and $K^{(\larr)}_j > K^{(\larr)}_{X_R}$ is at most $k$.  Hence, by partitioning the nodes of $G$ according to the position of their corresponding kernels under $\larr$, we have an equal
size subset partition whose cut is at most $k$.
\end{proof}

\begin{theorem}
\label{thm:emcw_np}
Let $G$ be a permutation multigraph and let $k$ be a positive integer. 
The problem of deciding whether $\emcw(G) \leq k$ is NP-complete.
\end{theorem}

\begin{proof}
By Lemma~\ref{l:bisec2emcutwidth} and Lemma~\ref{l:emcutwidth2bisec}, an algorithm that decides whether $\emcw(G) \leq k$ can be used to solve \mcess.  Thus deciding whether $\emcw(G) \leq k$ is NP-hard.  The problem is in NP because the maximum extended modified width of a linear arrangement can be verified in polynomial time.
\end{proof}

We now relate the notion of extended modified width for a permutation multigraph and the time complexity of a parsing algorithm using a linear strategy.  Let $s$ be a synchronous rule with $r$ nonterminals having the form in~(\ref{e:rule_with_pi}), and let $G_s$ be the  permutation multigraph associated with $s$. Let also $\sigma_s$ be some linear parsing strategy defined for $s$, and let $\nu_s$ be the linear arrangement associated with $\sigma_s$, as defined in Section~\ref{ssec:pmg}.  Recall from Section~\ref{ssec:membership} that the family of parsing algorithms we investigate in this article use parsing states to represent the boundaries (internal and external) that delimit the substrings of the input that have been parsed at some step, following our strategy $\sigma_s$.   

For some $i$ with $i \in [r-1]$, let us consider some parsing state with state type $(s, \sigma_s, i)$.  In the next parsing step $i+1$, we move to a new state with state type $(s, \sigma_s, i+1)$ by adding to our partial analyses the $(i+1)$-th pair of nonterminals from the right-hand side of $s$, defined according to $\sigma_s$.  As already observed in Section~\ref{ssec:membership}, this operation involves some updates to the sequence of boundaries of our old state of type $(s, \sigma_s, i)$.  More precisely, the new state is constructed from the old state by removing a number $\delta^{(-)}_{i+1}$ of boundaries, and by adding a number $\delta^{(+)}_{i+1}$ of new boundaries.  From the definition of $G_s$, we know that $\delta^{(-)}_{i+1}$ is the number of backward edges at vertex $i+1$, and $\delta^{(+)}_{i+1}$ is the number of forward edges at vertex $i+1$, where backward, forward and vertex $i+1$ are all defined relative to the linear arrangement $\nu_s$ and include the extended edges.  We also have $\delta^{(-)}_{i+1} + \delta^{(+)}_{i+1} = \Delta(\nu_s^{-1}(i+1))$, where $\Delta(\nu^{-1}(i+1))$ is the degree of the vertex at position $i+1$. 

The total number of boundaries $t_{i+1}$ involved in the parsing step $i+1$ is then the number of boundaries for state type $(s, \sigma_s, i)$, which includes the $\delta^{(-)}_{i+1}$ boundaries that need to be removed at such step, plus the new boundaries $\delta^{(+)}_{i+1}$.  We already know that $\ewg(G_s, \nu_s, i)$ is the number of boundaries for $(s, \sigma_s, i)$.  We can then write
\begin{eqnarray*}
t_{i+1} & = & \ewg(G_s, \nu_s, i) + \delta^{(+)}_{i+1} \\
 & = & \ewg(G_s, \nu_s, i) - \delta^{(-)}_{i+1} 
       + \delta^{(-)}_{i+1} + \delta^{(+)}_{i+1} \\
 & = & \emwg(G_s, \nu_s, i+1) + \Delta(\nu_s^{-1}(i+1)) \\
 & = & \emwg(G_s, \nu_s, i+1) + 4 \; .
\end{eqnarray*}
That is, the total number of boundaries involved in a parsing
step is the number of boundaries that are not affected by the
step, which correspond to edges passing over a vertex in the 
linear arrangement, $\emwg(G_s, \nu_s, i+1)$, plus the number 
of boundaries opened or closed by adding the new nonterminal,
which is the vertex's degree $\Delta(\nu_s^{-1}(i+1))$.

Let $w_1$ and $w_2$ be the input strings in our synchronous parsing problem, and let $n$ be the maximum between the lengths of $w_1$ and $w_2$.  Observe that there may be $\order{n^{t_{i+1}}}$ different instantiations of parsing step $i+1$ in our algorithm.  In order to optimize the time complexity of our algorithm, relative to synchronous rule $s$, we then need to choose a linear arrangement that achieves maximum extended modified width of $\emcw(G_s)$.  From Theorem~\ref{thm:emcw_np}, we then conclude that optimization of the time complexity of linear parsing for SCFGs is NP-hard.

\section{Discussion}
\label{disc}
In this section, we discuss the implications of our results
for machine translation.
Synchronous 
parsing is the problem of finding a suitable representation of
the derivations
of a string pair consisting of one string from each
language in the translation.  In the context of 
statistical machine translation, this problem arises
when we wish to analyze string pairs consisting 
of known parallel text in, say, English and Chinese,
for the purposes of counting how often each SCFG rule
is used and estimating its probability.  Thus,
synchronous parsing corresponds to the training phase
of a statistical machine translation system.
Our results show that it is NP-hard to find 
the linear synchronous parsing strategy with the 
lowest space complexity or the lowest time complexity.
This indicates that learning complex translation models
from parallel text is a fundamentally hard problem.

A separate, but closely related, problem arises when
translating new Chinese sentences into English,  
a problem known as decoding.  A simple decoding
algorithm consists of parsing the Chinese string with the
Chinese side of the SCFG, and simply reading the English
translation off of the English side of each rule used.
This can be accomplished in time $\order{n^3}$ using the CYK
parsing algorithm for (monolingual) context-free grammars,
since we use only one side of the SCFG.

More generally, we may wish to compute not only the single
highest-scoring translation, but a compact representation
of all English translations of the Chinese string.
Just as the chart constructed during monolingual parsing can be
viewed as a non-recursive CFG generating all analyses 
of a string, we can parse the Chinese string with the
Chinese side of the SCFG, retain the resulting chart, 
and use the English sides of the rules as a non-recursive
CFG generating all possible English translations of the
Chinese string.  However, in general, the rules cannot
be binarized in this construction, since the Chinese and English
side of each rule are intertwined.  This means that the
resulting non-recursive CFG has size greater than $\order{n^3}$,
with the exponent depending on the maximum length of the
SCFG rules.  One way to reduce the exponent is to factor
each SCFG rule into a sequence of steps, as in our linear
SCFG parsing strategies.

Machine translation systems do, in fact, require 
such a representation of possible translations,
rather than simply taking the single best translation 
according to the SCFG.
This is because the score from the SCFG is combined 
with a score from an English $N$-gram language model in order
to bias the output English string toward hypotheses
with a high prior probability, that is, strings that 
look like valid English sentences.  In order to incorporate
scores from an English $N$-gram language model
of order $m$, we extend the dynamic
programming algorithm to include in the state of each 
hypothesis the first and last $m-1$ words of each 
contiguous segment of the English sentence.  Thus, the
number of contiguous segments in English, which is the 
fan-out of the parsing strategy on the English side, 
again enters into the complexity \cite{Huang:2009}. 
Given a parsing strategy with fan-out $f_c$ on the 
Chinese side and fan-out $f_e$ on the English side,
the space complexity of the dynamic programming table
is $\order{n^{f_c}V^{2f_e(m-1)}}$, where
$n$ is the length of the Chinese input string,
$V$ is the size of the 
English vocabulary, and $m$ is the order of $N$-gram 
language model.  Under the standard assumption that
each Chinese word has a constant number of possible
English translations, this is equivalent to 
$\order{n^{f_c+2f_e(m-1)}}$.  Thus, our NP-hardness result
for the space complexity of linear strategies for synchronous parsing
also applies to the space complexity of linear strategies
for decoding with an integrated language model.

Similarly, the time complexity of language-model-integrated decoding is
related to the time complexity of synchronous parsing
through the order $m$ of the $N$-gram language model.
In synchronous parsing, the time complexity of a step combining of state
of type $(s,\sigma,k)$ and a nonterminal $(A_{1,k+1},A_{2,\pi^{-1}(k+1)})$ to produce a 
state of type $(s,\sigma,k+1)$ is $\order{n^{a+b+c}}$, where 
$a$ is the number of boundaries in states of type $(s,\sigma,k)$,
$b$ the number in nonterminal 
$(A_{1,k+1},A_{2,\pi^{-1}(k+1)})$, and $c$ the number in
type $(s,\sigma,k+1)$.  If we rewrite $a$ as $a_c+a_e$, where $a_c$
is the number of boundaries in Chinese and $a_e$ is
the number of boundaries in English, then the exponent 
for the complexity of synchronous parsing is:
\[ a_e + b_e + c_e + a_c + b_e + c_e \]
and the exponent for language-model-integrated decoding is:
\[ (m-1)(a_e + b_e + c_e) + a_c + b_e + c_e \]
Note that these two expressions coincide in the case where $m=2$.
Since we proved that optimizing the time complexity
of linear synchronous parsing strategies is NP-complete,
our result also applies to the more general problem 
of optimizing time complexity of language-model-integrated
decoding for language models of general order $m$.

\paragraph*{Open Problems}
This article presents the first NP-hardness 
result regarding parsing strategies for SCFGs.
However, there is a more general version of the problem 
whose complexity is still open.
In this article, we have restricted ourselves
to linear parsing strategies, that is, strategies
that add one nonterminal at a time to the subset of
right hand side nonterminals recognized so far.  In general,
parsing strategies may group right hand side nonterminals
hierarchically into a tree.  For some permutations,
hierarchical parsing strategies for SCFG rules can be more efficient 
than linear parsing strategies \cite{Huang:2009}.
Whether the time complexity of hierarchical
parsing strategies is NP-hard is not known
even for the more general class of LCFRS\@.
An efficient algorithm for 
minimizing the time complexity of hierarchical strategies
for LCFRS would imply an improved approximation
algorithm for the well-studied graph-theoretic
problem of treewidth \cite{gildea-cl11}.
Minimizing fan-out of hierarchical strategies,
on the other hand, is trivial, for both LCFRS and SCFG\@.
This is because the 
strategy of combining all right hand side nonterminals 
in one step (that is, forming a hierarchy of 
only one level) is optimal in terms of fan-out,
despite its high time complexity.

\bibliographystyle{alpha}
\bibliography{segmentalNumber}

\newcommand{\etalchar}[1]{$^{#1}$}
\begin{thebibliography}{GRKSW09}

\bibitem[AU69]{Aho:69a}
Alfred~V. Aho and Jeffery~D. Ullman.
\newblock Syntax directed translations and the pushdown assembler.
\newblock {\em J. Computer and System Sciences}, 3:37--56, 1969.

\bibitem[AU72]{AhoUll72}
Albert~V. Aho and Jeffery~D. Ullman.
\newblock {\em The Theory of Parsing, Translation, and Compiling}, volume~1.
\newblock Prentice-Hall, Englewood Cliffs, NJ, 1972.

\bibitem[BCLS87]{Bui:1987}
T.~Bui, S.~Chaudhuri, T.~Leighton, and M.~Sipser.
\newblock Graph bisection algorithms with good average case behavior.
\newblock {\em Combinatorica}, 7:171--191, 1987.

\bibitem[BDDM93]{Brown:93}
Peter~F. Brown, Stephen~A. {Della Pietra}, Vincent~J. {Della Pietra}, and
  Robert~L. Mercer.
\newblock The mathematics of statistical machine translation: {P}arameter
  estimation.
\newblock {\em Computational Linguistics}, 19(2):263--311, 1993.

\bibitem[CGM{\etalchar{+}}11]{head-driven-acl11}
Pierluigi Crescenzi, Daniel Gildea, Andrea Marino, Gianluca Rossi, and Giorgio
  Satta.
\newblock Optimal head-driven parsing complexity for linear context-free
  rewriting systems.
\newblock In {\em Proceedings of the 49th Annual Meeting of the Association for
  Computational Linguistics (ACL-11)}, 2011.

\bibitem[Chi04]{Chiang:04}
David Chiang.
\newblock {\em Evaluation of Grammar Formalisms for Applications to Natural
  Language Processing and Biological Sequence Analysis}.
\newblock PhD thesis, University of Pennsylvania, 2004.

\bibitem[Chi07]{ChiangCL}
David Chiang.
\newblock Hierarchical phrase-based translation.
\newblock {\em Computational Linguistics}, 33(2):201--228, 2007.

\bibitem[Ear70]{Earley:70}
Jay Earley.
\newblock An efficient context-free parsing algorithm.
\newblock {\em Communications of the ACM}, 6(8):451--455, 1970.

\bibitem[GHKM04]{galley-naacl04}
Michel Galley, Mark Hopkins, Kevin Knight, and Daniel Marcu.
\newblock What's in a translation rule?
\newblock In {\em Proceedings of the 2004 Meeting of the North American chapter
  of the Association for Computational Linguistics (NAACL-04)}, pages 273--280,
  Boston, 2004.

\bibitem[Gil11]{gildea-cl11}
Daniel Gildea.
\newblock Grammar factorization by tree decomposition.
\newblock {\em Computational Linguistics}, 37(1):231--248, 2011.

\bibitem[GRKS10]{gomez-naacl10}
Carlos G\'{o}mez-Rodr\'{i}guez, Marco Kuhlmann, and Giorgio Satta.
\newblock Efficient parsing of well-nested linear context-free rewriting
  systems.
\newblock In {\em Human Language Technologies: The 2010 Annual Conference of
  the North American Chapter of the Association for Computational Linguistics},
  pages 276--284, Los Angeles, California, 2010.

\bibitem[GRKSW09]{gomez-naacl09}
Carlos G{\'o}mez-Rodr{\'i}guez, Marco Kuhlmann, Giorgio Satta, and David Weir.
\newblock Optimal reduction of rule length in {Linear Context-Free Rewriting
  Systems}.
\newblock In {\em Proceedings of the 2009 Meeting of the North American chapter
  of the Association for Computational Linguistics (NAACL-09)}, pages 539--547,
  2009.

\bibitem[G{\v{S}}07]{gildea-stefankovic:2007:main}
Daniel Gildea and Daniel {\v{S}}tefankovi\v{c}.
\newblock Worst-case synchronous grammar rules.
\newblock In {\em Proceedings of the 2007 Meeting of the North American chapter
  of the Association for Computational Linguistics (NAACL-07)}, pages 147--154,
  2007.

\bibitem[HU79]{Hopcroft:79}
John~E. Hopcroft and Jeffrey~D. Ullman.
\newblock {\em Introduction to Automata Theory, Languages, and Computation}.
\newblock Addison-Wesley, Reading, MA, 1979.

\bibitem[HZGK09]{Huang:2009}
Liang Huang, Hao Zhang, Daniel Gildea, and Kevin Knight.
\newblock Binarization of synchronous context-free grammars.
\newblock {\em Computational Linguistics}, 35(4):559--595, 2009.

\bibitem[KOM03]{Koehn-naacl03}
Philipp Koehn, Franz~Josef Och, and Daniel Marcu.
\newblock Statistical phrase-based translation.
\newblock In {\em Proceedings of the 2003 Meeting of the North American chapter
  of the Association for Computational Linguistics (NAACL-03)}, pages 48--54,
  Edmonton, Alberta, 2003.

\bibitem[LS68]{LewisStearns68}
P.~M. Lewis, II and R.~E. Stearns.
\newblock Syntax-directed transduction.
\newblock {\em Journal of the ACM}, 15(3):465--488, 1968.

\bibitem[MPS85]{Makedon:1985p2188}
F.S. Makedon, C.H. Papadimitriou, and I.H. Sudborough.
\newblock Topological bandwidth.
\newblock {\em SIAM Journal on Algebraic and Discrete Methods}, 6:418--444,
  1985.

\bibitem[RSV95]{DBLP:conf/wg/RolimSV95}
Jos{\'e} D.~P. Rolim, Ondrej S{\'y}kora, and Imrich Vrto.
\newblock Optimal cutwidths and bisection widths of 2- and 3-dimensional
  meshes.
\newblock In {\em 21st International Workshop on Graph-Theoretic Concepts in
  Computer Science}, pages 252--264, 1995.

\bibitem[SMFK91]{Seki91}
H.~Seki, T.~Matsumura, M.~Fujii, and T.~Kasami.
\newblock On multiple context-free grammars.
\newblock {\em Theoretical Computer Science}, 88:191--229, 1991.

\bibitem[SP05]{Satta:05}
Giorgio Satta and Enoch Peserico.
\newblock Some computational complexity results for synchronous context-free
  grammars.
\newblock In {\em Proceedings of Human Language Technology Conference and
  Conference on Empirical Methods in Natural Language Processing (HLT/EMNLP)},
  pages 803--810, Vancouver, Canada, 2005.

\bibitem[SS90]{Shieber:90CO}
Stuart Shieber and Yves Schabes.
\newblock Synchronous tree-adjoining grammars.
\newblock In {\em Proceedings of the 13th International Conference on
  Computational Linguistics (COLING-90)}, volume III, pages 253--258, Helsinki,
  1990.

\bibitem[SS94]{SH94}
Yves Schabes and Stuart~M. Shieber.
\newblock An alternative conception of tree-adjoining derivation.
\newblock {\em Computational Linguistics}, 20:91--124, 1994.

\bibitem[SS10]{sagot-satta:2010:ACL}
Beno\^{i}t Sagot and Giorgio Satta.
\newblock Optimal rank reduction for linear context-free rewriting systems with
  fan-out two.
\newblock In {\em Proceedings of the 48th Annual Meeting of the Association for
  Computational Linguistics}, pages 525--533, Uppsala, Sweden, 2010.

\bibitem[VSWJ87]{LCFRS}
K.~Vijay-Shankar, D.~L. Weir, and A.~K. Joshi.
\newblock Characterizing structural descriptions produced by various
  grammatical formalisms.
\newblock In {\em Proceedings of the 25th Annual Conference of the Association
  for Computational Linguistics (ACL-87)}, pages 104--111, Stanford, CA, 1987.

\bibitem[You67]{Younger:67}
Daniel~H. Younger.
\newblock Recognition and parsing of context-free languages in time $n^3$.
\newblock {\em Information and Control}, 10:189--208, 1967.

\end{thebibliography}

\end{document}